\documentclass[11pt]{article}

\usepackage{etoc}

\usepackage{authblk}

\title{Fundamental Quality Bound on Optical Quantum Communication}

\author[1]{Tobias Rippchen\thanks{Contact information: tobias.rippchen@rwth-aachen.de}}
\author[2]{Ludovico Lami}
\author[3]{Gerardo Adesso}
\author[1]{Mario Berta}

\affil[1]{Institute for Quantum Information, RWTH Aachen University, Germany}
\affil[2]{Scuola Normale Superiore, Piazza dei Cavalieri 7, 56126 Pisa, Italy}
\affil[3]{School of Mathematical Sciences and Centre for the Mathematical and Theoretical Physics of Quantum Non-Equilibrium Systems, University of Nottingham, University Park, Nottingham, NG7 2RD, United Kingdom}

\date{}

\usepackage{geometry}
\usepackage{fancyhdr}
\usepackage[T1]{fontenc}
\usepackage[utf8]{inputenc} 

\usepackage{lmodern}
\usepackage[sc]{mathpazo}

\usepackage{graphicx}  
\usepackage[font=scriptsize]{caption} 

\usepackage{xcolor}

\usepackage{hyperref}
\hypersetup{
    colorlinks=true,
    linkcolor=teal,
    citecolor=purple,  
    urlcolor=blue,
    pdftitle={Gaussian Sanov Theorem},
    pdfpagemode=FullScreen,
}

\usepackage[
    backend=biber,
    style=numeric-comp,
    sorting=none
]{biblatex}
\bibliography{main.bib}

\usepackage{tcolorbox}

\usepackage{comment}

\usepackage{amsmath}
\usepackage{amssymb}
\usepackage{mathtools}

\DeclareMathOperator{\Id}{Id}

\DeclareMathOperator{\arcoth}{arcoth}

\newcommand{\abs}[1]{\left| #1 \right|}
\newcommand{\norm}[1]{\left\lVert#1\right\rVert}
\newcommand{\tr}[1]{\mathrm{Tr}\left[#1\right]}
\newcommand{\trA}[1]{\mathrm{Tr}_{A}\left[#1\right]}

\newcommand{\com}[2]{\left[#1,#2\right]}
\newcommand{\anticom}[2]{\left\{#1,#2\right\}}

\newcommand{\ket}[1]{|{#1}\rangle}
\newcommand{\bra}[1]{\langle{#1}|}
\newcommand{\braket}[2]{\langle{#1}|{#2}\rangle}
\newcommand{\ketbra}[2]{|{#1}\rangle\langle{#2}|}

\usepackage{amsthm}
\usepackage{framed}

\theoremstyle{plain}

\newtheorem{prototheorem}{Theorem}[section]
\newenvironment{theorem}
   {\begin{leftbar}\begin{prototheorem}}
   {\end{prototheorem}\end{leftbar}}

\newtheorem{protocorollary}{Corollary}[section]

\newtheorem{protoproposition}[prototheorem]{Proposition}
\newenvironment{proposition}
   {\begin{leftbar}\begin{protoproposition}}
   {\end{protoproposition}\end{leftbar}}

\newtheorem{protolemma}[prototheorem]{Lemma}
\newenvironment{lemma}
   {\begin{leftbar}\begin{protolemma}}
   {\end{protolemma}\end{leftbar}}

\newtheorem{prototheorem2}{Theorem}
\newtheorem{protoproposition2}[prototheorem2]{Main Finding}
\newenvironment{proposition2}
   {\begin{leftbar}\begin{protoproposition2}}
   {\end{protoproposition2}\end{leftbar}}

\theoremstyle{definition}

\newtheorem{definition}[prototheorem]{Definition}
\newtheorem{remark}[prototheorem]{Remark}

\begin{document}
	
\maketitle
	
\begin{abstract}
    Sending quantum information reliably over long distances is a central challenge in quantum technology in general, and in quantum optics in particular, since most quantum communication relies on optical fibres or free-space links. Here, we address this problem by shifting the focus from the quantity of information sent to the quality of the transmission, i.e.\ the rate of decay of the transmission error with respect to the number of channel uses. For the general class of teleportation-simulable channels, which includes all channels arising in quantum optical communication, we prove that the single-letter reverse relative entropy of entanglement of the Choi state upper bounds the (zero-rate) error exponent of two-way assisted quantum communication\,---\,paralleling the celebrated capacity bound of [Pirandola {\it et al.}, Nat.\ Comm.\ (2017)] in terms of the regularised relative entropy of entanglement. Remarkably, for Gaussian channels our bound can be computed efficiently through a convex program with simple constraints involving only finite-dimensional covariance matrices. As a prototypical application, we derive closed-form analytical expressions of our upper bound as well as random-coding-based lower bounds for several one-mode Gaussian channels. Extending recent work [Lami {\it et al.}, arXiv:2408.07067 (2024)] to infinite-dimensional systems, we further endow the reverse relative entropy of entanglement with an exact operational interpretation in entanglement testing, and show that it characterises the rate of entanglement distillation under non-entangling operations. These findings offer a new perspective on entanglement as a resource and sharpen the theoretical benchmarks for future quantum optical networks.
\end{abstract}

\begin{center}
    {\large \bfseries Introduction \par}
\end{center}

\paragraph{The Challenge with Capacities.} Quantum communication lies at the heart of emerging quantum technologies, from secure cryptography to distributed quantum computing \cite{Pirandola_2020}. It is thus of crucial importance to understand the fundamental limits to which the transmission of quantum information through noisy channels, such as optical fibres or free-space links, is subjected. Traditionally, this problem has been analysed in terms of \emph{capacities}, which quantify the maximum amount of information that can be transmitted reliably. Of particular importance is the \emph{quantum capacity}, which measures the ability of a channel to faithfully transmit qubits, enabling core primitives such as the distribution of entanglement.

Despite its importance, the quantum capacity of generic channels remains poorly understood. The main challenge is that optimal communication involves encoding the information across multiple uses of the channel. Consequently, the capacity is defined in an information-theoretic sense as an asymptotic rate, i.e.\ it quantifies how the performance scales with the number of channel uses. Notably, this is already true in the purely classical case of transmitting classical information via a classical channel. However, as Shannon showed in his foundational work~\cite{Shannon_1948}, the capacity in this case is mathematically given by a \emph{single-letter} formula that involves optimising an entropic quantity\,---\,the mutual information\,---\,over a single use of the channel. In stark contrast, the (one-way assisted) quantum capacity \(\mathcal{Q}\) of a quantum channel \(\mathcal{N}\) is given by a \emph{regularised} formula, i.e.\ by an optimisation of an entropic quantity\,---\,the coherent information \( I_\mathrm{coh}\)\,---\,in the limit of arbitrarily many uses of the channel~\cite{Lloyd_1996, Shor_2002, Devetak_2005_2}:
\begin{equation}\label{Eq:Quantum_Capacity}
	\mathcal{Q}( \mathcal{N} ) := \lim_{n \to \infty} \frac{1}{n} \mathcal{Q}_n( \mathcal{N} )  \quad \text{with} \quad \mathcal{Q}_n( \mathcal{N} ) := \sup_{ \rho_n } I_\mathrm{coh} \left( \mathcal{N}^{\otimes n}, \rho_n \right) \, ,
\end{equation}
where \(\mathcal{Q}_n( \mathcal{N} )\) is the coherent information maximised for \(n\) uses of the channel \(\mathcal{N}\) over a joint input state \(\rho_n\). 

These regularised formulas are notoriously difficult to compute, both analytically and numerically, because the underlying entropic measures are typically not additive~\cite{Shor_2004, Hastings_2009}. Moreover, even deciding whether a channel has non-zero quantum capacity generally requires an unbounded number of channel uses~\cite{Cubitt_2014}. Because of these computational hurdles, we cannot calculate the quantum capacity even for very simple and physically relevant channels, e.g.\ thermal attenuators, which constitute the paradigmatic model for optical quantum communication.

Recently, a shift in perspective has gained traction in quantum information theory that may be summarised as \emph{quality over quantity}~\cite{Lami_2024_2, Ji_2024, Nuradha_2025, Girardi_2025, Girardi_2025_2, Hayashi_2025, Hayashi_2025_2}. Rather than asking "how much" information can be transmitted optimally, one may instead ask about the quality of the transmission. In fact, in many practical scenarios\,---\,such as distributing entanglement for cryptography\,---\,it is often more desirable to obtain less entanglement of higher quality than a lot of it of poorer quality (see Figure~\ref{Fig:Distribution of Entanglement}). The central question of this work may thus be phrased as: 

\begin{tcolorbox}[colback=yellow!15,colframe=yellow!15, boxrule=0pt, left=2pt, right=2pt, top=2pt, bottom=2pt]
\noindent
Can we give quantitative bounds on the quality of quantum communication and entanglement distribution over (very noisy) optical channels for which the quantity perspective has not led to strong capacity characterisations? 
\end{tcolorbox}

\begin{figure}[htbp]
 	\centering
    \includegraphics[width=\linewidth]{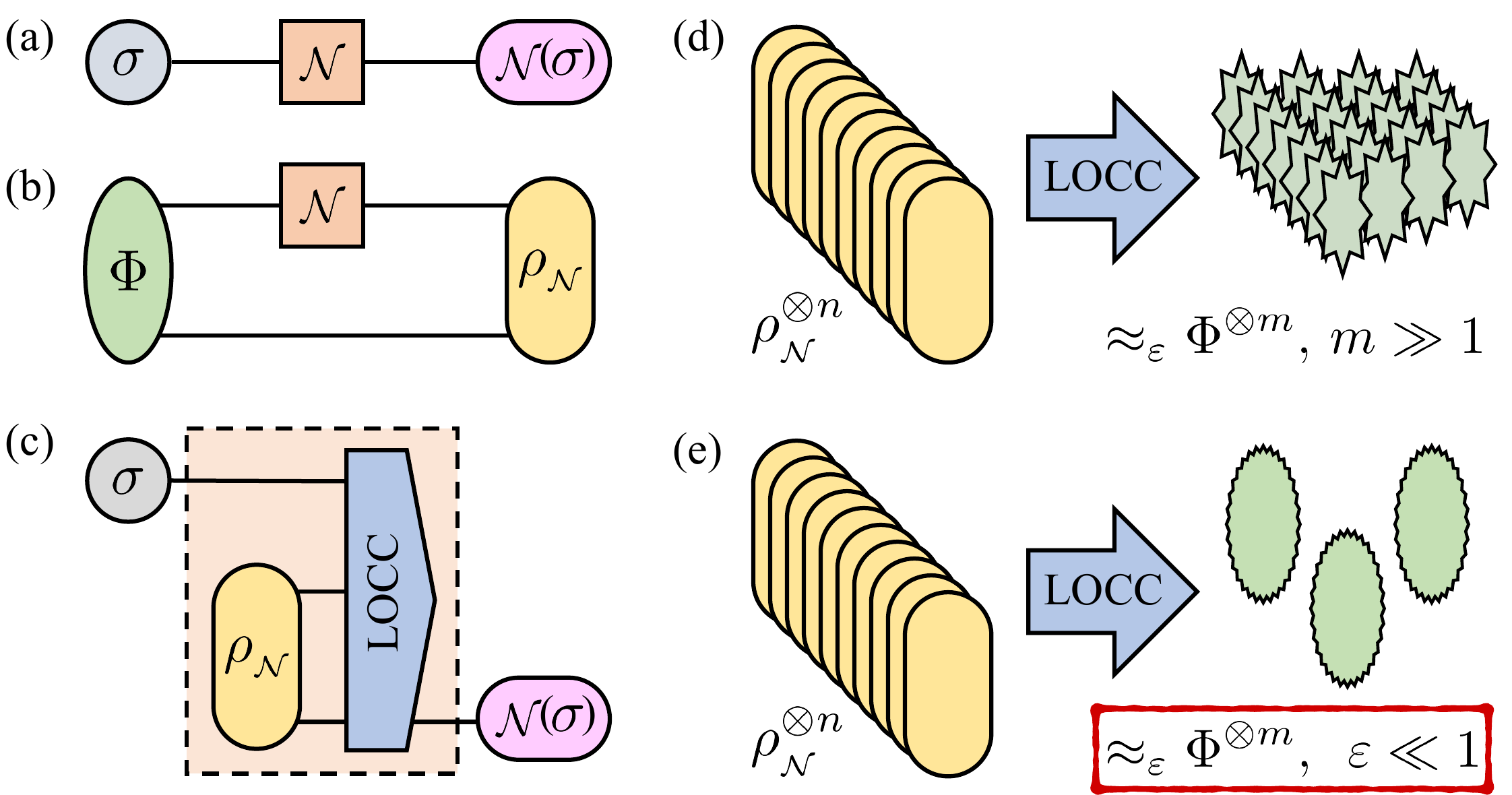}
    \caption{(a) A quantum channel \(\mathcal{N}\) mapping an input state \(\sigma\) to an output state \(\mathcal{N}(\sigma)\). (b) The Choi state \(\rho_{\mathcal{N}} = (\mathcal{N} \otimes I )(\Phi)\) of the channel \(\mathcal{N}\), where \(\Phi\) is a maximally entangled state. (c) The channel \(\mathcal{N}\) is said to be teleportation-simulable if its action on any input state \(\sigma\) can be reproduced by means of an LOCC protocol supplemented by the Choi state \(\rho_{\mathcal{N}}\) distributed between sender and receiver. The quantum communication scenario studied in this work can be recast in terms of how well one can distill \(m\) copies of approximately maximally entangled states \(\approx_{\varepsilon}\Phi^{\otimes m}\) from \(n \gg 1\) copies \(\rho^{\otimes n}_{\mathcal{N}}\) of the Choi state \(\rho_{\mathcal{N}}\) of a channel \(\mathcal{N}\), for which two paradigms can be considered:  (d) aiming for \emph{quantity}, i.e.\ maximising the number of output copies \(m\), or (e) aiming for \emph{quality}, i.e.\ maximising the exponent at which the error \(\varepsilon\) decays as \(n \rightarrow \infty\). This work investigates the latter scenario.}
    \label{Fig:Distribution of Entanglement}
\end{figure}

\paragraph{A Link to Entanglement Theory.} The answer to this question has important consequences beyond communication theory itself. In particular, the theory of entanglement measures is deeply linked with quantum communication. To make this connection explicit, it is instructive to consider the framework of (unlimited) two-way assistance, referred to in technical terms as \emph{adaptive local operations and classical communication} (adaptive LOCC). This represents the most general class of protocols that sender and receiver may employ to aid quantum communication without relying on pre-shared entanglement. For example, when distributing entanglement this is a natural assistance model, since classical communication is considered inexpensive in this context. The associated \emph{two-way assisted capacity} thus provides an important fundamental benchmark for quantum communication.

The two-way assisted capacity is substantially easier to analyse for the class of \emph{teleportation-simulable channels}\,---\,those whose action can be simulated by running the quantum teleportation protocol using their own Choi state as a resource (see Figure~\ref{Fig:Distribution of Entanglement} for a visual definition). Importantly, this encompasses all physically relevant channels in optical communication. For these channels, the two-way assisted quantum capacity coincides with the \emph{distillable entanglement} of the Choi state~\cite{Bennett_1996_3, Pirandola_2017}. The latter, in turn, characterises the asymptotic yield of \emph{entanglement distillation}, a fundamental primitive in entanglement theory, in which the goal is to convert noisy entangled states into pure, maximally entangled ones~\cite{Bennett_1996_1, Bennett_1996_2, Bennett_1996_3}. This establishes a rigorous and quantitative link between the theory of entanglement measures and quantum communication. 

Naturally, the same non-additivity problems that plague quantum communication theory also arise in the theory of entanglement measures~\cite{Vollbrecht_2001, Hayden_2001, Devetak_2005}. More precisely, the distillable entanglement is related to the \emph{regularised relative entropy of entanglement}~\cite{Vedral_1997, Vedral_1998,Brandao_2010_1}, defined via
\begin{equation}
	D^\infty(\rho\|\mathrm{SEP}) := \lim_{n \to \infty} \frac{1}{n} \inf_{\sigma_n \in \mathrm{SEP}}D \left( \rho^{\otimes n} \middle\| \sigma_n \right) \, ,
\end{equation}
i.e.\ as the asymptotic relative-entropy distance between the bipartite product state \(\rho^{\otimes n}\) and the set of states that is not entangled across the bipartite cut of the global system, denoted \(\mathrm{SEP}\). Crucially, the set \(\mathrm{SEP}\) lacks the product structure of \(\rho^{\otimes n}\), leading to a fundamental non-additivty that necessitates regularisation. The difficulty of evaluating this regularised formula remains one of the main bottlenecks to theoretical progress. 

However, shifting the perspective from \emph{quantity to quality} has already led to fundamentally new insights in entanglement theory. In a prior work~\cite{Lami_2024_2}, some of us discovered that the \emph{reverse relative entropy of entanglement}~\cite{Vedral_1997, Eisert_2003}, defined as
\begin{equation}
	D( \mathrm{SEP} \| \rho ) := \inf_{\sigma \in \mathrm{SEP} } D( \sigma \| \rho) \, ,
\end{equation}
 plays a central role in the theory of entanglement measures when adopting the quality framework. Specifically, this measure characterises the (zero-rate) \emph{error exponent} of entanglement distillation under \emph{non-entangling} operations. The latter class of operations is a useful relaxation of the traditional LOCC framework, whose study has already been fruitful in the past~\cite{Brandao_2010_1, Brandao_2010_2, Brandao_2015, Hayashi_2024, Lami_2024_1}. Strikingly, the reverse relative entropy of entanglement enjoys a highly desirable property\,---\,it is additive on product inputs, which makes its regularisation superfluous. Thus, precisely for the previously poorly understood mixed quantum states, the shift of perspective from the \emph{quantity} to the \emph{quality} of the distilled entanglement enables a single-letter characterisation of the asymptotic performance. This gives us a strong intuition that this entanglement measure should also play a similar role in the theory of quantum communication.

%%%%%%%%%%%%%%%%%%%%%%%%%%%%%%%%%%%%%%%%%%%%%%%%%%%%%%%%%%%%%%%%%%%%%%%%%%%

\begin{center}
    {\large \bfseries Results \par}
\end{center}

\paragraph{The Paradigm Shift.}
In following the above intuition, one faces three immediate problems: (1)~the previous work~\cite{Lami_2024_2} dealt with states and not channels, (2)~the proof is clearly limited to finite-dimensional systems and, most importantly, (3)~even if you could reprove the generalised Sanov theorem that underpins their main findings, the resulting expression would in general \emph{not} be efficiently computable, as it would involve an optimisation over infinite-dimensional separable states. Here, we solve all three of these problems, finding \emph{fundamental bounds on the quality} of entanglement distribution and quantum communication over the most optically relevant quantum channels. 

Our finding complements the celebrated work by Pirandola et al.~\cite{Pirandola_2017}, which established fundamental bounds on the \emph{quantity} of entanglement that can be distributed over the same class of optical channels (see also Bennett et al.~\cite{Bennett_1996_3}). The duality between our result and theirs is striking: while~\cite{Pirandola_2017} showed that, for teleportation-simulable channels, the standard relative entropy of entanglement yields an upper bound on the two-way assisted quantum capacity, we prove that the \emph{reverse} relative entropy of entanglement upper bounds the two-way assisted (zero-rate) \emph{error exponent}. Moreover, unlike in the case of~\cite{Pirandola_2017}, we can prove that the quantifiers we find have an exact \emph{operational interpretation} in entanglement distillation, when relaxing the LOCC framework to non-entangling operations. The experience with entanglement measures suggests that our single-letter bound on the error exponent may already be among the tightest possible, in the sense that going substantially beyond it will unavoidably run into the fundamental problem of NPT bound entanglement, which is currently beyond our understanding~\cite{Horodecki-open-problems} (see Remark \ref{Rem:Tight_Bound} in the Supplementary Material for more details). We substantiate this intuition by deriving explicit random-coding-based achievability bounds on the error exponent for the most important one-mode Gaussian channels, indicating that our upper bound is not far from optimal.

%%%%%%%%%%%%%%%%%%%%%%%%%%%%%%%%%%%%%%%%%%%%%%%%%%%%%%%%%%%%%%%%%%%%%%%%%%%
\paragraph{The Zero-Rate Error Exponent of Quantum Communication.} 

Consider two parties, Alice and Bob, connected by a noisy bosonic channel \(\mathcal{N}\), whose aim is to distribute entanglement between them. In order to achieve this task, they may employ the help of adaptive LOCC which we denote in the following by \(\mathrm{LOCC}_\leftrightarrow\) (see Figure~\ref{Fig:Distribution of Entanglement} for the setup of this problem). Observe that in this setup once they can reliably establish entanglement, they can also transmit arbitrary quantum information using the quantum teleportation protocol. 

More formally, the two parties will apply some protocol \(\Lambda_n\), that uses the the channel \(n\) times, and produces at the output a state \(\rho(\Lambda_n)\) that approximates \(m\) copies of the maximally entangled qubit state \(\Phi = \ketbra{\Phi}{\Phi}\) with \(\ket{\Phi} = \frac{1}{\sqrt{2}} \left(\ket{00} + \ket{11} \right)\). Denoting the approximation error as \(\varepsilon_n\), we can write this as \(\rho(\Lambda_n ) \approx_{\varepsilon_n} \Phi^{\otimes m} \), using a suitable measure of distance\,---\,either the trace distance or, equivalently, the infidelity. Although we allow for a finite error, we do require that it satisfies \(\lim_{n \to \infty} \varepsilon_n = 0\); that is, as the channel can be used more often, the quality of the distributed entanglement should increase, becoming perfect in the asymptotic limit. 

Traditionally, the analysis of this task was focussed on the \emph{quantity} of distributed entanglement. In this setting, the rate \(\frac{m}{n}\) of the protocol is considered the figure-of-merit and the goal is to find the largest asymptotic rate \( \lim_{n \to \infty} \frac{m}{n} \) that can be achieved among all feasible protocols. This then leads to the definition of the two-way assisted quantum capacity. In this work, we shift the perspective to the \emph{quality} of the obtained entanglement. That is, we require that \(\varepsilon_n \sim 2^{-c n}\) and characterise the optimal achievable error exponent \(c\). The \emph{error exponent of two-way assisted quantum communication} is then defined as the largest such exponent that can be achieved in the asymptotic limit\footnote{To be technically precise, this defines the \emph{zero-rate} error exponent in Shannon theory parlance.}
\begin{align}
    Q_{\leftrightarrow, \mathrm{err} }(\mathcal{N}) := \lim_{m \to \infty} \sup \bigg\{ \lim_{n \to \infty} - \frac{1}{n} \log \varepsilon_n : \rho(\Lambda_n ) \approx_{\varepsilon_n} \Phi^{\otimes m}, \; \Lambda_n \in \mathrm{LOCC}_{\leftrightarrow}(\mathcal{N}^{\times n}) \bigg\} \, ,
\end{align}
where we optimise over sequences of adaptive LOCC protocols that use the quantum channel \(n\) times. Observe that this definition no longer places any importance on the precise number of maximally entangled copies (provided that it can be made as large as desired), but only on the exponentially decreasing error. We then find the following \emph{single-letter} upper bound.

\begin{proposition2}\label{Prop:Main_1}
	For the class of teleportation-simulable channels, the reverse relative entropy of entanglement provides an upper bound on their (zero-rate) error exponent of two-way assisted quantum communication. Specifically, we have for a channel \(\mathcal{N}\) that acts on \(m\) bosonic modes that
\begin{equation}
	Q_{\leftrightarrow, \mathrm{err} }(\mathcal{N}) \leq \liminf_{r \to \infty} D( \mathrm{SEP} \| \rho_\mathcal{N}(r) ) =: \mathrm{UB}(\mathcal{N})\, ,
\end{equation}
where \(\mathrm{SEP}\) denotes the set of separable states and the quasi-Choi state, \( \rho_\mathcal{N}(r) \coloneqq (\mathcal{N} \otimes \mathcal{I})\big(\Phi(r)^{\otimes m}\big) \), is obtained by sending one half of the state \(\Phi(r)^{\otimes m}\), where $\Phi(r)$ is a two-mode squeezed vacuum state,
through the channel.
\end{proposition2}

In our proof, we start from the observation that for teleportation-simulable channels any adaptive LOCC protocol that employs $n$ uses of the channel can be reduced to an equivalent protocol that acts on $n$ copies of the associated Choi state (cf.~\cite{Bennett_1996_3} and~\cite{Pirandola_2017}). Our bound then follows by exploiting the specific properties of the reverse relative entropy of entanglement, most notably its additivity, after a suitable reformulation of the error exponent. The technical details can be found in Sec.\ \ref{Sec:Upper_Bound} and \ref{Sec:Upper_Bound_2} of the Supplementary Material.

%%%%%%%%%%%%%%%%%%%%%%%%%%%%%%%%%%%%%%%%%%%%%%%%%%%%%%%%%%%%%%%%%%%%%%%%%%%
\paragraph{Gaussian Reverse Relative Entropy of Entanglement.} 

Although the reverse relative entropy of entanglement does not require regularisation, that does not mean it is efficiently computable in general; in fact, the optimisation over separable states is NP-hard (see e.g.~\cite{Ioannou_2006} and references therein), and in infinite-dimensional systems one cannot even resort to SDP hierarchies (via extendibility) to soften this issue. Fortunately, we show that these problems evaporate in the case of $m$-mode bosonic Gaussian states. Specifically, we prove that the computation of the reverse relative entropy of entanglement reduces to a convex program over $2m$-dimensional quantum covariance matrices with two simple positive semi-definite constraints.  
\begin{proposition2}\label{Prop:Main_2}
The reverse relative entropy of entanglement of the Gaussian state \( \rho_G \) on the bipartite system \(A \otimes B\) with quantum covariance matrix \( \boldsymbol{V}_\rho\) can be computed via the \textbf{convex} program
\begin{align}
     \displaystyle \min_{\boldsymbol{V}_\sigma, \boldsymbol{\gamma}_{A} > 0} \quad & \frac{\tr{ \boldsymbol{V}_\sigma ( \boldsymbol{G}[\boldsymbol{V}_\rho] - \boldsymbol{G}[\boldsymbol{V}_\sigma] ) }}{2 \ln(2)} + \log \sqrt{  \frac{ \det \left( \boldsymbol{V}_\rho + i \boldsymbol{\Omega}_{AB} \right)}{\det \left( \boldsymbol{V}_\sigma + i \boldsymbol{\Omega}_{AB} \right)} } \\
     \mathrm{s.t.} \quad & \boldsymbol{V}_\sigma \geq \boldsymbol{\gamma}_A \oplus i \boldsymbol{\Omega}_B \nonumber \quad \mathrm{and} \quad \boldsymbol{\gamma}_A \geq i \boldsymbol{\Omega}_A \nonumber 
\end{align}
where \(\boldsymbol{\Omega}\) is the canonical symplectic form of the system, and \(\boldsymbol{G}[\boldsymbol{V}]\) is the Gibbs matrix associated with \(\boldsymbol{V}\).
\end{proposition2}

Our result has fundamental consequences in the theory of entanglement measures in continuous-variable systems. While Gaussian states are fundamental and ubiquitous, prior to our work it was unknown whether any single Gaussian entanglement measure could be at the same time (A)~operationally meaningful and (B)~efficiently computable. The negativity~\cite{Vidal_2002} and Gaussian entanglement of formation~\cite{Wolf_2004} are both efficiently computable, but they do not enjoy a strong operational interpretation in themselves. It is the actual (regularised) entanglement of formation that has an operational interpretation, but we do not know in general if it coincides with its Gaussian version (except in special cases~\cite{Giedke_2003, Akbari_2015}); and, crucially, we do not know how to regularise it. Here, we find:
	\begin{tcolorbox}[colback=yellow!15,colframe=yellow!15, boxrule=0pt, left=2pt, right=2pt, top=2pt, bottom=2pt]
		\noindent
		The reverse relative entropy of entanglement is -- to the best of the authors' knowledge -- the first \emph{operational} entanglement measure that is also \emph{efficiently computable} for Gaussian states.
	\end{tcolorbox}

The key to this finding is that, for Gaussian inputs, the regular and Gaussian reverse relative entropy of entanglement coincide. The main technical hurdle to prove this is to show that the optimisation can be restricted to separable states with finite second moments. We achieve this with an argument based on the variational formula for the measured relative entropy from~\cite{Ferrari_2023}. Once this has been established, a Gaussification argument allows us to restrict the optimisation further to Gaussian separable states. This is especially noteworthy, as it is unknown whether a similar restriction holds for the standard relative entropy of entanglement. The convex program follows by combining the efficient description for Gaussian separability from~\cite{Lami_2018_1} with the characterisation of the relative entropy in terms of statistical moments from~\cite{Pirandola_2017}. Convexity of the resulting program is established by lifting the convexity of the relative entropy on states via our Gaussification argument. The technical derivations are presented in  Sec.\ \ref{Sec:Gaussian_Reverse_REE} and \ref{Sec:Convex_Program} of the Supplementary Material.

%%%%%%%%%%%%%%%%%%%%%%%%%%%%%%%%%%%%%%%%%%%%%%%%%%%%%%%%%%%%%%%%%%%%%%%%%%%
\paragraph{Thermal Attenuator Channel.}

As a paradigmatic example from optical communication, we consider the \emph{thermal attenuator channel} \(\mathcal{N}_{\lambda,n}\), which serves as the predominant model for realistic optical quantum links. This is also the simplest non-trivial example because the error exponent of the (previously well understood) pure-loss channel diverges\,---\,as it should, because zero-error entanglement generation over a pure-loss channel is possible via dual-rail encoding. Moreover, it is also precisely this type of \emph{noisier} channel, for which the capacity bound of Pirandola et al.~\cite{Pirandola_2017} is no longer tight, making them harder to understand from the \emph{quantity} perspective on quantum communication.

Mathematically, the single-mode thermal attenuator channel is modelled by mixing the input mode at a beamsplitter of transmissivity \( \lambda \in [0,1]\) with an ancillary mode prepared in a Gaussian thermal state with vanishing first moments and covariance matrix \(n \boldsymbol{\sigma}_0\), where \(\boldsymbol{\sigma}_0\) denotes the \(2 \times 2\) identity matrix. This precisely corresponds to the Gibbs state of the free Hamiltonian associated with the thermal noise parameter \(n \geq 1\). 

We find the following upper bound on its (zero-rate) error exponent:
\begin{align}\label{Eq:Upper_Attenuator}
	\mathrm{UB} \left( \mathcal{N}_{\lambda, n} \right) &= \frac{ n_\mathrm{sep} \left(  \arcoth(n) - \arcoth(n_\mathrm{sep}) \right)  }{\ln(2)}  + \log_2\left( \sqrt{ \frac{n^2-1}{n_\mathrm{sep}^2-1} } \right)  \quad \mathrm{with} \quad n_\mathrm{sep}(\lambda) := \frac{1+\lambda}{1-\lambda}
\end{align}
for \( 1 \leq n \leq n_\mathrm{sep}(\lambda) \) and zero otherwise (see Sec.\ \ref{App:Example} for details). As expected, this diverges for \(n \to 1 \), i.e.\ in the case of the pure-loss channel. Additionally, we find that, aesthetically pleasing, the asymptotically closest separable state coincides with the asymptotic Choi state obtained by sending the \emph{other} half of the two-mode squeezed vacuum state through a thermal attenuator channel with the same transmissivity and thermal noise parameter given by \(n_\mathrm{sep}(\lambda)\). 

We go on and show that Eq.\ \ref{Eq:Upper_Attenuator} is not far from optimal by obtaining the first lower bounds on the (zero-rate) error exponent of this channel. We do this by lifting a novel achievability result for the error exponent of LOCC-assisted channel coding from \cite{Berta_2026} in the zero-rate limit to Gaussian channels (see Sec.\ \ref{App:Achievability} for details). We obtain the following lower bound on our exponent:
\begin{equation}
    \mathrm{LB}\left( \mathcal{N}_{\lambda, n}\right) = - \log_2\left((1-\lambda) \left( n + \sqrt{n^2-1} \right) \right) \quad  \text{for} \quad 1 \geq \lambda > 1 - \left( n + \sqrt{n^2-1} \right)^{-1} \, .
\end{equation}
We find that our lower bound becomes tight as the thermal noise parameter \(n\) increases (see Fig.\ \ref{Fig:Quality Bound} for a comparison of our lower and upper bounds). 

\paragraph{Thermal Amplifier Channel.} Another essential example for optical communication is the \emph{thermal amplifier channel} \(\mathcal{A}_{\eta, n}\), which models amplification in the presence of thermal noise. It can be described by the interaction with a thermal mode\,---\,with covariance matrix \(n \boldsymbol{\sigma}_0\), and null first moments\,---\,through a two-mode squeezing with gain \(\eta \geq 1\). 

We find that its error exponent is bounded qualitatively by the same expression as the thermal attenuator:
\begin{align}
	\mathrm{UB}(\mathcal{A}_{\eta,n}) = \frac{ n_\mathrm{sep} \left(  \arcoth(n) - \arcoth(n_\mathrm{sep}) \right)  }{\ln(2)}  + \log_2\left( \sqrt{ \frac{n^2-1}{n_\mathrm{sep}^2-1} } \right)  \quad \mathrm{with} \quad n_\mathrm{sep}(\eta) := \frac{\eta +1}{\eta-1}
\end{align}
for \( 1 \leq n \leq n_\mathrm{sep}(\eta) \) and zero otherwise (see Sec.\ \ref{App:Example} for details). As above, this diverges for \(n_\mathrm{th} \to 1\), which corresponds to the quantum limited amplifier. Moreover, we find that the error exponent is bounded from below by
	\begin{equation}
    	\mathrm{LB}\left( \mathcal{A}_{\eta, n}\right) = - \log_2\left(\frac{\eta-1}{\eta} \left( n + \sqrt{n^2-1} \right) \right) \quad  \text{for} \quad 1 \leq \eta < 1 + \left( n-1 + \sqrt{n^2-1} \right)^{-1}
	\end{equation}
(see Sec.\ \ref{App:Achievability} for details).

\paragraph{Additive-Noise Channel.} Perhaps the conceptually simplest model of decoherence in optical communication is given by the additive-noise Gaussian channel \(\mathcal{N}_\mu\). This can be seen as the action of random Gaussian displacement on the input mode and is modelled by adding the covariance matrix \(\mu \boldsymbol{\sigma}_0\), with \(\mu \geq 0\), to the input covariance matrix. 

We find that its error exponent is bounded above by 
\begin{equation}
	\mathrm{UB}(\mathcal{N}_\mu) = \frac{2 - \mu}{\mu \ln(2) } + \log_2\left( \frac{\mu}{2} \right)
\end{equation}
provided that \( 0 \leq \mu \leq 2 \) and zero otherwise. From below, we find
\begin{equation}
	\mathrm{LB}(\mathcal{N}_\mu) = - \log_2(2\mu) \quad \mathrm{for} \quad 0 \leq \mu \leq \frac{1}{2} \, .
\end{equation}
The technical derivations can be found in Sec.\ \ref{App:Example} and \ref{App:Achievability}.
\begin{figure}[htbp]
 	\centering
    \includegraphics[width=\linewidth]{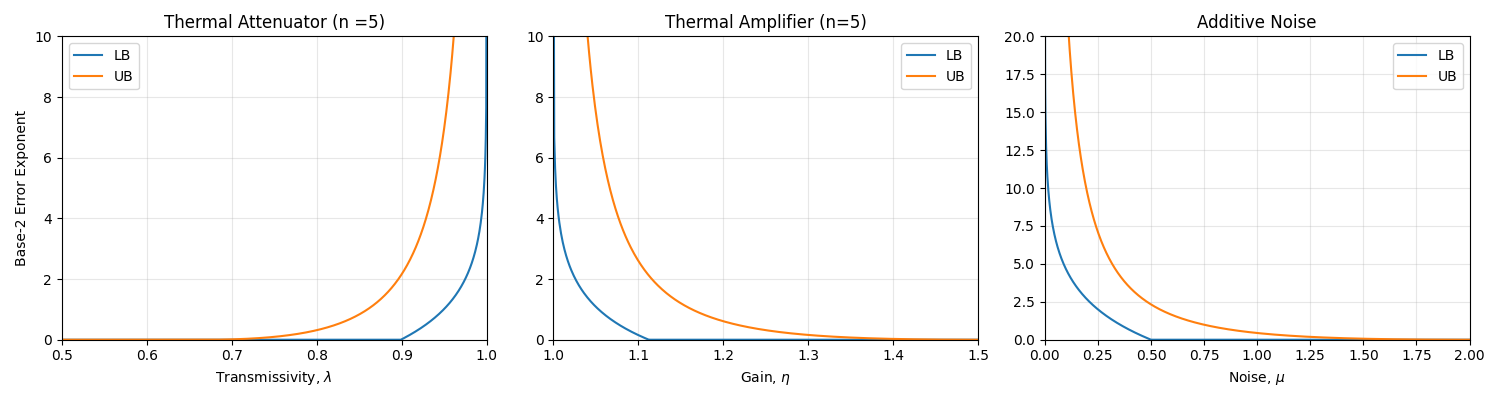}
    \caption{Comparison of our converse upper bound (UB, orange) with our achievability lower bound (LB, blue) for the thermal attenuator, thermal amplifier and Gaussian additive-noise channel, respectively.  For the former two channels, we fix the thermal noise parameter at \(n = 5 \). For the thermal attenuator, the bound shown is achieved by a forward-assisted protocol, whereas for the amplifier it is achieved via a backward-assisted protocol. For the additive-noise channel, both protocols yield the same bound. See Sec.\ \ref{App:Achievability} of the Supplementary material for details.}
    \label{Fig:Quality Bound}
\end{figure}

%%%%%%%%%%%%%%%%%%%%%%%%%%%%%%%%%%%%%%%%%%%%%%%%%%%%%%%%%%%%%%%%%%%%%%%%%%%
\paragraph{Operational Interpretations.} The reverse relative entropy of entanglement was endowed with an operational interpretation in~\cite{Lami_2024_2}, but this was limited to finite-dimensional systems only. In this work, we rigorously extend this interpretation to \emph{general} (separable) infinite-dimensional quantum systems, establishing the reverse relative entropy of entanglement as an operationally meaningful entanglement measure in this setting. 

We start by adopting the quality framework from the prior work \cite{Lami_2024_2} to analyse the task of entanglement distillation. Consider again Alice and Bob that now share \(n\) copies of a bipartite (infinite-dimensional) state \(\rho_{AB}\) with the aim to convert them into pure, maximal entanglement. To that end, they apply some protocol \(\Lambda_n\) such that, when acting on \(n\) copies of \(\rho_{AB}\), the final state approximates \(m\) copies of the maximally entangled state \(\Phi\). Regarding the set of feasible protocols, we follow \cite{Lami_2024_2} and relax the physically well-motivated but mathematically difficult LOCC framework to the class of non-entangling operations (NE)~\cite{Brandao_2010_1, Brandao_2010_2, Brandao_2015, Hayashi_2024, Lami_2024_1}\,---\,by definition, \(\Lambda_n(\sigma_n)\) must remain unentangled for all unentangled states \(\sigma_n\). 

Considering the quality of the distilled entanglement as the figure-of-merit, leads to the following definition for the (zero-rate) \emph{error exponent of entanglement distillation}:
\begin{equation}
    E_{d, \mathrm{err}}(\rho_{AB}) := \lim_{m \to \infty} \sup \bigg\{ \lim_{n \to \infty} - \frac{1}{n} \log \varepsilon_n : \Lambda_n( \rho^{\otimes n} ) \approx_{\varepsilon_n} \Phi^{\otimes m}, \; \Lambda_n \in \mathrm{NE} \bigg\} \, .
\end{equation}

We then show that, as in the finite-dimensional case, this exponent is fundamentally connected to a task from quantum state discrimination~\cite{Hiai_1991,Ogawa_2000, Hayashi} known as entanglement testing (see Sec.\ \ref{Sec:Error_Exponent_Distillation} for the technical details). In entanglement testing, the goal is to distinguish the given entangled state \(\rho_{AB}^{\otimes n}\) from the set of all separable states by performing a collective measurement on the global system. As usual, one can distinguish between two types of errors: the type-1 error occurs when mistaking \(\rho^{\otimes n}_{AB}\) for a separable state; conversely, the type-2 error occurs when mistaking a separable state for \(\rho^{\otimes n}_{AB}\). For a fixed type-2 error probability, the type-1 error decays exponentially and the asymptotically optimal error exponent is referred to as the \emph{Sanov exponent}, denoted \( \mathrm{Sanov}(\rho_{AB}\| \mathcal{S}_{A:B}) \).

As our main operational interpretation, we then prove that the latter exponent coincides exactly with the reverse relative entropy of entanglement, giving the solution to the \emph{generalised quantum Sanov theorem}~\cite{Bjelakovic_2005, Noetzel_2014,Hayashi_2025_2, Lami_2024_2} of entanglement testing (see Sec.\ \ref{Sec:Sanov} for the full argument). To summarise, we establish for (separable) infinite-dimensional systems the equalities:
\begin{align}
	E_{d, \mathrm{err}}(\rho_{AB}) = \mathrm{Sanov}(\rho_{AB}\| \mathcal{S}_{A:B})  = D(\mathcal{S}_{A:B} \|\rho_{AB}) \, .
\end{align}

It is noteworthy that the proof does \emph{not} follow from a standard truncation argument. Indeed, even if we truncate the space so that the tails of $\rho_{AB}$ are eliminated, there is no a priori guarantee that the separable state we have to discriminate it from will lie in the same truncated space: the technical difficulty of our result lies in proving precisely this. Once that is done, our lifting procedure crucially relies on a semi-continuity property of the reverse relative entropy of entanglement, which we establish with a particular choice of operator topology.

%%%%%%%%%%%%%%%%%%%%%%%%%%%%%%%%%%%%%%%%%%%%%%%%%%%%%%%%%%%%%%%%%%%%%%%%%%%
\begin{center}
    {\large \bfseries Outlook \par}
\end{center}

In this work, we have studied quantum communication from a quality perspective, establishing a fundamental bound on the error exponent of two-way assisted quantum communication. This bound parallels the bound of Pirandola et.\ al.~\cite{Pirandola_2017} (see also Bennett et.\ al.~\cite[Section~5]{Bennett_1996_3}) for the two-way assisted capacity, which turns out to be very tight even for unassisted communication in certain regimes. Moreover, the derivation of the single-letter capacity bound follows by a generally non-tight relaxation of the regularised relative entropy of entanglement to its single-letter version. Crucially, this relaxation step does not occur in our analysis due to the additivity of the reversed entanglement measure. Additionally, we have provided an analysis of the achievability of our exponent using a random-coding result based on the Choi state of the channel. This could potentially be improved upon and makes a thorough investigation of the achievability of our bound a promising avenue for future work.

Additionally, in the Gaussian case, we showed that our bound can be computed via a finite-dimensional convex program with two positive semi-definite constraints. As such, it can straightforwardly be solved with off-the-shelf solvers based on interior-point methods, which are well-known to be efficient in praxis \cite{Boyd}. However, it might still be desirable to establish a precise complexity-theoretic result for the efficiency of this program. For this, one would (after rewriting it into a standard conic program using the epigraph formulation) need to construct a self-concordant logarithmic barrier function of the resulting convex cone. This would then give provable convergence guarantees using the barrier method, i.e.\ polynomial-time solvability in the number of modes. A possible path to achieve this would be to adapt the specialised literature on relative entropy optimisation such as~\cite{Fawzi_2023} to this special case (see also Remark \ref{Rem:Efficiency} in the Supplementary Material). 

Another interesting direction for future work is to extend our analysis to general quantum resource theories \cite{Chitambar_2019}. Under certain assumptions, the reverse relative entropy of resource attains an operational interpretation in the task of resource testing (see the Supplementary Material and \cite{Lami_2024_2} for more details). Thus, another interesting open question is whether the resulting measure is also efficiently computable for other Gaussian resource theories, such as non-classicality or, more generally, \(\lambda\)-negativity.

%%%%%%%%%%%%%%%%%%%%%%%%%%%%%%%%%%%%%%%%%%%%%%%%%%%%%%%%%%%%%%%%%%%%%%%%%%%
\paragraph{Acknowledgments.} The authors thank Bartosz Regula for discussions in the early stages of the project. TR thanks Julius Zeiss for insightful discussions at numerous stages of this work. MB and TR acknowledge support from the Excellence Cluster - Matter and Light for Quantum Computing (ML4Q-2) and funding by the European Research Council (ERC Grant Agreement No. 948139). G.A. acknowledges funding from the UK Engineering and Physical Sciences Research Council (EPSRC Grant No. EP/X010929/1). LL acknowledges financial support from the European Union (ERC StG ETQO, Grant Agreement no.\ 101165230). Views and opinions expressed are however those of the authors only and do not necessarily reflect those of the European Union or the European Research Council. Neither the European Union nor the granting authority can be held responsible for them.

%%%%%%%%%%%%%%%%%%%%%%%%%%%%%%%%%%%%%%%%%%%%%%%%%%%%%%%%%%%%%%%%%%%%%%%%%%%
\printbibliography

%%%%%%%%%%%%%%%%%%%%%%%%%%%%%%%%%%%%%%%%%%%%%%%%%%%%%%%%%%%%%%%%%%%%%%%%%%%
%%%%%%%%%%%%%%%%%%%%%%%%%%%%%%%%%%%%%%%%%%%%%%%%%%%%%%%%%%%%%%%%%%%%%%%%%%%

\begin{center}
	\vspace{0.5cm}
    {\LARGE \bfseries Supplementary Material \par}
    \vspace{0.5cm}
\end{center}

\setcounter{section}{0}

\numberwithin{equation}{section}
\setcounter{equation}{0}

%%%%%%%%%%%%%%%%%%%%%%%%%%%%%%%%%%%%%%%%%%%%%%%%%%%%%%%%%%%%%%%%%%%%%%%%%%%

\section{Notation and Preliminaries}

Let us start with the relevant mathematical preliminaries. We refer the reader to Holevo's textbook~\cite{Holevo} for a detailed treatment of general quantum information theory and Serafini's textbook~\cite{Serafini} as well as the review article by Weedbrook et al.~\cite{Weedbrook_2012} for details on the Gaussian theory. 

Throughout this work, we consider quantum systems that are represented mathematically by a separable (infinite-dimensional) Hilbert space \(\mathcal{H}\). Two important sets of linear operators on \(\mathcal{H}\) are the set of bounded operators \(\mathcal{B}(\mathcal{H})\) and the set of trace-class operators \(\mathcal{T}(\mathcal{H})\). The former consists of all operators with finite operator norm \(\norm{\cdot}_\infty \) and the latter of those with finite trace norm \(\norm{\cdot}_1 \).\footnote{The trace norm of the operator \(X\) is defined via \(\norm{X}_1 := \tr{ \sqrt{X X^\dagger} }\).} Equipping each set with the respective norm makes them Banach spaces, and we have the duality \(\mathcal{T}(\mathcal{H})^\star= \mathcal{B}(\mathcal{H})\) on the level of Banach spaces. 

An operator \(X\) is called positive semi-definite (PSD), denoted \(X\geq 0\), if \( \bra{\psi} X \ket{\psi} \geq 0 \) for all \(\ket{\psi} \in \mathcal{H}\). If the inequality is strict for all nonzero \(\ket{\psi}\), the operator is referred to as positive definite, denoted as \(X > 0\). The set of positive semi-definite bounded operators \( \mathcal{B}_+(\mathcal{H}) \) forms a convex cone that induces a partial order on \(\mathcal{B}(\mathcal{H})\): for \(X, Y \in \mathcal{B}(\mathcal{H})\), we write \( X \leq Y \) if and only if \( Y - X \in \mathcal{B}_+(\mathcal{H})\). Analogously, we will denote by \(\mathcal{T}_+(\mathcal{H})\) the cone of positive semi-definite trace-class operators.

A \emph{quantum state} is represented by a density operator on \(\mathcal{H}\), i.e.\ a positive semi-definite operator that is normalised (trace-class with unit trace). The convex set of states on \(\mathcal{H}\) will be denoted  by
\begin{equation}
	\mathcal{D}(\mathcal{H}) := \bigg\{ X \in \mathcal{T}(\mathcal{H}) : X \geq 0, \tr{X} = 1 \bigg\} \, .
\end{equation}
If the density operator is a positive definite operator, i.e.\ \(\rho > 0 \), we will also refer to it as \emph{faithful}. 

Given \(N\) quantum systems with \(N \in \mathbb{N}_+\) each associated with a separable Hilbert space \(\mathcal{H}_i\), the composite system is represented by the tensor product Hilbert space \(\mathcal{H} =  \bigotimes_{i=1}^N\mathcal{H}_i\). A fundamental phenomenon that appears in these multipartite systems is entanglement. Considering a bipartite quantum system \(\mathcal{H}_A \otimes \mathcal{H}_B\), we denote with \(\mathcal{S} = \mathcal{S}_{A:B}\) the set of \emph{separable} %\footnote{Note that this has nothing to do with the separability of the Hilbert space.} [I THINK THIS COMMENT IS SUPERFLUOUS -- GA]
 states defined as
\begin{equation}\label{Eq:Separable}
	\mathcal{S}(\mathcal{H}) := \text{cl}_\text{tn} \bigg( \text{conv} \bigg\{ \ketbra{\psi}{\psi}_A \otimes \ketbra{\phi}{\phi}_B : \ket{\psi}_A \in \mathcal{H}_A, \ket{\phi}_B \in \mathcal{H}_B, \braket{\psi}{\psi}_A =  1 =\braket{\phi}{\phi}_B \bigg\} \bigg) \, ,
\end{equation}
i.e.\ the closed (w.r.t.\ the trace norm topology) convex hull of the set of pure product states~\cite{Werner_1989}. Entanglement is then defined by negation, i.e.\ a state is called \emph{entangled} if it is not separable.

A \emph{quantum channel} is a bounded linear map \(\Lambda : \mathcal{T}(\mathcal{H}_1) \mapsto \mathcal{T}(\mathcal{H}_2) \) that is completely positive and trace-preserving (CPTP). Here, complete positivity means that \(\Lambda \otimes \mathcal{I}_{\mathcal{H}^\prime} \)  maps positive semi-definite operators on \(\mathcal{H}_1 \otimes \mathcal{H}^\prime \) to positive semi-definite operators on \( \mathcal{H}_2 \otimes \mathcal{H}^\prime \), for all auxiliary Hilbert spaces \(\mathcal{H}^\prime\). Trace preservation in turn means that \( \tr{\Lambda(X) } = \tr{X}\) for all inputs \(X\). The set of all such operations will be denoted \(\mathrm{CPTP}(\mathcal{H}_1 \to \mathcal{H}_2)\). Moreover, we define the adjoint \(\Lambda^\dagger\) of the channel \(\Lambda\) as the linear map \(\Lambda^\dagger : \mathcal{B}(\mathcal{H}_2) \mapsto \mathcal{B}(\mathcal{H}_1) \) that satisfies \( \tr{ X_1 \Lambda^\dagger(Y_2) } = \tr{ \Lambda(X_1) Y_2 } \) for all \(X_1 \in \mathcal{T}(\mathcal{H}_1) \) and \( Y_2 \in \mathcal{B}( \mathcal{H}_2  )\). 

An important subclass of channels are \emph{quantum measurements}. Mathematically, they are represented by a Positive Operator-Valued Measures (POVM). For a measurement with finite number of outcomes, the associated POVM is a finite collection of positive semi-definite bounded operators \(M_i \geq 0\) that form a resolution of the identity, i.e.\ \( \sum_i M_i = \Id_{\mathcal{H}} \) with \(\Id_{\mathcal{H}}\) the identity operator on \(\mathcal{H}\).
 
\begin{remark}
	We will usually denote Hilbert spaces by the capital letters \(A, B, C\) etc.\ and use subscripts to denote which space an operator acts on. However, to simplify the notation, we will drop the explicit reference to the Hilbert space whenever it is clear from context. 
\end{remark}

We now specialise to bosonic continuous-variable (CV) systems with a finite number of modes. A single bosonic mode is mathematically represented by a pair of self-adjoint quadrature operators \(( \hat{x} ,\hat{p})\) that satisfy the \emph{canonical commutation relations} (CCR), 
	\begin{equation}
		\com{\hat{x} }{ \hat{p}} = i  \, ,
	\end{equation}
	where we used natural units with \(\hbar = 1\). The Hilbert space of each bosonic mode is separable but necessarily infinite-dimensional. Its canonical basis is given by the Fock states \( \{ \ket{n} \}_{n=0}^{\infty}\), which are the eigenvectors of the associated number operator \( \hat{N} = \hat{x}^2 + \hat{p}^2 - \frac{1}{2} \).
	
	An \(m\)-mode CV quantum system with \(m \in \mathbb{N}_+\) is associated with the tensor-product Hilbert space \(\mathcal{H} =  \bigotimes_{i=1}^m\mathcal{H}_i\), where \(\mathcal{H}_i\) denotes the infinite-dimensional separable Hilbert space of the \(i\)-th mode. Using the vector notation \( \hat{\boldsymbol{r}} := (\hat{x}_1, \hat{p}_1, ... , \hat{x}_m, \hat{p}_m)^T \), we can express the CCR compactly as \( \com{ \hat{\boldsymbol{r}} }{ \hat{\boldsymbol{r}}^T } = i \boldsymbol{\Omega} \). Here, the commutator on the left side denotes a \(2m \times 2m\) matrix with entries \( \com{ \hat{r}_i }{ \hat{r}_j } \), and on the right side we have the symplectic form
    \begin{align}
        \boldsymbol{\Omega} &:= \bigoplus_{j=1}^{m} \boldsymbol{\Omega}_1 & & \text{with} &\boldsymbol{\Omega}_1 &= \begin{bmatrix} 0 & 1 \\ -1 & 0 \end{bmatrix} \, .
    \end{align}
    
 \begin{remark}
 	Throughout, we use bold symbols to denote finite-dimensional vectors and matrices. 
 \end{remark}

The most important class of CV quantum states are the so-called \emph{Gaussian} states. These are the Gibbs states of quadratic Hamiltonians in the quadrature operators and completely characterised by their first and second statistical moments. A Gaussian state \( \rho_G := \rho_G[ \boldsymbol{\mu},\boldsymbol{V} ]\) is then uniquely specified by its displacement vector \( \boldsymbol{\mu} \), defined via \( \mu_j := \tr{ \rho \hat{r}_j } \), and \emph{covariance matrix} \( \boldsymbol{V}\) with entries
	\begin{equation}
		V_{j,k} := \tr{ \rho \anticom{\hat{r}_j - \mu_j}{  \hat{r}_k - \mu_k} } \, ,
	\end{equation}
	where \( \{ \cdot , \cdot \}\) denotes the anti-commutator.\footnote{Note that different conventions for the covariance matrix are used in the literature. We follow here the convention used in Serafini's textbook~\cite{Serafini}.}
	
The covariance matrix is a \(2m \times 2m\), real and symmetric matrix which must satisfy the Heisenberg uncertainty principle	\begin{equation}\label{Eq:Bona_Fide}
		\boldsymbol{V} + i \boldsymbol{\Omega} \geq 0 \, .
	\end{equation}
	The latter ensures that \(\boldsymbol{V}\) is a \emph{bona fide} quantum covariance matrix. Note that Eq.~\eqref{Eq:Bona_Fide} implies in particular that the covariance matrix must be positive definite. By Williamson's theorem~\cite{Williamson_1936}, any positive definite \(\boldsymbol{V}\) admits a \emph{Williamson decomposition} \( \boldsymbol{V} = \boldsymbol{S} \boldsymbol{D} \boldsymbol{S}^T\), where \(\boldsymbol{S} \in \mathrm{Sp}(2m)\) is a symplectic matrix and \( \boldsymbol{D} = \mathrm{diag}(\nu_1, \nu_1, ..., \nu_m, \nu_m)\) contains the symplectic eigenvalues. Here, the symplectic group \( \mathrm{Sp}(2m) \)\footnote{We refer the interested reader to~\cite{Arvind_1995} for a review of the symplectic group in physics.} is defined as the set of transformations that preserve \(\boldsymbol{\Omega}\) by congruence, i.e.\
	\begin{equation}
		\boldsymbol{S} \in \mathrm{Sp}(2m) \iff \boldsymbol{S} \boldsymbol{\Omega} \boldsymbol{S}^T = \boldsymbol{\Omega} \, ,
	\end{equation}
	and the symplectic spectrum \(\{ \nu_i\}_{i=1}^m\) is given by the standard eigenspectrum of \(|i \boldsymbol{\Omega} \boldsymbol{V} |\).

In the following, we will denote the set of Gaussian states on \(\mathcal{H}\) by \(\mathcal{G}(\mathcal{H})\). According to \cite[Appendix A]{Banchi_2015} (see also~\cite{Chen_2005,Krueger, Holevo_2010}), we can express the density operator of an arbitrary Gaussian state in Gibbs-type exponential form as
\begin{equation}\label{Eq:Gibbs_Form}
    \rho_G [\boldsymbol{\mu}, \boldsymbol{V}] = \frac{1}{Z[\boldsymbol{V}]} \cdot \exp \bigg[- ( \boldsymbol{\hat{r}} - \boldsymbol{\mu})^T \boldsymbol{G}[\boldsymbol{V}]( \boldsymbol{\hat{r}} - \boldsymbol{\mu})  \bigg]
\end{equation}
with the Gibbs matrix \( \boldsymbol{G}[\boldsymbol{V}] =  i \boldsymbol{\Omega} \arcoth(  \boldsymbol{V} i \boldsymbol{\Omega} ) \) and normalisation factor \( Z[\boldsymbol{V}] =  \sqrt{ \det \left( \frac{1}{2} \left( \boldsymbol{V} + i \boldsymbol{\Omega} \right)  \right) }\). 

\begin{remark}
	We use \(\exp(\cdot) \) to denote the inverse of the natural logarithm \(\ln(\cdot)\).
\end{remark}

Due to the simple description of Gaussian states in terms of their statistical moments, a large and important class of quantum channels are similarly easy to describe. By definition, a \emph{Gaussian} channel is a channel that maps Gaussian states to Gaussian states. Its action on Gaussian states is completely characterised by two \(2m \times 2m \) real matrices \(\boldsymbol{X}\) and \(\boldsymbol{Y}\) with \(\boldsymbol{Y} = \boldsymbol{Y}^T\), which act on the statistical moments of the state as
\begin{align}
	\boldsymbol{\mu} &\mapsto \boldsymbol{X} \boldsymbol{\mu} && \text{and} & \boldsymbol{V} \mapsto \boldsymbol{X} \boldsymbol{V} \boldsymbol{X}^T + \boldsymbol{Y} \, .
\end{align}

The \emph{Choi-Jamiolkowski isomorphism} is a well-known bijective mapping between quantum states and quantum channels. For finite-dimensional systems \(A\), it is defined via the canonical maximally entangled state \(\Phi_d = \frac{1}{d} \sum_{i,j=1}^d \ketbra{i}{j}_A \otimes \ketbra{i}{j}_{A^\prime} \), where \(d = \mathrm{dim}(A)\) and \(A^\prime \simeq A\). The infinite-dimensional analogue of \(\Phi_d\) is not normalisable and thus not a valid quantum state. However, one can approach it by a limit of normalisable Gaussian states~\cite{Giedke_2002} (see also~\cite{Holevo_2011} for an alternative definition). Following an operational approach, we consider the two-mode squeezed vacuum state (TMSV) given by
	\begin{equation}
		\ket{\Phi (r)  } := \frac{1}{\cosh (r) } \sum_{n=0}^\infty \tanh^n(r)\ket{ n }_A \otimes \ket{ n }_{A^\prime}
	\end{equation}
with the squeezing parameter \(r \in [0, \infty)\). In the limit \(r \to \infty\), the state \( \ket{\Psi(r)}\) tends to an evenly weighted superposition of tensor products of Fock states and hence approaches the canonical maximally entangled state on \(A \otimes A^\prime\). With this, we then define the \emph{quasi-Choi state} of the CPTP-map \( \mathcal{N}\) that acts on \(m\) bosonic modes via
\begin{equation}\label{Def:Choi}
	\rho_{\mathcal{N}}(r) := ( \mathcal{N}_{A \to B} \otimes \mathcal{I}_{A^\prime} ) ( \ketbra{\Phi(r)}{\Phi(r)}^{\otimes m} ) \, ,
\end{equation}
where \(\mathcal{I}\) denotes the identity channel. 

Given a Gaussian channel \(\mathcal{N} := \mathcal{N}(\boldsymbol{X}, \boldsymbol{Y}) \), its quasi-Choi state \(\rho_{\mathcal{N}}(r)\) is a Gaussian state with covariance matrix
\begin{align}
	\boldsymbol{V}_{\mathcal{N} }(r) & = \begin{pmatrix}
		\cosh(2r)  \boldsymbol{X} \boldsymbol{X}^T + \boldsymbol{Y} & \sinh(2r) \boldsymbol{X} \boldsymbol{\Sigma}_3 \\ \sinh(2r) \boldsymbol{\Sigma}_3 \boldsymbol{X}^T & \cosh(2r) \boldsymbol{\Sigma}_0
	\end{pmatrix} && \text{with} &\boldsymbol{\Sigma}_i &:= \bigoplus_{j=1}^m \boldsymbol{\sigma}_i \, .
\end{align}
Here, we introduced the Pauli sigma matrices, defined as
 	\begin{align}
 		\boldsymbol{\sigma}_0 &:= \begin{pmatrix} 1 \, , & 0 \\ 0 &  1 \end{pmatrix}\, ,& \boldsymbol{\sigma}_1 &:= \begin{pmatrix} 0 & 1 \\ 1 &  0 \end{pmatrix} \, , & \boldsymbol{\sigma}_2 &:= \begin{pmatrix} 0 & -i \\ i &  0 \end{pmatrix} \, , && \text{and} & \boldsymbol{\sigma}_3 &:= \begin{pmatrix} 1 & 0 \\ 0 &  -1 \end{pmatrix} \, .
 	\end{align}

%%%%%%%%%%%%%%%%%%%%%%%%%%%%%%%%%%%%%%%%%%%%%%%%%%%%%%%%%%%%%%%%%%%%%%%%%%%
%%%%%%%%%%%%%%%%%%%%%%%%%%%%%%%%%%%%%%%%%%%%%%%%%%%%%%%%%%%%%%%%%%%%%%%%%%%
\section{Infinite-Dimensional Quantum Systems}

We start the technical analysis with general infinite-dimensional quantum systems (of type I). In what follows, our main goal is to formally establish the reverse relative entropy of entanglement as an entanglement measure with a precise operational interpretation in the context of entanglement distillation. This lifts the operational interpretation established in the finite-dimensional work~\cite{Lami_2024_2}.

%%%%%%%%%%%%%%%%%%%%%%%%%%%%%%%%%%%%%%%%%%%%%%%%%%%%%%%%%%%%%%%%%%%%%%%%%%%
\subsection{Reverse Relative Entropy of Entanglement}\label{Sec:Reverse_REE}

In the literature, a plethora of entanglement measures have been studied (see e.g.~\cite{Horodecki_2001, Vidal_2009} for reviews). A prominent subclass is formed by measures derived from a distance function on state space; here, the state's distance to the set of separable states (w.r.t.\ the chosen distance function) is used as a quantifier of the state's entanglement. In this work, we investigate one such measure derived from the relative entropy function. 

Given two positive semi-definite trace class operators \( X, Y \in \mathcal{T}_+(\mathcal{H}) \) -- with spectral decomposition \(X = \sum_i x_i \ketbra{e_i}{e_i}\) and \(Y = \sum_j y_j \ketbra{f_j}{f_j} \) -- their \emph{relative entropy}~\cite{Umegaki_1962,Lindblad_1973} is defined as
 \begin{align}
	D(X\|Y) &:= \tr{ X ( \log X - \log Y) + Y - X}  \\
					 &:= \sum_{i,j} |\braket{e_i}{f_j} |^2 \left( x_i \log x_i - x_i \log y_i + y_i - x_i \right) \label{Eq:Relative_Entropy}\, .
\end{align} 

\begin{remark}
	We will use \(\log(\cdot) \) to denote the logarithm to base two, corresponding to the canonical choice of measuring (quantum) information in (qu-)bits.
\end{remark}

Hereby, the expression in the first line is to be understood as specified by the second line. As explained in \cite[Section 2]{Lindblad_1973}, the convexity of \(x\log x\) ensures that all summands in Eq.~\eqref{Eq:Relative_Entropy} are positive.\footnote{Explicitly, one uses that a convex differentiable function satisfies \(f(y) \geq f(x) + f^\prime (x) \cdot (y-x)\) for all \(x,y\) in its domain.} Hence, the sum is well-defined (albeit possibly infinite) and the order of summation is irrelevant. By convention, we set \(D(0\|0) = 0\) and \(D(X\|0) = + \infty\) if \(X \not = 0\). It is then evident from Eq.~\ref{Eq:Relative_Entropy} that a necessary condition for \(D(X\|Y) < \infty\) is that \(X\ll Y\), i.e.\ the support of the first argument is contained in the support of the second.\footnote{Note that only in the case of a finite-dimensional Hilbert space, this is also sufficient.}

Let us also briefly mention when we can simplify Eq.~\eqref{Eq:Relative_Entropy} into a more familiar form. For this, we define the (von-Neumann) \emph{entropy} of \(X\) via
\begin{equation}
	H(X) := - \tr{X \log X} :=\sum_i -x_i \log x_i \, .
\end{equation}
Note that \(- x \log x\) is non-negative for all \(x \in [0,1]\) and, as \(X\) is trace-class, only finitely many of its eigenvalues can lie outside this interval; thus, the sum is well-defined (albeit possibly infinite). Provided that  \(H(X) < \infty \), Eq.~\eqref{Eq:Relative_Entropy} can then be reduced to the well-known expression
\begin{equation}
	D(X\| Y)  = - H(X) - \tr{X \log Y} + \tr{Y} -\tr{ X}
\end{equation}
(see \cite[Section II.B]{Lami_2023} for more details). Under the assumption of finite entropy, the above expression is well-defined as the only term that can possibly diverge is the second, which is to be understood as the series 
 \begin{equation}
 	- \tr{X \log Y} := - \sum_{i,j} | \braket{e_i}{f_j} |^2  x_i \log y_i \, .
 \end{equation}

The quantum Stein lemma~\cite{Hiai_1991, Ogawa_2000} endows the relative entropy with an operational interpretation in asymmetric quantum hypothesis testing. This allows us to interpret it as a statistical distance measure on the set of quantum states. However, importantly, it is not a distance function in the strict mathematical sense, as it is not symmetric in general. This asymmetry gives rise to two possible entanglement measures, depending on whether the optimisation over the separable set is carried out in the first or second argument. The standard choice is the second argument, resulting in the well-known \emph{relative entropy of entanglement}~\cite{Vedral_1997, Vedral_1998}. Here, we focus instead on the other case, which -- using the terminology of~\cite{Vedral_1997, Eisert_2003} -- yields the \emph{reverse relative entropy of entanglement}.
\begin{definition}\label{Def:Reverse_REE}
	Let \( \mathcal{H}_{AB} = \mathcal{H}_A \otimes \mathcal{H}_B\) be a bipartite separable (possibly infinite-dimensional) Hilbert space and \( \mathcal{S} = \mathcal{S}_{A:B}\) be the set of separable states on \(\mathcal{H}_{AB}\). The reverse relative entropy of entanglement of the state \(\rho \in \mathcal{D}\) is then defined as 
\begin{equation}
    D( \mathcal{S} \| \rho) := \inf_{\sigma \in \mathcal{S} } D(\sigma\|\rho) \, .
\end{equation}
\end{definition}

\begin{remark}
	Note that this definition can be extended straightforwardly to the framework of general quantum resource theories (see~\cite{Chitambar_2019} for a review). Here, one defines a set of \emph{free states} \(\mathcal{F} \subseteq \mathcal{D} \) that are said to contain no resource in the context of the specific theory. The \emph{reverse relative entropy of resource} \(D(\mathcal{F}\|\rho) \) of the state \(\rho\) is then analogously defined as the relative entropy distance to the set of free states when optimising w.r.t.\ the first argument. 	
\end{remark}

It is straightforward to show that the reverse relative entropy of entanglement is a measure of entanglement in the resource-theoretic sense of~\cite{Chitambar_2019}. We collect its properties in the following lemma using the language of general quantum resource theories. Note that the main computational advantage of this measure is that it is additive, therefore eliminating the need for regularisation. Additionally, we obtain a semi-continuity result that is similar to \cite[Theorem 5]{Lami_2023} and will be a key ingredient in later proofs. 
\begin{lemma}\label{Lem:Reverse_REE_Properties}
		Let \( \mathcal{H}\) be a separable (infinite-dimensional) Hilbert space and \(\mathcal{F} \subseteq \mathcal{D}\) a trace-norm closed set of free states. Then, the reverse relative entropy of resource \( D( \mathcal{F} \| \rho)  \) is a valid resource monotone. That is, the functional is
	\begin{enumerate}
		\item Faithful, i.e. \( D( \mathcal{F} \| \rho) \geq 0 \) and \( D (\mathcal{F} \| \rho) = 0 \) if and only if \( \rho \in \mathcal{F}\).
		\item Monotone under free operations, i.e.\ \( D( \mathcal{F} \|  \Lambda(\rho ) ) \leq D ( \mathcal{F} \| \rho ) \) for any channel \(\Lambda \) such that \( \Lambda(\rho) \in \mathcal{F}\) for all \(\rho \in \mathcal{F}\).
	\end{enumerate}
	
	Moreover, the map \( \rho \mapsto D( \mathcal{F} \| \rho)  \)  is
	\begin{enumerate}
		\setcounter{enumi}{2}
		\item Convex if \(\mathcal{F}\) is closed under convex combinations.
		\item Additive on tensor products if \( \mathcal{F} \) is closed under tensor products and partial traces. 
		\item Lower semicontinuous in the trace norm topology if \(\mathrm{cone}(\mathcal{F}) := \{ \lambda \sigma : \lambda \in [0,\infty), \sigma \in \mathcal{F} \} \) is weak\(^\star\)-closed. Moreover, in this case there always exists \(\sigma_\star \in \mathcal{F}\) such that \( D( \sigma_\star \| \rho ) = D(\mathcal{F}\|\rho) \).
	\end{enumerate}
	
	In particular, all the above properties hold when specialising to \(\mathcal{F} = \mathcal{S}\) (see \cite[Lemma 25]{Lami_2021} for the closure of  \(\mathrm{cone}(\mathcal{S})\) in the weak\(^\star\)-topology).
\end{lemma}

\begin{remark}\label{Rem:Weak_Star}
	The weak\(^\star\)-topology on \(\mathcal{T}(\mathcal{H})\) is the topology induced on the space of trace-class operators by thinking of it as the dual of the space of compact operators on \(\mathcal{H}\) (see e.g.\ \cite[Chapter 2]{Megginson} for details). The closure of \(\mathrm{cone}(\mathcal{F})\) can be established e.g.\ via the method from \cite[Theorem 7]{Lami_2023}. 
\end{remark}

\begin{remark}
	Note that in finite-dimensional resource theories, one typically demands that a resource montone should satisfy asymptotic continuity (see e.g.\ \cite[Section VI.A]{Chitambar_2019}). However, as argued in \cite[Section II.C]{Lami_2021}, in the infinite-dimensional setting lower semicontinuity is the more natural continuity requirement. 
\end{remark}

\begin{proof} Most of these properties -- except the lower semi-continuity result -- are well established for the finite-dimensional case in \cite[Section III]{Eisert_2003}. For the sake of completeness, we verify that they carry over into our infinite-dimensional setting. The proof is mostly standard using the properties of the relative entropy (cf.\ e.g.\ \cite[Chapter 5]{Ohya}). However, for our proof of lower semi-continuity, we additionally need some results from general Banach space theory (see e.g.\ \cite[Chapter 2]{Megginson}).
\begin{enumerate}
	\item Non-negativity follows immediately from the respective property of the relative entropy. Now, assume that \(D (\mathcal{F} \| \rho) = 0 \). By definition of the infimum, this implies the existence of a sequence of free states \( \{\sigma_k\}_{k\in\mathbb{N}} \) such that \( D( \sigma_k \| \rho ) \to 0 \) as \(k \to \infty \). By Pinsker's inequality (see \cite[Theorem 5.5]{Ohya}), this implies that
	\begin{equation}
		\lim_{k \to \infty} \frac{1}{2}\norm{ \sigma_k - \rho }_1 \leq \lim_{k \to \infty} \sqrt{ \frac{1}{2} D( \sigma_k \| \rho ) } = 0 \, .
	\end{equation}
	Since \(\mathcal{F}\) is by assumption closed in the trace norm topology, we conclude that \(\rho \in \mathcal{F}\).
	\item Monotonicity is a direct consequence of the monotonicity of the relative entropy under general CPTP maps~\cite{Lindblad_1975, Uhlmann_1977}. Consider any \(\sigma \in \mathcal{F}\), we then have for any free operation \(\Lambda \in \mathcal{O} \subseteq \mathrm{CPTP}\) that
	\begin{equation}
		 D( \sigma \| \rho) \geq D( \Lambda(\sigma) \| \Lambda(\rho) ) \geq D( \mathcal{F} \| \Lambda(\rho) ) \, ,
	\end{equation}
	where the second inequality holds since \(\Lambda(\sigma) \in \mathcal{F}\) by definition. As this holds for any \(\sigma \in \mathcal{F}\), we can take the infimum on the left-hand side to obtain the claim.
	
	\item This follows by the argument given in \cite[Section 3.2.5]{Boyd} exploiting that the relative entropy is a jointly convex function (see \cite[Theorem 5.4]{Ohya}). 
	
	Let \(\rho_1, \rho_2 \in \mathcal{D}\) and \( \lambda \in [0,1]\). Due to the approximation property of the infimum, we can find for any \(\varepsilon > 0 \) two states \(\sigma_1, \sigma_2 \in \mathcal{F} \)  such that
	\begin{align}
		D(\sigma_1 \| \rho_1 ) &\leq D(\mathcal{F}\|\rho_1) + \varepsilon &\text{and} && D(\sigma_2 \| \rho_2 ) &\leq D(\mathcal{F}\|\rho_2) + \varepsilon \, .
	\end{align}
	Observe that \(\lambda \sigma_1 + (1-\lambda) \sigma_2 \in \mathcal{F} \) by assumption and thus
	\begin{align}
		D(\mathcal{F}\|\lambda \rho_1 + (1-\lambda)\rho_2) &\leq D(\lambda \sigma_1 + (1-\lambda) \sigma_2 \|\lambda \rho_1 + (1-\lambda)\rho_2) \\
		& \leq \lambda D(\sigma_1 \|\rho_1) + (1-\lambda) D(\sigma_2 \| \rho_2 ) \\
		&\leq \lambda D(\mathcal{F} \| \rho_1 ) + (1-\lambda) D(\mathcal{F} \| \rho_2) + \varepsilon \, ,
	\end{align}
	where the second inequality follows from joint convexity of the relative entropy. Taking the limit \(\varepsilon \to 0 \) proves the claim. 
	
	\item Let \( \mathcal{F} \subseteq \mathcal{D}_{AB}\) be a set of free states on a bipartite Hilbert space \(\mathcal{H}_{AB}\), that is  closed under tensor products and partial traces.
	
	The closure of \(\mathcal{F}\) under partial traces implies super-additivity. To see this, we may assume w.l.o.g.\ that \(D(\mathcal{F} \| \rho_A \otimes \omega_B) < \infty\). For any \(\varepsilon > 0 \), there exists \( \sigma_{AB} \in \mathcal{F} \) that satisfies 
	\begin{equation}
		D(\sigma_{AB} \| \rho_A \otimes \omega_B) \leq D(\mathcal{F} \| \rho_A \otimes \omega_B) + \varepsilon \, .
	\end{equation}
	It then holds that
	\begin{align}
		\varepsilon + D(\mathcal{F} \| \rho_A \otimes \omega_B ) \geq D(\sigma_{AB} \| \rho_A \otimes \omega_B) &\geq  D(\sigma_A \| \rho_A) + D(\sigma_B \| \omega_B) \\
		&\geq D(\mathcal{F} \| \rho_A ) + D(\mathcal{F} \| \omega_B ) \, ,
	\end{align}
	where the second inequality holds due the super-additivity of relative entropy (see \cite[Corollary 5.2.1]{Ohya}) and the third since the set of free states is by assumption closed under taking the partial trace. The claim follows upon taking \(\varepsilon \to 0\).
	
	Similarly, sub-additivity is a consequence of the closure of \(\mathcal{F}\) under tensor products. We may assume w.l.o.g.\ that both \( D(\mathcal{F} \| \rho_A ) \) and \(D(\mathcal{F} \| \omega_B) \) are finite. Then, for all \(\varepsilon > 0 \) there exists \(\sigma, \tau\in \mathcal{F} \) such that
		\begin{align}
			D(\sigma \| \rho_A ) &\leq  D(\mathcal{F} \| \rho_A ) + \frac{\varepsilon}{2} & \text{and} && D(\tau \| \omega_B ) &\leq  D(\mathcal{F} \| \omega_B ) + \frac{\varepsilon}{2} \, .
		\end{align}
		Since both relative entropies are finite (by assumption), the support condition implies that \(\sigma \in \mathcal{D}_A \) and \(\tau \in \mathcal{D}_B \). Consequently, \( \sigma \otimes \tau \in \mathcal{F} \) by our assumption and the claim follows from
		\begin{align}
			D( \mathcal{F} \| \rho_A \otimes \omega_B ) \leq D( \sigma \otimes \tau \| \rho_A \otimes \omega_B) = D(\mathcal{F} \| \rho_A ) + D(\mathcal{F} \| \omega_B) + \varepsilon \, ,
		\end{align}
		where in the second step we used the additivity of the relative entropy on tensor product states (see \cite[Corollary 5.2.1]{Ohya}).		
		
		\item The result is similar to \cite[Theorem 5]{Lami_2023}, but our proof works a bit differently.
	
		Let \( \{ \rho_m \}_{m \in \mathbb{N}} \) be a sequence of bipartite states \( \rho_m \) that converges to some \( \rho \in \mathcal{D}_{AB}\) in the trace-norm topology, i.e.\ \( \lim_{m \to \infty} \frac{1}{2} \norm{ \rho_m - \rho}_1 = 0\). We need to prove that
    \begin{equation}
        D ( \mathcal{F} \| \rho) \leq \liminf_{m \to \infty } D ( \mathcal{F} \| \rho_m ) \, .
    \end{equation}
    If the right-hand side is infinite, there is nothing to prove. Otherwise, we can assume w.l.o.g.\ (up to extracting sub-sequences) that
    \begin{equation}\label{Eq:Convergence_Assumption}
        \lim_{m \to \infty}  D( \mathcal{F} \| \rho_m ) = R < \infty \, . 
    \end{equation}
    
    Using the approximation property of infimum, we can then construct a sequence of free states \( \{ \sigma_m \}_{m \in \mathbb{N}} \) such that
    \begin{equation}\label{Eq:Separable_Sequence}
        D(\mathcal{F}\| \rho_m ) \leq D( \sigma_m \| \rho_m ) \leq D(\mathcal{F}\| \rho_m ) + \frac{1}{m} \leq R + 1 < \infty \, 
    \end{equation}
    holds for infinitely many \(m\).
    
    Observe that the set of sub-normalised free states \( \mathcal{F}_{\leq} := \{  \lambda \sigma :  \lambda \in [0,1] , \sigma \in \mathcal{F} \} \) is weak\(^\star\)-compact and hence sequentially weak\(^\star\)-compact.\footnote{It is a known fact that given a separable Hilbert space any weak\(^\star\)-compact subset of \(\mathcal{T}(\mathcal{H})\) is also sequentially weak\(^\star\)-compact (cf.\ \cite[Remark 1]{Lami_2023}).}  To verify this, note that we can write \( \mathcal{F}_{\leq} = \text{cone}(\mathcal{F}) \cap B_1  \), where \( B_1 := \{ X \in \mathcal{T}(\mathcal{H}) : \norm{X}_1 \leq 1 \}\) is weak\(^\star\)-compact by the Banach-Alaoglu theorem \cite[Theorem 2.6.18]{Megginson}, and \( \text{cone}(\mathcal{F}) \) is per assumption weak\(^\star\)-closed. It is then an elementary fact from topology that the intersection of a compact set with a closed set is itself compact. 
        
    Since \(\mathcal{F}_{\leq}\) is sequentially weak\(^\star\)-compact, we can extract from \( \{ \sigma_m \}_{m \in \mathbb{N}} \) a sub-sequence \( \{ \sigma_{m_k} \}_{k \in \mathbb{N} }\) such that \( \sigma_{m_k} \to \lambda \sigma_\star\) converges in the weak\(^\star\)-topology as \(k \to \infty\), for some state \( \sigma_\star \in \mathcal{F} \) and some real coefficient \(\lambda \in [0,1]\).\footnote{In the case \(\lambda=0\), it would be more accurate to say that we can choose the state to be free.} If we could show that \(\lambda=1\) we would be done, because then the claim follows from
    \begin{equation}\label{Eq:Claimed_Limit}
        D (\mathcal{F} \| \rho) \stackrel{(1)}{\leq} D(\sigma_\star \| \rho ) \stackrel{(2)}{\leq}  \liminf_{k \to \infty} D(\sigma_{m_k} \| \rho_{m_k} ) \stackrel{(3)}{\leq} \liminf_{k \to \infty} \left( D (\mathcal{F} \| \rho_{m_k}) + \frac{1}{m_k} \right) \stackrel{(4)}{=}   R \, .
    \end{equation}
    Here, (1) uses the assumption that \(\lambda=1\), (2) follows from the lower semi-continuity of the relative entropy in the product weak\(^\star\)-topology (cf.\ \cite[Lemma 4]{Lami_2023}), (3) by construction of the sequence \( \{ \sigma_m \}_{m \in \mathbb{N}} \) and (4) since the limit exists and we can thus take it along any subsequence.
    
    Now, the point of the proof is to show that indeed \(\lambda = 1\). For this, let \( \{ P_m \}_{m \in \mathbb{N}} \) be a sequence of finite-rank projectors that converges strongly to the identity on \( A \otimes B\), i.e.\ \( P_m \xrightarrow{s} \Id_{AB} \) as \(m\to \infty\). Note that this directly implies \( \lim_{m \to \infty} \tr{ P_m \rho } = 1\) for any state \(\rho \in \mathcal{D}\). 
    
    For all sufficiently large \(k \) and \( m\), we then have that 
    \begin{align}
        R + 1   &\stackrel{(1)}{\geq} D( \sigma_{m_k} \| \rho_{m_k} ) \\
                &\stackrel{(2)}{\geq} D_\text{bin}( \tr{ P_m \sigma_{m_k} } \| \tr{ P_m \rho_{m_k} } ) \\
                &\stackrel{(3)}{\geq} -1 - (1- \tr{ P_m \sigma_{m_k} }) \log ( 1 - \tr{ P_m \rho_{m_k} }) \, ,
    \end{align}     
    where (1) follows by Eq.~\eqref{Eq:Separable_Sequence}, in (2) we used the monotonicity of the relative entropy and in (3)  the elementary inequality
    \begin{equation}\label{Eq:Elementary_inequality}
        D_\text{bin}(q\|p) = - h_\text{bin}(q) - q \log p - (1-q) \log (1-p) \geq -1 - (1-q) \log (1-p) \, .
    \end{equation}
    Here, \( D_\mathrm{bin}(q\|p) \) denotes the relative entropy between the binary probability distributions \( (q,1-q)\) and \( (p,1-p)\). Re-arranging, we find that
    \begin{equation}
        1- \tr{ P_m \sigma_{m_k} } \leq \frac{ R + 2}{ - \log ( 1 - \tr{ P_m \rho_{m_k} }) } \, .
    \end{equation}
    We can now take the limit \( k \to \infty \) and -- using that \( P_m \) is finite-rank and thus compact -- the above inequality becomes
    \begin{equation}
        1 - \lambda \tr{ P_m \sigma_\star } \leq \frac{ R + 2}{ - \log ( 1 - \tr{ P_m \rho }) }
    \end{equation}
    by the weak\(^\star\)-convergence of the sub-sequence \( \{ \sigma_{m_k} \}_{k\in\mathbb{N}} \). Finally, taking \( m\to \infty\) and using the fact that \( \lim_{m \to \infty} \tr{ P_m \sigma_\star } = 1 = \lim_{m \to \infty} \tr{ P_m \rho } \) yields \( 1 - \lambda \leq 0\), which together with \( \lambda \leq 1\) shows that indeed \( \lambda = 1\).
    
    Lastly, the attainability of the minimum follows by specialising the above argument to the particular sequence \(\rho_m = \rho\) for all \(m\). In this case, \( R = D(\mathcal{F}\|\rho) \)  and Eq.~\ref{Eq:Claimed_Limit} immediately implies that \( D(\sigma_\star \| \rho ) = D(\mathcal{F}\|\rho) \) proving attainability. 
	\end{enumerate}
\end{proof}

%%%%%%%%%%%%%%%%%%%%%%%%%%%%%%%%%%%%%%%%%%%%%%%%%%%%%%%%%%%%%%%%%%%%%%%%%%%
\subsection{The Error Exponent of Entanglement Distillation and Entanglement Testing}\label{Sec:Error_Exponent_Distillation}

A practically meaningful entanglement measure should admit an operational interpretation, i.e. it should be directly linked to a
concrete task in quantum information processing. For the reverse relative entropy of entanglement, this was established in finite-dimensional systems by Lami et al.\ in~\cite{Lami_2024_2} in the context of entanglement distillation. Our goal is to extend this operational interpretation into the general infinite-dimensional setting.

Let us start with a formal description of entanglement distillation. Two parties, called Alice and Bob, are given \(n\) copies of a bipartite quantum state \(\rho_{AB}\) on a separable Hilbert space \(A \otimes B\). Their task is to convert the given state into as much pure entanglement as possible. To be precise, the goal is to maximise the number of copies \(m\) of the maximally entangled state on the fixed two-qubit system \(A_0 \otimes B_0\), i.e.\ \( \ket{\Phi}_{A_0 B_0} = \frac{1}{\sqrt{2}} \left( \ket{00} + \ket{11} \right) \). Moreover, we do not require that they achieve their goal perfectly, but allow for some small non-zero error \(\varepsilon_n\). Thus, the goal is to find a quantum channel \(\Lambda\) that achieves
\begin{align}
	 \Lambda \left( \rho_{AB}^{\otimes n} \right) &\approx_{\varepsilon_n} \Phi_{A_0 B_0}^{\otimes m} && \text{with} & \Phi_{A_0 B_0} := \ketbra{\Phi}{\Phi}_{A_0 B_0}\, .
\end{align} 

In order to get an interesting theory, not all possible quantum channels \(\Lambda\) are allowed to achieve this goal, but only a subset are deemed \emph{free operations} \(\mathcal{O}\).  Historically, entanglement distillation was studied with \(\mathcal{O}\) given by the \emph{local operations and classical communication} (LOCC) paradigm~\cite{Bennett_1996_1, Bennett_1996_2, Bennett_1996_3}. Although operationally well-defined, the set of LOCC channels has a very complicated mathematical structure~\cite{Chitambar_2014}. Therefore, we consider instead the largest physically 
consistent class of transformations. Specifically, we consider the set of \emph{non-entangling} operations, denoted \(\textsf{NE}\), which encompasses all CPTP maps that do not add entanglement to the state~\cite{Brandao_2008,Brandao_2010_1}. Formally, we define
\begin{equation}
	\textsf{NE}_{n \to m} := \bigg\{ \Lambda \in \mathrm{CPTP}(A^nB^n \to A_0^mB_0^m) : \Lambda(\sigma) \in \mathcal{S}_{A_0^m:B_0^m} \; \forall \sigma \in \mathcal{S}_{A^n:B^n} \bigg\} \, .
\end{equation}

The last piece of the theory is a figure-of-merit that allows us to compare the performance of two different distillation protocols. For this, we require that the associated error satisfies \(\varepsilon_n \to 0\) as \(n\to \infty\), i.e.\ as more and more copies are available the quality of distilled entanglement should increase, becoming perfect in the asymptotic limit. Previous works then focussed on the \emph{quantity} of the obtained entanglement, i.e.\ the asymptotic yield of the protocol was used to quantify performance. This in turn leads to the notion of \emph{distillable entanglement} as figure-of-merit (see e.g.~\cite{Rains_1999} for a formal definition). 

In this work, we instead follow~\cite{Lami_2024_2} and consider the \emph{quality} of the distilled entanglement as our performance quantifier. This means we study the distillation error exponent -- that is, we ask the question how fast the quality of the distilled entanglement improves. Specifically, we require that the error behaves as \(\varepsilon \sim 2^{-c \cdot n} \) and characterise the optimal achievable exponent \(c\). Formally, we define the \emph{distillable entanglement error exponent} for \(m \) copies via
\begin{equation}
	E^{(m)}_{d,\text{err}}(\rho_{AB}) := \sup \bigg\{ \liminf_{n \to \infty} - \frac{1}{n} \log \varepsilon_n : \varepsilon_n = 1 - F \left( \Lambda_n( \rho_{AB}^{\otimes n}) , \Phi_{A_0 B_0}^{\otimes m} \right) , \; \Lambda_n \in \textsf{NE}_{n \to m} \bigg\} 
\end{equation}
with the Uhlmann fidelity function \( F(\rho, \sigma) := \norm{\sqrt{\rho} \sqrt{\sigma} }_1^2\)~\cite{Uhlmann_1976, Josza_1994}.\footnote{Note that we could equivalently use the trace-norm distance to quantify the \(\varepsilon_n\)-closeness.} The asymptotic (zero-rate) error exponent of entanglement distillation is then given by
\begin{equation}
	E_{d,\mathrm{err}} (\rho) = \lim_{m \to \infty} E^{(m)}_{d,\mathrm{err}}(\rho) \, .
\end{equation}

The (zero-rate) error exponent of entanglement distillation is closely connected to another primitive in quantum information theory: quantum state discrimination~\cite{Hiai_1991,Ogawa_2000, Hayashi}. Specifically, Lami et.\ al.~\cite{Lami_2024_2} showed that in finite-dimensions this exponent exactly coincides with the \emph{Sanov exponent} of the composite testing problem known as \emph{entanglement testing}. Here, the goal is to determine if a given quantum state is entangled or not. 

More formally, the null hypothesis is that the unknown state is given by \(\rho_{AB}^{\otimes n}\), i.e.\ \(n \) i.i.d.\ copies of the bipartite entangled state \(\rho_{AB}\). The alternative is that the state is separable w.r.t.\ the bipartite cut \(A_n:B_n\), i.e.\ the whole set \( \mathcal{S}_{A_n:B_n} \). Crucially, while the null hypothesis is simple the alternative is composite, significantly complicating the analysis. The task is then to discriminate between these two hypotheses using a measurement on the global system. Mathematically, this can be modelled by a binary POVM \(T_n\) on the whole Hilbert space, a so-called \emph{quantum test}, where the measurement outcome corresponds to the acceptance or rejection of the null hypothesis. 
 
 As usual, we can associate two types of error with each test \(T_n\). The type-1 error occurs when we mistake the entangled state \(\rho_{AB}^{\otimes n}\) for a separable state; while the type-2 error occurs when we identify a separable state with the entangled \(\rho_{AB}^{\otimes n}\). Requiring that the type-2 error is bounded by some nonzero threshold \(\varepsilon\), we ask for the optimal \emph{asymptotic} decay of the type-1 error as \(\varepsilon \to 0\). This is termed the \emph{Sanov exponent} of entanglement testing.\footnote{Note that this is conceptually different from other versions of the quantum Sanov theorem such as~\cite{Bjelakovic_2005, Noetzel_2014,Hayashi_2025_2,Hayashi_2025_3}.}
 
 At this point, it is convenient to introduce the hypothesis-testing relative entropy~\cite{Buscemi_2010, Wang_2012} given by
\begin{equation}
    D_H^\varepsilon \left( \sigma \middle\| \rho \right) := - \log \inf \bigg\{ \tr{M \rho } : 0 \leq M \leq \Id, \tr{ (\Id - M) \sigma } \leq \varepsilon \bigg\} \, .
\end{equation}
If we interpret \( M \) as an arbitrary POVM element, we can understand the pair \( \{ M, \Id - M\}\) as the most general test we can use to discriminate between \(\rho\) and \(\sigma\). Assigning the first outcome of this measurement to the state \(\sigma\) and the second to \(\rho\), \(  D_H^\varepsilon \left( \sigma \middle\| \rho \right) \) quantifies exactly the optimal type-1 error of testing \(\rho\) against \(\sigma\) given a threshold of \(\varepsilon\) on the type-2 error. Based on this, we can formally define the hypothesis-testing relative entropy of entanglement testing via
\begin{align}
	 D_H^\varepsilon \left( \mathcal{S}_{A:B} \middle\| \rho_{AB} \right) :=& \inf_{\sigma \in \mathcal{S}_{A:B} } D_H^\varepsilon \left( \sigma_{AB} \middle\| \rho_{AB} \right) \\
	 =& - \log \sup_{\sigma \in \mathcal{S}_{A:B} } \inf \big\{ \tr{ M \rho_{AB} } : 0 \leq M \leq \Id, \; \tr{ M \sigma} \geq 1 - \varepsilon \big\} \, .
\end{align}
The Sanov exponent of entanglement testing can then formally be defined as
	\begin{equation}
        \mathrm{Sanov}(\rho_{AB} \| \mathcal{S}_{A:B}) := \lim_{\varepsilon \to 0 } \liminf_{ n \to \infty } \frac{1}{n} D_H^\varepsilon \left( \mathcal{S}_{A_n:B_n} \middle\| \rho_{AB}^{\otimes n}\right) \, .
    \end{equation}  
    
The importance of the Sanov exponent in our analysis stems from the following lemma, that generalises \cite[Lemma 1]{Lami_2024_2} to the infinite-dimensional setting and establishes a direct connection between the Sanov exponent and the (zero-rate) error exponent of entanglement distillation.

\begin{lemma}\label{Lem:Error_Exponent}
		Let \(\mathcal{H}_{AB} = \mathcal{H}_A \otimes \mathcal{H}_B\) be a bipartite separable (infinite-dimensional) Hilbert space and let \(\mathcal{D}_{AB}\) be the set of quantum states. Using the definitions introduced above, the asymptotic error exponent of entanglement distillation under non-entangling operations equals the Sanov exponent of entanglement testing. Formally, we have for any \(\rho_{AB} \in \mathcal{D}_{AB}\) that
		\begin{equation}
			E_{d, \mathrm{err}}(\rho_{AB}) = \mathrm{Sanov}(\rho_{AB}\| \mathcal{S}_{A:B}) \, .
		\end{equation}
\end{lemma}

The proof generally follows along the same lines as the finite-dimensional one for \cite[Lemma 1]{Lami_2024_2}, because the output space is always finite-dimensional. We provide a full proof below for the sake of completeness, where we additionally simplify the second part of the original argument. However, we first need to verify that a crucial rewriting of the hypothesis-testing divergence also goes through in the infinite-dimensional setting. 
\begin{lemma}
	Using the definitions introduced above, we have
	 \begin{align}
	 D_H^\varepsilon \left( \mathcal{S}_{A:B} \middle\| \rho_{AB} \right) =& - \log \min \bigg\{ \tr{ M \rho_{AB} } : 0 \leq M \leq \Id, \;  \tr{ (\Id- M) \sigma} \leq \varepsilon \; \forall \sigma \in \mathcal{S}_{A:B} \bigg\} \, .
\end{align}
\end{lemma}

\begin{proof} The general idea for the proof comes from \cite[Lemma 1]{Hayashi_2025}.

	To begin, let us introduce some notation for the two relevant sets of quantum tests:
	\begin{equation}
		\mathfrak{T}_{\varepsilon, \sigma} := \bigg\{ T \in \mathcal{B}(\mathcal{H})  :  0 \leq T \leq \Id, \; \tr{ (\Id - T) \sigma} \leq \varepsilon \bigg\}
	\end{equation}
	and
	\begin{equation}
		\mathfrak{T}_{\varepsilon, \mathcal{S} } := \bigg\{ T \in \mathcal{B}(\mathcal{H}) :  0 \leq T \leq \Id, \; \tr{ (\Id - T) \sigma} \leq \varepsilon \; \forall \sigma \in \mathcal{S} \bigg\} \, .
	\end{equation}
	With this, the claim of the lemma can be rewritten as 
	\begin{equation}
		\sup_{ \sigma \in \mathcal{S} } \min_{ T \in \mathfrak{T}_{\varepsilon, \sigma} } \tr{ T \rho } = \min_{ T \in \mathfrak{T}_{\varepsilon, \mathcal{S} } } \tr{ T \rho } \, .
	\end{equation}
	
	We first argue that both minima are indeed attained. For this, we consider the weak\(^\star\)-topology on \(\mathcal{B}(\mathcal{H})\) (induced on \(\mathcal{B}(\mathcal{H})\) by its pre-dual \(\mathcal{T}(\mathcal{H}) \), cf.\ also Remark~\ref{Rem:Weak_Star}). By the Banach-Alaoglu theorem \cite[Theorem 2.6.18]{Megginson}, the operator interval \( 0 \leq T \leq \Id \), i.e.\ the set of POVM elements, is compact in this topology. Moreover, by definition all functionals of the form \( T \mapsto \tr{T \rho} \) with \(\rho \in \mathcal{T}(\mathcal{H}) \) are continuous. Both constraints thus define a weak\(^\star\)-closed half-space. Consequently, each set of quantum tests can be written as the intersection of a compact with a closed set, and is therefore compact as well. As we minimise a continuous function w.r.t.\ to a compact set, both minima are indeed attained.  

	Now, since \( \mathfrak{T}_{\varepsilon, \mathcal{S} } \subseteq \mathfrak{T}_{\varepsilon, \sigma} \), it immediately follows by a subset argument that
	\begin{equation}\label{Eq:One_Direction}
		\sup_{ \sigma \in \mathcal{S} } \min_{ T \in \mathfrak{T}_{\varepsilon, \sigma} } \tr{ T \rho } \leq \min_{ T \in \mathfrak{T}_{\varepsilon, \mathcal{S} } } \tr{ T \rho } \, .
	\end{equation}
	
	For the reverse direction, consider any \(\sigma \in \mathcal{S} \) and note that
	\begin{equation}\label{Eq:Reverse}
		\min_{ T \in \mathfrak{T}_{\varepsilon, \sigma } } \tr{ T \rho } \leq \sup_{ \sigma \in \mathcal{S} } \min_{ T \in \mathfrak{T}_{\varepsilon, \sigma } } \tr{ T \rho } =: \delta \, .
	\end{equation}
	Note that the optimal test \(T^\star\) on the LHS in Eq.~\eqref{Eq:Reverse} satisfies \( \tr{ T^\star \rho} \leq \delta \) and \( \tr{ (\Id - T^\star) \sigma} \leq \varepsilon \). Therefore, \( \Id - T^\star \in \mathfrak{T}_{ \delta, \rho } \) and we have
	\begin{equation}
		\min_{ T \in \mathfrak{T}_{ \delta, \rho }} \tr{ T \sigma } \leq \tr{ (\Id - T^\star) \sigma} \leq \varepsilon
	\end{equation}
	As this holds for arbitrary \(\sigma \in \mathcal{S}\), we must have
	\begin{equation}
		\sup_{ \sigma \in \mathcal{S} } \min_{ T \in \mathfrak{T}_{ \delta, \rho } } \tr{ T \sigma } \leq \varepsilon \, .
	\end{equation}
	
	Now, note that \( \mathcal{S} \) is a convex subset of the space of trace class operators (endowed with the trace norm topology) and \( \mathfrak{T}_{ \delta, \rho } \) is a convex and compact subset of the space of bounded linear operators (endowed with the weak\(^\star\)-topology). Moreover, \( T \mapsto \tr{ T \sigma } \) is a linear and weak\(^\star\)-continuous function for all \(\sigma \in \mathcal{S}\) and \( \sigma \mapsto \tr{ T \sigma } \) is a linear and trace-norm-continuous function for all quantum tests \(T\).\footnote{Continuity holds in the first case by definition. In the second, this follows from Hölder's inequality stating that
	\begin{equation}
		\norm{ \tr{ T \sigma } - \tr{ T \rho } } \leq \norm{T}_\infty \cdot \norm{ \sigma - \rho\ }_1
	\end{equation} }
	Consequently, we can apply Sion's minimax theorem~\cite{Sion_1958} and interchange the optimisations to obtain
	\begin{equation}
		\min_{ T \in \mathfrak{T}_{ \delta, \rho }  } \sup_{ \sigma \in \mathcal{S} } \tr{ T \sigma } = \sup_{ \sigma \in \mathcal{S} } \min_{ T \in \mathfrak{T}_{ \delta, \rho} } \tr{ T \sigma } \leq \varepsilon \, .
	\end{equation}
	
	This implies that there exists a test \(T^\star\) such that
	\begin{align}
		\sup_{ \sigma \in \mathcal{S} } \tr{ T^\star \sigma } &\leq \varepsilon & \text{and} && \tr{ (\Id - T^\star) \rho } &\leq \delta \, .
	\end{align}
	Note that \( \Id - T^{\star} \in \mathfrak{T}_{\varepsilon, \mathcal{S} } \) and we can complete the proof with
	\begin{equation}
		\min_{ T \in \mathfrak{T}_{\varepsilon, \mathcal{S} }} \tr{T \rho } \leq \tr{ (\Id - T^{\star}) \rho } \leq \delta = \sup_{ \sigma \in \mathcal{S} } \min_{ T \in \mathfrak{T}_{\varepsilon, \sigma } } \tr{ T \rho }\, .
	\end{equation}
	Together with Eq.~\eqref{Eq:One_Direction} this establishes the claimed equality.
\end{proof}

\begin{proof}[Proof of Lemma~\ref{Lem:Error_Exponent}]	
	First, we show how to construct a feasible test for entanglement testing from any feasible distillation protocol. Consider an arbitrary non-entangling distillation protocol \(\Lambda_n \in \textsf{NE}_{n \to m}\) that achieves an error \(\varepsilon_n\). By definition, we have
	\begin{align}
		1 - \varepsilon_n = F \left( \Lambda_n(\rho_{AB}^{\otimes n} ) , \Phi_{A_0 B_0}^{\otimes m} \right) &= \tr{ \Lambda_n(\rho_{AB}^{\otimes n}) \Phi_{A_0 B_0}^{\otimes m} } = \tr{ \rho_{AB}^{\otimes n} \Lambda_n^\dagger ( \Phi_{A_0 B_0}^{\otimes m} ) } \, ,
	\end{align}
	where in the first equality we used that the second argument is pure and in the second step we introduced the adjoint \(\Lambda_n^\dagger\) of the protocol. Since \(\Lambda_n^\dagger\) is completely positive and unital as the adjoint of a CPTP map~\cite{Holevo}, we can identify the operator \(M_n := \Id - \Lambda_n^\dagger (\Phi^{\otimes m} ) \) as a valid POVM element on \(\mathcal{H}^{\otimes n}\).  Moreover, for any \(\sigma_n \in \mathcal{S}_{A_n:B_n}\), it holds that
	\begin{align}
		\tr{ M_n \sigma_n} &= 1 - \tr{ \Lambda_n(\sigma_n) \Phi_{A_0 B_0}^{\otimes m} } \geq 1 - 2^{-m} \, ,
	\end{align}
	where we observed that \(\Lambda_n(\sigma_n) \) is separable since the channel is non-entangling and then invoked a standard bound on the so-called singlet fraction in a separable state from~\cite{Horodecki_1999_2}. Note that the latter directly applies as the output space of \(\Lambda_n\) is finite-dimensional.
	
	Thus, the binary POVM \(\{M_n , \Id - M_n \}\) is a feasible test for the hypothesis-testing relative entropy of  entanglement testing with a type-2 error threshold of \(2^{-m}\). Consequently, we obtain the general bound
	\begin{align}
		\inf_{\sigma_n \in \mathcal{S}_{A_n:B_n} } D_H^{2^{-m}} \left( \sigma_n \middle\| \rho_{AB}^{\otimes n} \right) \geq - \log \tr{ M_n \rho_{AB}^{\otimes n} } \geq - \log \varepsilon_n \, .
	\end{align}
	Dividing by \(n\), taking the limit \(n \to \infty\) and finally the supremum over all non-entangling protocols \( \{ \Lambda_n \}_{n \in \mathbb{N}}\) yields
	\begin{equation}
		E_{d,\text{err}}^{(m)} (\rho_{AB}) \leq \liminf_{n \to \infty} \frac{1}{n} \inf_{\sigma_n \in \mathcal{S}_{A:B} } D_H^{2^{-m}} \left( \sigma_n \middle\| \rho_{AB}^{\otimes n} \right) = \liminf_{n \to \infty} \frac{1}{n} D_H^{2^{-m}} \left( \mathcal{S}_{A_n:B_n} \middle\| \rho_{AB}^{\otimes n} \right) \, .
	\end{equation}
	
	For the other direction, we construct a distillation protocol from any feasible test for entanglement testing. Consider any test operator \(M_n\) satisfying \(\tr{ M_n \sigma_n } \geq 1 - 2^{-m} \) for all \(\sigma_n \in \mathcal{S}_{A_n:B_n} \). With this, we construct the CPTP map \(\Lambda_n \in \mathrm{NE}_{n \to m}\) via the mapping
	\begin{equation}
		X_{A_nB_n} \mapsto \Lambda_n (X_{A_nB_n} ) = \tr{ (\Id - M_n ) X_{A_nB_n} } \Phi_{A_0 B_0}^{\otimes m} + \tr{ M_n X_{A_nB_n} } \frac{ \Id - \Phi_{A_0 B_0}^{\otimes m} }{2^{2m} -1 } \, .
	\end{equation}
	To verify that this map is non-entangling, let us first observe that the output of the map is by construction an isotropic state. It is known that for isotropic states separability is completely characterised in terms of the singlet fraction~\cite{Horodecki_1999_1}. Now, observe that the singlet fraction at the output satisfies
	\begin{align}
		F( \Lambda_n(\sigma_n) , \Phi_{A_0 B_0}^{\otimes m} ) =  1 - \tr{ M_n \sigma_n } \leq 2^{-m} 
	\end{align}
	 by assumption for any separable state \(\sigma_n\). This in turn is equivalent to the separability of \(\Lambda_n(\sigma_n)\)~\cite{Horodecki_1999_1}. As this holds for any separable input state, the map is indeed non-entangling. By construction, we then have \( F( \Lambda_n(\rho_{AB}^{\otimes n}) , \Phi_{A_0 B_0}^{\otimes m} ) = 1 - \tr{ M_n \rho_{AB}^{\otimes n}  } \). Consequently, we obtain a feasible protocol for entanglement distillation that achieves an error \(\varepsilon_n = \tr{ M_n \rho_{AB}^{\otimes n}  }\). 
	
	Now, picking the sequence of optimal tests \(M_n^\star\) that achieves the hypothesis-testing relative entropy we obtain the general bound
	\begin{equation}
		E_{d,\text{err}}^{(m)} (\rho_{AB}) \geq \liminf_{n \to \infty} - \frac{1}{n} \log \tr{ M_n^\star \rho_{AB}^{\otimes n} } = \liminf_{n \to \infty} \frac{1}{n} D_H^{2^{-m}} \left( \mathcal{S}_{A_n:B_n} \middle\| \rho_{AB}^{\otimes n} \right) \, .
	\end{equation}
	
	Finally, we observe that the function \( \varepsilon \to D_H^\varepsilon\left( \mathcal{S}_{A_n:B_n} \middle\| \rho_{AB}^{\otimes n} \right)  \) is monotone non-decreasing. Therefore, the left-sided limit \(\varepsilon \to 0\) exists and we may take it along any sequence, in particular \(\varepsilon_m = 2^{-m} \) with \(m \to \infty\).
	
	 Combining all partial results, we finally conclude that
	\begin{equation}
		E_{d, \mathrm{err}}(\rho_{AB}) = \lim_{m \to \infty} E_{d, \text{err}}^{(m)}(\rho) = \lim_{m \to \infty}  \liminf_{n \to \infty} \frac{1}{n} D_H^{2^{-m}} \left( \mathcal{S}_{A_n:B_n} \middle\| \rho_{AB}^{\otimes n} \right) = \mathrm{Sanov}(\rho_{AB} \| \mathcal{S}_{A:B} ) \, .
	\end{equation}
	\end{proof}
	
%%%%%%%%%%%%%%%%%%%%%%%%%%%%%%%%%%%%%%%%%%%%%%%%%%%%%%%%%%%%%%%%%%%%%%%%%%%
\subsection{Generalised Sanov Theorem of Entanglement Testing}\label{Sec:Sanov}

We can now give our first main technical result, which is the extension of \cite[Corollary 15]{Lami_2024_2} to infinite-dimensional systems. Specifically, we establish that the Sanov exponent of entanglement testing is exactly given by the single-letter reverse relative entropy of entanglement. By the previous analysis, this endows the reverse relative entropy of entanglement with an operational interpretation in entanglement distillation, independent of the Hilbert space dimension. 

\begin{theorem}\label{Thm:General_Sanov}
    Let \(\mathcal{H}_{AB} = \mathcal{H}_A \otimes \mathcal{H}_B\) be a bipartite separable (infinite-dimensional) Hilbert space and let \(\mathcal{D}_{AB}\) be the set of quantum states on \(\mathcal{H}_{AB}\). Using the definitions from above, for any state \(\rho_{AB} \in \mathcal{D}_{AB}\), it holds that
    \begin{align}
        \lim_{n\to\infty} \frac{1}{n} D_H^\varepsilon \left(\mathcal{S}_{A_n:B_n} \middle\| \rho_{AB}^{\otimes n}\right) = D(\mathcal{S}_{A:B} \|\rho_{AB}) \quad \forall \varepsilon \in (0,1) \, .
    \end{align}  
    
    Consequently, we have the following equalities:
		\begin{equation}
			E_{d, \mathrm{err}}(\rho_{AB}) = \mathrm{Sanov}(\rho_{AB}\| \mathcal{S}_{A:B}) = D(\mathcal{S}_{A:B} \|\rho_{AB}) \, .
		\end{equation} 
\end{theorem}

Let us begin with a small technical lemmata that we will need in our proof of the main theorem.
\begin{lemma}\label{Lem:Binary_Relative_Entropy}
    Let \( D_\mathrm{bin}(q\|p) \) denote the relative entropy between the binary probability distributions \( (q,1-q)\) and \( (p,1-p)\). Then, for all \(T \in \mathbb{R}\) we have
    \begin{equation}
        \lim_{p \to 1^-} \inf_{q \in [0,1] } \big\{ D_\emph{bin}(q\|p)  + q T \big\} = T \, .
    \end{equation}
\end{lemma}

\begin{proof}
    For all \( p \in (0,1)\), the function \( q \to  D_\text{bin}(q\|p) + q T \) is convex and differentiable on the domain \( (0,1)\). Its minimum can be found by setting the derivative to zero; it is achieved at
    \begin{equation}
        q = q(p,T) := \frac{p}{2^T (1-p) + p} \, .
    \end{equation}

    Since \( \lim_{p \to 1^-} q(p,T) = 1\), by lower semi-continuity we see that
    \begin{equation}
        \liminf_{p \to 1^-} \inf_{q \in [0,1] } \big\{  D_\text{bin}(q\|p) + q T \big\} =  \liminf_{p \to 1^-} \big\{  D_\text{bin}(q(p,T)\|p) + q(p,T) T \big\} \geq  D_\text{bin}(1\|1) + T = T \, .
    \end{equation}
    The converse bound can be obtained with the simple ansatz \( q = p\).
\end{proof}

\begin{proof}[Proof of Theorem~\ref{Thm:General_Sanov}]
    First, note that the finite-dimensional case is exactly \cite[Corollary 15]{Lami_2024_2}.
    
    The converse is an immediate consequence of the standard quantum Stein's lemma in infinite dimensions. Note that here the strong converse exponents are known as well (see~\cite{Mosonyi_2023} for more details). For any given separable state \( \sigma \in \mathcal{S}_{A:B}\), we have
    \begin{align}
        \limsup_{n \to \infty} \frac{1}{n} D_H^\varepsilon \left(\mathcal{S}_{A_n:B_n} \middle\| \rho^{\otimes n}\right) \leq \lim_{n \to \infty} \frac{1}{n} D_H^\varepsilon \left( \sigma^{\otimes n} \middle\| \rho^{\otimes n}\right) = D(\sigma \| \rho)
    \end{align}
    and the claim then follows by taking the infimum over \(\sigma \in \mathcal{S}_{A:B} \) on the right-hand side.

    The achievability of this exponent is obtained by lifting the finite-dimensional result. Let us start by considering a finite-rank tensor product projector \( P = \Pi_{A \to A'} \otimes \Pi_{B \to B'} \) and apply to every copy the LOCC channel \(  \Lambda : AB \to A'X_A : B' X_B\) defined via
    \begin{align}
        \rho \mapsto \Lambda(\rho) = P \rho P^\dagger \otimes \ketbra{00}{00}_{X_AX_B} + \tau \cdot \tr{ (\Id - P^\dagger P) \rho } \otimes \ketbra{11}{11}_{X_AX_B} \, ,
    \end{align}
    where \( \tau = \tau_{A'B'}\) is an arbitrary separable state on the finite-dimensional bipartite space \(A' \otimes B'\) where \(P\) is supported and \(X_A\) and \(X_B\) are fixed classical single-bit systems. 
    
    Observe that this channel can be implemented by the following LOCC protocol: Alice and Bob perform the local test associated with the finite-rank projector and communicate their outcomes to each other. If both tests succeeded, they keep the state and set their classical registers to \(0\). Otherwise, they discard the state, prepare the fixed separable state \(\tau\) and then set their classical registers to \(1\).

    We can then write
    \begin{align}
         \liminf_{n \to \infty} \frac{1}{n} D_H^\varepsilon \left(\mathcal{S}_{A_n:B_n} \middle\| \rho^{\otimes n}\right) & \stackrel{(1)}{\geq} \liminf_{n \to \infty} \frac{1}{n} D_H^\varepsilon \left( \mathcal{S}_{ A_n^\prime X_A :B_n^\prime X_B}  \middle\| \Lambda(\rho)^{\otimes n}\right) \\
         & \stackrel{(2)}{=} D \left( \mathcal{S}_{A'X_A:B'X_B} \middle\| \Lambda(\rho) \right) \\
         & \stackrel{(3)}{=}\inf_{\stackrel{\sigma_1, \sigma_2 \in \mathcal{S}_{A':B'}}{ q \in [0,1]} } D \bigg( q \sigma_1 \otimes \ketbra{00}{00} + (1-q) \sigma_2 \otimes \ketbra{11}{11} \bigg\| \Lambda(\rho) \bigg) \\
         & \stackrel{(4)}{=} \inf_{\stackrel{\sigma_1, \sigma_2 \in \mathcal{S}_{A':B'}}{ q \in [0,1]} }  D_\text{bin} \left( q \middle\| \tr{P^\dagger P \rho} \right) + q D\left(\sigma_1 \middle\|\rho_P\right) + (1-q) D\left(\sigma_2\middle\|\tau\right) \\
         & \stackrel{(5)}{=} \inf_{ q \in [0,1]} D_\text{bin} \left( q \middle\| \tr{P^\dagger P \rho} \right) + q D(\mathcal{S}_{A^\prime: B^\prime} \|\rho_P)  \, .
    \end{align}

    Here, (1) follows from monotonicity under LOCC maps and in (2) we applied the finite-dimensional generalised Sanov theorem from \cite[Corollary 15]{Lami_2024_2}. To see that we used the most general ansatz in (3), recall that a necessary condition to keep the relative entropy finite is that the support of the first argument has to be contained in the support of the second. Any non-zero cross term in the classical register would violate this support condition. Furthermore, the state \(\sigma_1\) is the post-measurement state when Alice and Bob both measure \(0\) in their classical register. As the global state is separable it therefore has to be separable as well, since it is not possible to create entanglement with local measurements and post-selection alone. The same reasoning then shows that \(\sigma_2\) is also separable. In (4), we expanded the previous expression and introduced the simplifying notation 
    \begin{equation}
        \rho_P := \frac{ P \rho P^\dagger}{ \tr{ P^\dagger P \rho} } \, ,
    \end{equation}
    and in (5) we simply observed that the choice \( \sigma_2 = \tau\) is optimal. 
    
    In the last line, we can now replace \(  D(\mathcal{S}_{A^\prime: B^\prime} \|\rho_P) \) w.l.o.g. by \( D(\mathcal{S}_{A: B} \|\rho_P) \). To see this, observe that the support of \(\rho_P\) is by construction contained in \(A^\prime \otimes B^\prime \). Any \(\sigma \in \mathcal{S}_{A:B}\) that has support outside of \(A^\prime \otimes B^\prime \) leads to an infinite relative entropy, as it violates the support condition. By definition as an infimum, we can thus safely expand the feasible set without changing the optimal value.

    We now take \( P = P_m\) as the \(m\)-th element of a sequence \(\{P_m\}_{m} \) of finite-rank projectors that converges strongly to the identity (note that this also implies the convergence \(\rho_{P_m} \to \rho\) in trace-norm). Since \( \lim_{m\to \infty} \frac{1}{2} \norm{ \rho_{P_m} - \rho }_1 = 0\), by the lower semi-continuity of the reverse relative entropy of entanglement (see Lemma~\ref{Lem:Reverse_REE_Properties}), we have for all \(\delta > 0\) that \( D(\mathcal{S}_{A: B}\|\rho_{P_m}) \geq D(\mathcal{S}_{A: B}\|\rho) - \delta \) for all sufficiently large \( m\). 
    
    Hence, the above argument shows that for all sufficiently large \( m\) we have 
    \begin{align}
        \liminf_{n \to \infty} \frac{1}{n} D_H^\varepsilon \left( \mathcal{S}_{A_n:B_n} \middle\| \rho^{\otimes n}\right) &\geq \inf_{ q \in [0,1]}  \bigg\{ D_\text{bin} \left( q \middle\| \tr{P_m^\dagger P_m \rho} \right) + q D(\mathcal{S}_{A: B} \|\rho_{P_m})  \bigg\} \\ 
        &\geq \inf_{ q \in [0,1]} \bigg\{ D_\text{bin} \left( q \middle\| \tr{P_m^\dagger P_m \rho} \right) + q \bigg( D(\mathcal{S}_{A: B} \|\rho) - \delta \bigg) \bigg\} \, . 
    \end{align}
    We can then take the limit \(m \to \infty\) using Lemma~\ref{Lem:Binary_Relative_Entropy} (note that \( \tr{P_m^\dagger P_m \rho} \to 1\) from below) to deduce that for all \(\delta>0\) we have
    \begin{align}
    	\liminf_{n \to \infty} \frac{1}{n} D_H^\varepsilon \left( \mathcal{S}_{A_n:B_n} \middle\| \rho^{\otimes n}\right) &\geq D(\mathcal{S}_{A:B} \|\rho) - \delta \, .
    \end{align}
    We conclude the proof by letting \(\delta \to 0\).
\end{proof}

\begin{remark} Our argument extends to general quantum resource testing provided that the finite-dimensional Sanov theorem holds. We only require that the cone generated by the set of free states \(\mathrm{cone}(\mathcal{F})\) is weak\(^\star\)-closed (which is often the case) and the set of free operations includes our choice of measure-and-prepare channel.
\end{remark}

%%%%%%%%%%%%%%%%%%%%%%%%%%%%%%%%%%%%%%%%%%%%%%%%%%%%%%%%%%%%%%%%%%%%%%%%%%%
%%%%%%%%%%%%%%%%%%%%%%%%%%%%%%%%%%%%%%%%%%%%%%%%%%%%%%%%%%%%%%%%%%%%%%%%%%%
\section{Bosonic Continuous-Variable Systems}

In what follows, we specialise our analysis to bosonic continuous-variable (CV) systems with a finite number of modes. We prove that, for Gaussian input states, the reverse relative entropy of entanglement is efficiently computable. We then use it to derive an upper bound on the error exponent of quantum communication for the class of teleportation-simulable channels. We go on to derive explicit analytical expressions of our bound for the most relevant one-mode Gaussian channels. Finally, we provide a preliminary analysis of the achievability of our bound.

%%%%%%%%%%%%%%%%%%%%%%%%%%%%%%%%%%%%%%%%%%%%%%%%%%%%%%%%%%%%%%%%%%%%%%%%%%%
\subsection{Gaussian Reverse Relative Entropy of Entanglement}\label{Sec:Gaussian_Reverse_REE}

Let us first revisit the definition of the reverse relative entropy of entanglement given in Sec.~\ref{Sec:Reverse_REE}. Note that even though the reverse relative entropy of entanglement is an operational entanglement measure, it is not efficiently computable in general due to the optimisation over separable states (see the discussion in the main text). However, suppose that Alice and Bob share a bipartite quantum state of a bosonic CV system with \( m = m_A + m_B \) modes, where Alice holds \(m_A\) modes and Bob the remaining \(m_B\). If the quantum state they share is Gaussian, the characterisation of its separability considerably simplifies. This is because the entanglement properties of bipartite Gaussian states can be conveniently translated at the level of quantum covariance matrices in terms of simple positive-semi definite constraints~\cite{Werner_2001} (see also~\cite{Lami_2018_1} for the state-of-the-art results).

In light of this simple characterisation of Gaussian separability, a common approach in the theory of entanglement measures of CV systems is to \emph{Gaussify} the measure under consideration. This then leads to an efficiently computable entanglement measure with notable examples being the entanglement of formation~\cite{Wolf_2004}, squashed entanglement~\cite{Lami_2017} and the standard relative entropy of entanglement~\cite{Chen_2005}. Naturally, we can also define a Gaussian version of the reverse relative entropy of entanglement via
\begin{equation}
	D( \mathcal{S}_G \| \rho)  := \min_{\sigma \in \mathcal{S}_G } D(\sigma \| \rho) \, ,
\end{equation}
where \( \mathcal{S}_G := \mathcal{S} \cap \mathcal{G}\) denotes the set states that are separable and Gaussian.
\begin{remark}
	Note that both the set of Gaussian states \(\mathcal{G}\) \cite[Lemma 1 in Appendix A]{Lami_2018_2} and the separable set are closed w.r.t.\ the trace norm topology; hence, \( \mathcal{S}_G \) is trace-norm closed as the intersection of two closed sets. Moreover, the cones generated by both sets are known to be weak\(^\star\)-closed (see \cite[Lemma 34]{Lami_2023} for the Gaussian case). Since \(\mathrm{cone}(\mathcal{S}_G) = \mathrm{cone}(\mathcal{S}) \cap \mathrm{cone}(\mathcal{G})\), the cone generated by \( \mathcal{S}_G\) is weak\(^\star\)-closed as well. By Lemma~\ref{Lem:Reverse_REE_Properties}, the Gaussian relative entropy of entanglement is a lower-semicontinuous entanglement monotone and the infimum is always attained. However, note also that the set of Gaussian states is \emph{not} convex in general.
\end{remark}

However, a common issue with these Gaussian measures is that they typically lose their operational interpretation and are not known to coincide with their regular counterparts, except in special cases~\cite{Giedke_2003,Akbari_2015}. However, the reverse relative entropy of entanglement is an exception to this rule. Specifically, we show below that the regular and Gaussian reverse relative entropy of entanglement coincide for Gaussian inputs. This is a particularly noteworthy property, as it is believed that the standard relative entropy of entanglement does \emph{not} have this property. 

\begin{lemma}\label{Lem:Gaussian_Reverse_REE}
    Let \(\mathcal{H}_{AB} = \mathcal{H}_A \otimes \mathcal{H}_B\) denote the Hilbert space of a bosonic \((m_A + m_B)\)-mode CV system, and let \(\mathcal{S}\) and \(\mathcal{G}\) denote the set of separable and Gaussian quantum states on \( \mathcal{H}_{AB}\), respectively. Using the definitions introduced above, for any \(\rho_G \in \mathcal{G} \), we have the equality
    \begin{equation}
        D( \mathcal{S} \| \rho_G) =  D( \mathcal{S}_G \| \rho_G) \, ,
    \end{equation}
    where \( \mathcal{S}_G := \mathcal{S} \cap \mathcal{G}\) denotes the set of states that are both separable and Gaussian.
\end{lemma}

The key technical tool in our proof is a method known as Gaussification of quantum states (see e.g.~\cite{Marian_2013}). Let \(\rho\) be an arbitrary \(m\)-mode CV quantum state with finite first and second moments. We can then associate with \(\rho\)  the unique Gaussian state that has the same first and second statistical moments. This is often termed the \emph{Gaussificiation} of \(\rho\) and will be denoted in the following by \(\rho_G\). In the following lemma, we collect the key properties of the Gaussification that we need.
\begin{lemma}\label{Lem:Gaussification}
	Let \(\mathcal{H}_{AB}\) be the bipartite Hilbert space of a bosonic \((m_A + m_B)\)-mode CV system and consider an arbitrary quantum state \(\rho\) with finite first and second moment. Its Gaussification \(\rho_G\), as defined above, satisfies the following properties:
    \begin{enumerate}
         \item \( \tr{ \rho \log \tau_G } = \tr{ \rho_G \log \tau_G } \) for any Gaussian state \(\tau_G \in \mathcal{G}(\mathcal{H})\).
         
        \item  \( H(\rho) \leq H(\rho_G) \) provided that \( H(\rho) < \infty \).

        \item If \( \rho \) is separable w.r.t.\ to the cut \(A:B\), then \( \rho_G \) is separable as well.
    
    \end{enumerate}

\end{lemma}

\begin{proof} Note that all of these results are well-known in the literature and we only provide their proof here for the sake of completeness.

\begin{enumerate} 

    \item This observation was made e.g.\ in \cite[Appendix]{Holevo_1999} (see also \cite[Theorem 1]{Marian_2013}).

    Note that by definition the Gaussification satisfies for all \( 1 \leq j, k \leq 2 m \) that
    \begin{align}
        \tr{ ( \rho - \rho_G) \hat{r}_k} &= 0 && \text{and} &\tr{ ( \rho - \rho_G) \anticom{\hat{r}_j}{\hat{r}_k}} &= 0 \, .
    \end{align}
    Moreover, we also have \( \tr{ ( \rho - \rho_G) \com{\hat{r}_j}{\hat{r}_k}} = 0 \) due to the CCR. Consequently, \( \tr{ ( \rho - \rho_G) f(\hat{\boldsymbol{q}}) } = 0 \) holds for any second-order polynomial \(f\)  in the canonical quadrature operators. 
    
    Given an arbitrary Gaussian state \( \tau_G\), the operator \( \log \tau_G \) is such a second-order polynomial. The latter follows immediately from the Gibbs exponential form given in Eq.~\ref{Eq:Gibbs_Form}.
    
    \item This is referred to as the maximum entropy principle (see e.g.\ \cite[Appendix]{Holevo_1999}). This in turn can be seen as a special case of the general extremality principle from \cite[Lemma 1]{Wolf_2006}.

    We have that
    \begin{equation}
        H(\rho_G) - H(\rho) = \tr{ \rho( \log \rho - \log \rho_G) } + \tr{ (\rho - \rho_G) \log \rho_G } \geq 0 \, ,
    \end{equation}
    where we identified the first term as the relative entropy \( D(\rho\|\rho_G) \), which is known to be non-negative, and the second term vanishes by the previous property.

    \item This was first established in \cite[Proposition 1]{Werner_2001}. We follow the proof given in \cite[Chapter 7.2]{Serafini}.

    We first show that the for an arbitrary separable state with covariance matrix \( \boldsymbol{V}_{AB} \), there always exists covariance matrices \(\boldsymbol{V}_A\) and \( \boldsymbol{V}_B \) such that \( \boldsymbol{V}_{AB} \geq \boldsymbol{V}_A \oplus \boldsymbol{V}_B \). Note that by definition any separable state can be decomposed into a convex mixture of product states with convex weights \(p_k\) as
    \begin{equation}
        \rho_{AB} = \sum_k p_k ( \rho_A^k \otimes \rho_B^k ) \, .
    \end{equation}
    Let the component product states \(\rho_A^k \otimes \rho_B^k\) have displacement vectors \(\boldsymbol{\mu}^k\) and block diagonal covariance matrices \( \boldsymbol{V}^k = \boldsymbol{V}_A^k \oplus \boldsymbol{V}_B^k\).

    By linearity of the trace, the displacement vector of \(\rho_{AB}\) has components \( \mu_i = \sum_k p_k \mu^k_i \). Similarly, the components of its covariance matrix satisfy
    \begin{equation}\label{Eq:Separable_Covariance_Matrix}
        V_{i, j} + 2 \mu_i \mu_j = \sum_k p_k ( V_{i,j}^k + 2 \mu_i^k \mu_j^k ) \, .
    \end{equation}
    The difference between \( \boldsymbol{V}\) and the block-diagonal matrix \( \sum_k p_k \boldsymbol{V}^k \) is thus given by
    \begin{equation}
        \Delta_{i,j} = 2 \left( \sum_k p_k \mu_i^k \mu_j^k  - \sum_{k,l} p_k p_l \mu_i^k \mu_j^l\right) \, .
    \end{equation}
    This defines a positive semi-definite matrix since for any \( \boldsymbol{x}\) we have
    \begin{equation}
        \boldsymbol{x}^T \boldsymbol{\Delta} \boldsymbol{x} = \sum_{j,j} x_i \Delta_{i,j} x_j = \sum_{k,l} p_k p_l (s_k - s_l)^2 \geq 0 \, ,
    \end{equation}
    where we introduced \( s_k = \sum_{i} x_i \mu_i^k \). 
    
    The above observation implies that we can write \( \boldsymbol{V}_{AB} = \boldsymbol{V}_A \oplus \boldsymbol{V}_B  + \boldsymbol{Y}\), where \(\boldsymbol{Y}\) is a positive semi-definite matrix. Hence, the Gaussian state with covariance matrix \( \boldsymbol{V}_{AB}\) may be obtained from an uncorrelated Gaussian state with covariance matrix \(\boldsymbol{V}_A \oplus \boldsymbol{V}_B\) by the action of an additive noise channel. Mathematically, this can be represented as the action of random local unitaries weighted by a Gaussian probability distribution (see e.g.\ \cite[Chapter 5.3.2]{Serafini} for more details). Crucially, the action of such a channel preserves the separability of the input state, proving the claim.
    \end{enumerate}
\end{proof}

The main technical hurdle in our proof is to show that we can restrict without loss of generality to separable states with finite second moments. Once this is established, the proof follows almost immediately by the properties of the Gaussification discussed above.
\begin{proof}[Proof of Lemma~\ref{Lem:Gaussian_Reverse_REE}]

First, observe that since \( \mathcal{S}_G \subseteq \mathcal{S} \), we have by a subset argument that
\begin{equation}\label{Eq:Gaussian_Bound}
    D(\mathcal{S}\|\rho) \leq D(\mathcal{S}_G\|\rho)
\end{equation}
for any quantum state \(\rho \in \mathcal{D}(\mathcal{H})\). The point of the proof is to show that the reverse holds provided that the state is Gaussian, i.e.\, we will show for any \(\rho_G \in \mathcal{G}\) that
\begin{equation}\label{Eq:Claim}
	D(\mathcal{S}_G\|\rho_G) \leq D(\mathcal{S}\|\rho_G) \, .
\end{equation}

We now show that we can restrict the optimisation to states with finite second moments. Note that this automatically implies the finiteness of its first moments. The idea for this part of the proof comes from \cite[Corollary 35]{Lami_2023}. 

Without loss of generality, we assume that \(D(\mathcal{S}\|\rho_G)\) is finite (since otherwise the claimed inequality holds trivially). Thus, consider an arbitrary state \( \sigma  \in \mathcal{S} \) with \( D(\sigma\|\rho_G) < \infty \). Following \cite[Chapter 3]{Serafini}, we can write the Gaussian state as \[ \rho_G =  U_G  \left( \ketbra{0}{0}_k \otimes \frac{1}{Z} \exp( - \hat{H} ) \right) U_G^\dagger \, ,\] where \( \ketbra{0}{0}_k\) is the \(k\)-mode vacuum state, \(U_G\) a Gaussian unitary operator, 
	\begin{equation}
		\hat{H} = \sum_{j=1}^{m-k} \frac{\omega_j}{2} ( \hat{x}_j^2 + \hat{p}_j^2 ) 
	\end{equation}
	 is the canonical Hamiltonian of the system with \(\omega_j > 0 \) and \(Z\) a normalisation constant. Due to the unitary invariance of the relative entropy, we have
	\begin{equation}
		D( \sigma \| \rho_G ) = D \left( U_G^\dagger \sigma U_G \middle\| \ketbra{0}{0}_k \otimes Z^{-1} \exp \left( - \hat{H} \right)  \right) \, . 
	\end{equation}
	Given that \( D(\sigma\|\rho_G) < \infty \) holds, we can write \( U_G^\dagger \sigma U_G= \ketbra{0}{0}_k \otimes \sigma^\prime \) due to the support condition. Consequently, we have established the equality \( D( \sigma \| \rho_G ) =  D \left( \sigma^\prime \middle\| Z^{-1} \exp \left( - \hat{H} \right) \right) \). 
	
	As the state \( Z^{-1} \exp \left( - \hat{H} \right) \) is faithful, we can now invoke the variational expression for the measured relative entropy from \cite[Lemma 20]{Ferrari_2023} to obtain
\begin{align}
	D \left( \sigma^\prime \middle\| Z^{-1} \exp \left( - \hat{H} \right) \right)  &\geq \sup_{ L > 0 } \left\{ \tr{ \sigma^\prime \log L } - \log \tr{ Z^{-1} \exp \left( - \hat{H} \right) L }  \right\} \\
		&\geq \frac{1}{2} \tr{ \sigma^\prime \hat{H}_n } - \log \tr{ Z^{-1} \exp \left( - \frac{\hat{H}_n}{2} \right) } \, , \label{Eq:Bound}
\end{align}
where in the last step we simply picked the feasible operator \( L =  \exp \left( \hat{H}_n /2\right) >  0\) with \( \hat{H}_n = P_n \hat{H}\), where \(P_n\) denotes the spectral projector of \(H\) corresponding to the interval \([0,n]\).  As this holds for any \(n\), we can take the limit \(n \to \infty\), resulting (after rewriting) in
\begin{equation}
	\frac{1}{2} \tr{ \sigma^\prime \hat{H} } \leq D \left( \sigma^\prime \middle\| Z^{-1} \exp \left( - \hat{H} \right) \right)  +\log \tr{ Z^{-1} \exp \left( - \frac{\hat{H}}{2} \right) } \, .
\end{equation}

Note that the RHS is finite as 
\begin{equation}
	\log \tr{ Z^{-1} \exp \left( - \frac{\hat{H}}{2} \right) } = \sum_{j=1}^{m-k} \log \left( 2 \cosh \left( \frac{\omega_j}{4}\right) \right) < \infty \, .
\end{equation}
We conclude that \( \tr{ \sigma^\prime \hat{H} } < \infty\), which in turn implies that \( \tr{ \sigma \hat{H} } < \infty\) and we must have \(\tr{ \sigma ( \hat{x}_j^2 + \hat{p_j}^2 ) } < \infty \) for all modes of the system. This is equivalent to \(\sigma\) having finite first and second moments. Note that this also implies that the state has finite entropy (see e.g.~\cite{Shirokov_2006}).

 Having established the viability of this restriction, we can now apply the method of Gaussification. Fix an arbitrary \( \sigma \in \mathcal{S}\) with finite second moments (which implies finite first moments) and its associated Gaussian state \(\sigma_G\) as defined above. Using the properties from Lemma~\ref{Lem:Gaussification}, we have that
\begin{align}
    D(\sigma\|\rho_G) = - H(\sigma) - \tr{ \sigma \log \rho_G} \geq - H(\sigma_G) - \tr{ \sigma_G \log \rho_G} = D(\sigma_G\|\rho_G) \geq D(\mathcal{S}_G\|\rho_G) \, ,
\end{align}
where the first inequality follows by Property 1 and 2 of Lemma~\ref{Lem:Gaussification} and the second by Property 3 of Lemma~\ref{Lem:Gaussification}. We then finish the proof by taking the infimum over \(\sigma \in \mathcal{S}\) on the LHS.
\end{proof}

%%%%%%%%%%%%%%%%%%%%%%%%%%%%%%%%%%%%%%%%%%%%%%%%%%%%%%%%%%%%%%%%%%%%%%%%%%%
\subsection{Efficiently Computable Entanglement Measure}\label{Sec:Convex_Program}

Lemma~\ref{Lem:Gaussian_Reverse_REE} is the key to the first main insight of this chapter, which shows that the reverse relative entropy of entanglement is an operationally meaningful entanglement measure that is also efficiently computable for Gaussian inputs. Specifically, we use it to derive a characterisation of the reverse relative entropy of entanglement for Gaussian states as a \emph{convex} program on the level of quantum covariance matrices with two simple PSD constraints. Consequently, the reverse relative entropy of entanglement is -- to the best of the authors' knowledge -- the first \emph{operational} entanglement measure that is also \emph{efficiently computable} for Gaussian states. 

\begin{proposition}[Main finding \ref{Prop:Main_2} from main text]\label{Prop:Convex_Program}
	Let \(\mathcal{H}_{AB}\) be the Hilbert space of a bosonic \((m_A + m_B)\)-mode CV system and \(\mathcal{G}\) be the set of Gaussian quantum states on \( \mathcal{H}_{AB} \).  The reverse relative entropy of entanglement of \( \rho_G \in\mathcal{G}\) with covariance matrix \( \boldsymbol{V}_\rho\) can be computed by the following convex program:
\begin{align}
     \displaystyle \min_{\boldsymbol{V}_{AB}, \boldsymbol{\gamma}_A } \quad & \frac{\tr{ \boldsymbol{V}_{AB} ( \boldsymbol{G}[\boldsymbol{V}_\rho] - \boldsymbol{G}[\boldsymbol{V}_{AB}] ) }}{ 2 \ln(2)} + \log \sqrt{  \frac{ \det \left( \boldsymbol{V}_\rho + i \boldsymbol{\Omega} \right)}{\det \left( \boldsymbol{V}_{AB} + i \boldsymbol{\Omega} \right)} } \\
     \text{s.t.} \quad & \boldsymbol{V}_{AB} \geq \boldsymbol{\gamma}_A \oplus i \boldsymbol{\Omega}_B \nonumber \quad \mathrm{and} \quad \boldsymbol{\gamma}_A \geq i \boldsymbol{\Omega}_A \nonumber
\end{align}
\end{proposition}

\begin{proof}
	Building on Lemma~\ref{Lem:Gaussian_Reverse_REE}, the proof essentially follows by combining the separability criterion from \cite[Theorem 5]{Lami_2018_1} with the characterisation of the relative entropy between two Gaussian states in terms of their statistical moments from \cite[Theorem 7]{Pirandola_2017}.

	To begin with, observe that we can shift the vector of first moments of \(\rho_G\), denoted \(\boldsymbol{\mu}_\rho \), to zero using local Gaussian unitaries~\cite{Serafini}. Moreover, note that the reverse relative entropy of entanglement is invariant under such local unitary operations. This follows directly from its monotonicity under local operations (cf.\ Lemma~\ref{Lem:Reverse_REE_Properties}) combined with the reversibility of unitary channels. In what follows, we can therefore assume \(\boldsymbol{\mu}_\rho = 0\) without loss of generality.

	Thus, let \(\rho_G := \rho_G [ 0, \boldsymbol{V}_\rho ]\) and consider an arbitrary Gaussian state \( \sigma_G := \sigma_G [\boldsymbol{\mu}_\sigma, \boldsymbol{V}_{\sigma} ]\).  According to \cite[Theorem 7]{Pirandola_2017}, we can write their relative entropy as
	\begin{equation}
    	D( \sigma_G [\boldsymbol{\mu}_\sigma, \boldsymbol{V}_{\sigma} ] \| \rho_G [ 0, \boldsymbol{V}_\rho ] ) = - \Sigma ( \boldsymbol{V}_\sigma, \boldsymbol{V}_\sigma,0) + \Sigma\left( \boldsymbol{V}_\sigma, \boldsymbol{V}_\rho, \boldsymbol{\mu}_\sigma \right) \, ,
	\end{equation}
	where the functional \(\Sigma\) is given by
	\begin{equation}
    	2 \ln(2) \cdot \Sigma\left( \boldsymbol{V}_1, \boldsymbol{V}_2, \boldsymbol{\mu}_1 - \boldsymbol{\mu}_2 \right) :=  2 \ln Z[\boldsymbol{V}_2] + \tr{ \boldsymbol{V}_1 \boldsymbol{G}[\boldsymbol{V}_2] }+ (\boldsymbol{\mu}_1 - \boldsymbol{\mu}_2)^T \boldsymbol{G}[\boldsymbol{V}_2] (\boldsymbol{\mu}_1 - \boldsymbol{\mu}_2) \, .
	\end{equation}

	The Gibbs matrix is positive-definite and thus \( (\boldsymbol{\mu}_\sigma - \boldsymbol{\mu}_\rho)^T \boldsymbol{G}(\boldsymbol{V}_2) (\boldsymbol{\mu}_\sigma - \boldsymbol{\mu}_\rho) > 0 \) for all \(\boldsymbol{\mu}_\sigma \not= \boldsymbol{\mu}_\rho = 0\). Since we are searching for a minimum, we can thus assume w.l.o.g.\ that \( \boldsymbol{\mu}_\sigma = \boldsymbol{\mu}_\rho = 0\). Consequently, the objective function is given explicitly by
	\begin{equation}
    	D(\sigma_G[0, \boldsymbol{V}_\sigma] \| \rho_G[0, \boldsymbol{V}_\rho] ) =  \frac{\tr{ \boldsymbol{V}_{\sigma} ( \boldsymbol{G}[\boldsymbol{V}_\rho] - \boldsymbol{G}[\boldsymbol{V}_{\sigma}] ) }}{2 \ln(2)} + \log \sqrt{  \frac{ \det \left( \boldsymbol{V}_\rho + i \boldsymbol{\Omega} \right)}{\det \left( \boldsymbol{V}_{\sigma} + i \boldsymbol{\Omega} \right)} } \; .
	\end{equation}
	
	Moreover, the separability of Gaussian states can be expressed in terms of their covariance matrix only. Concretely, we take the criterion from \cite[Theorem 5]{Lami_2018_1} which states that a \((m_A+m_B)\)-mode Gaussian state with covariance matrix \(\boldsymbol{V}_{AB} \) is separable if and only if there exists a \emph{bona-fide} \(m_A\)-mode quantum covariance matrix \(\boldsymbol{\gamma}_A \geq i \boldsymbol{\Omega}_A \) such that \( \boldsymbol{V}_{AB} \geq \boldsymbol{\gamma}_A \oplus i \boldsymbol{\Omega}_B \). Combining this with the above expression results in the program as stated in the proposition.

	Lastly, we prove the convexity of this program. The convexity of the feasible set is clear. Regarding the objective function, since we can assume that \(\boldsymbol{\mu}_\sigma = \boldsymbol{\mu}_\rho = 0\), we can lift the convexity of the relative entropy on the level of states can to the level of covariance matrices. That is, we will show convexity of the map
	\begin{equation}
		\boldsymbol{V_\sigma} \mapsto D( \sigma_G[0, \boldsymbol{V}_\sigma] \| \rho_G[0, \boldsymbol{V}_\rho] ) \, .
	\end{equation} 
	Consider an arbitrary collection \(\{ p_i, \boldsymbol{V}_i\}_i \) with \( \sum_i p_i  =1\). It then follows that
	\begin{align}
		\sum_i p_i D( \sigma_G[ 0, \boldsymbol{V}_i ] \| \rho_G )  &\geq D \left( \sum_i p_i  \sigma_G[ 0, \boldsymbol{V}_i ] \middle\| \rho_G \right) \\
			& \geq D \left( \left(  \sum_i p_i  \sigma_G[ 0, \boldsymbol{V}_i  ] \right)_G \middle\| \rho_G \right)  =  D \left( \sigma_G \left[ 0 , \sum_i p_i \boldsymbol{V}_i \right]  \middle\| \rho_G \right) \, ,
	\end{align}
	where in the first inequality step we used the convexity of the relative entropy and in the second we used our Gaussification argument from the proof of Lemma~\ref{Lem:Gaussian_Reverse_REE}. In the final step, we then observed that \(\sum_i p_i  \sigma_G[ 0, \boldsymbol{V}_i ] \) has zero first moments and its covariance matrix is given by the convex mixture of the individual covariance matrices (cf.\ Eq.~\ref{Eq:Separable_Covariance_Matrix} to see this).
\end{proof}

\begin{remark}\label{Rem:Efficiency}
	The claim of efficiency for this program derives from the fact that it is a finite-dimensional program with two simple PSD constraints. As such, it can straightforwardly be solved with off-the-shelf solvers based on interior-point methods, which are well-known to be efficient in praxis (cf.\ e.g.~\cite{Boyd}). For a rigorous complexity-theoretic efficiency statement, one would (after rewriting it into a standard conic program using the epigraph formulation) need to construct a self-concordant logarithmic barrier function of the resulting convex cone. There exist universal constructions of such barrier functions that would lead to provable convergence guarantees using the barrier method, i.e.\ polynomial-time solvability. However, these construction themselves are not efficient in general. A possible solution would be to adapt the specialised literature on relative entropy optimisation such as~\cite{Fawzi_2023} to this special case. 
\end{remark}

%%%%%%%%%%%%%%%%%%%%%%%%%%%%%%%%%%%%%%%%%%%%%%%%%%%%%%%%%%%%%%%%%%%%%%%%%%%
\subsection{The Error Exponent of Two-Way Assisted Quantum Communication}\label{Sec:Upper_Bound}
	
We now turn our focus to the problem of quantum communication, referred to in technical as subspace transmission. Suppose Alice holds an \(m\)-qubit system \(M\) that she wishes to transmit to Bob such that any entanglement with a reference system \(R\)  is preserved. To achieve this, they may use the noisy quantum channel \(\mathcal{N} = \mathcal{N}_{A\to B}\) \(n\) times and employ unlimited two-way classical communication and adaptive local operations, referred to as adaptive LOCC and denoted \(\mathrm{LOCC}_{\leftrightarrow} (\mathcal{N}^{\times n}) \). 
 
 Let us denote the effective channel realised by the protocol \( \Lambda_n \in \mathrm{LOCC}_{\leftrightarrow} (\mathcal{N}^{\times n}) \) as \( \mathcal{M}_{M \to \overline{M}} \), where \(\overline{M} \simeq M \) is hold by Bob with \(|M| = |\overline{M}| = 2^m\). The goal of a quantum communication protocol is to approximate the identity channel \( \mathcal{I}_{M \to \overline{M}} \). We can define the error of the protocol \(\Lambda_n\) as
 \begin{equation}
 	\varepsilon_{\mathrm{QC},\Delta}( \Lambda_n) := \sup_{ \phi_{MR} } \Delta \bigg( \mathcal{M}_{M \to \overline{M}}( \phi_{MR}) , \phi_{\overline{M}R} \bigg) \, ,
 \end{equation}
 where \(\Delta\) is the trace distance (or the purified distance), and the optimisation runs over all pure states \( \phi_{MR} \).
 
 Previous works in this setting characterised performance in terms of the \emph{quantity} of the transmitted qubits, asking at what rate faithful qubit transmission is possible asymptotically (see e.g.~\cite{Pirandola_2017,Bennett_1996_3}). This leads to the \emph{two-way assisted quantum capacity} of the channel as the relevant figure-of-merit. In this work, we instead take again the \emph{quality} of the transmission as our performance measure. Specifically, we define the error exponent of two-way assisted quantum communication for \(m\) qubits as
 \begin{equation}
 	Q_{\leftrightarrow, \Delta}^{(m)} ( \mathcal{N} ) := \sup_{ \{\Lambda_n\} } \bigg\{ \liminf_{n\to\infty} - \frac{1}{n} \log \varepsilon_{\mathrm{QC},\Delta}( \Lambda_n) : \Lambda_n \in \mathrm{LOCC}_{\leftrightarrow}(\mathcal{N}^{\times n}) \bigg\} \, .
 \end{equation}
Ultimately, our goal here is to characterise the (zero-rate) \emph{error exponent of two-way assisted quantum communication}, defined as 
 \begin{equation}
 	Q_{\leftrightarrow, \Delta} ( \mathcal{N} ) := \lim_{m \to \infty} Q_{\leftrightarrow, \Delta}^{(m)} ( \mathcal{N} ) \, .
 \end{equation}
 
 A closely related task in this assistance framework is the distribution of entanglement, where Alice and Bob aim to establish \(m\) copies of the maximally entangled state \(\Phi_{A_0B_0}\) between them. Denoting the output state of the protocol as \(\rho(\Lambda_n)\), we now want that \( \rho(\Lambda_n) \approx_\varepsilon \Phi^{\otimes m}\). The error of an entanglement distribution protocol \(\Lambda_n\) is then defined as
 \begin{equation}
 	\varepsilon_{\mathrm{ED}, \Delta}( \Lambda_n) := \Delta \left( \rho(\Lambda_n), \Phi^{\otimes m} \right) \, ,
 \end{equation}
 where typically choices for the error metric \(\Delta\) are the trace distance, infidelity (or purified distance). We can then define the error exponent of two-way assisted entanglement distribution for \(m\) ebits as
\begin{align}\label{Eq:Entanglement_Distribution_Exponent}
	E_{\leftrightarrow, \Delta}^{(m)}(\mathcal{N}) &:= \sup_{ \{\Lambda_n \} } \left\{ \liminf_{n \to \infty} -\frac{1}{n} \log \varepsilon_{\mathrm{ED}, \Delta}( \Lambda_n)    : \Lambda_n \in \mathrm{LOCC}(\mathcal{N}^{\times n}) \right\} \, .
\end{align}

\begin{remark}
	Observe that the target state is pure and thus an immediate consequence of the Fuchs–van de Graaf inequalities is that
	\begin{equation}
		E^{(m)}_{\leftrightarrow, \mathrm{P} } (\mathcal{N})  \leq E^{(m)}_{\leftrightarrow, \mathrm{Tr}} (\mathcal{N}) \leq E^{(m)}_{\leftrightarrow, \mathrm{F}} (\mathcal{N}) \quad \forall m \, .
	\end{equation}
\end{remark}

The (zero-rate) \emph{error exponent of two-way assisted entanglement distribution} is then defined as 
\begin{equation}
	E_{\leftrightarrow, \Delta} (\mathcal{N}) := \lim_{m \to \infty} E_{\leftrightarrow, \Delta}^{(m)}(\mathcal{N}) \, .
\end{equation}

Note that the distribution of entanglement is a special case of the transmission of arbitrary qubits. Moreover, once Alice and Bob can reliably establish maximally entangled qubits, they can also transmit an arbitrary qubit with the help of the quantum teleportation protocol. Below we formally show that this results in the error exponents for these two tasks being equal. We are thus justified to refer to the exponent in Eq.\ \ref{Eq:Entanglement_Distribution_Exponent} in more general terms as \emph{the} (zero-rate) error exponent of two-way assisted quantum communication, as we have done in the main text.
\begin{lemma}\label{Lem:QC_Exponent}
	For all \(m\in\mathbb{N}\), we have the identity
	\begin{equation}
		E^{(m)}_{\leftrightarrow, \Delta} ( \mathcal{N}) = Q^{(m)}_{\leftrightarrow,\Delta} (\mathcal{N}) \quad \text{where} \quad \Delta \in \{ \mathrm{P}, \mathrm{Tr} \} \, .
	\end{equation}
\end{lemma}

\begin{proof}
	The inequality \( E^{(m)}_{\leftrightarrow, \Delta} ( \mathcal{N}) \geq Q^{(m)}_{\leftrightarrow,\Delta} (\mathcal{N}) \) follows by observing that Alice can locally prepare \(m\) ebits and then use the quantum communication protocol to transmit one part of each ebit to Bob achieving an error of at most \( \varepsilon_{\mathrm{QC},\Delta}( \Lambda_n) \).
	
	For the reverse inequality, we use the quantum teleportation protocol. We can phrase this as follows: There exists an LOCC channel \( \Pi_{MA_0^m: B_0^m \to \emptyset : \overline{M} }\) such that
	\begin{equation}
		\Pi_{MA_0^m: B_0^m \to \emptyset : \overline{M} } \left( (\cdot) \otimes \Phi_{A_0B_0}^{\otimes m} \right) = \mathcal{I}_{M \to \overline{M}} (\cdot) \, .
	\end{equation} 
	Let the quantum communication protocol consists of first running the entanglement distribution protocol \(\Lambda_n\) followed by quantum teleportation. This modified protocol \(\tilde{\Lambda}\)  achieves the error 
	\begin{align}
		\varepsilon_{\mathrm{QC},\Delta}( \tilde{\Lambda}_n) &= \sup_{ \phi_{MR} } \Delta \bigg( \Pi_{MA_0^m: B_0^m \to \emptyset : \overline{M} } ( \phi_{MR} \otimes \rho(\Lambda_n) )  , \phi_{\overline{M}R} \bigg) \\
		&= \sup_{ \phi_{MR} } \Delta \bigg( \Pi_{MA_0^m: B_0^m \to \emptyset : \overline{M} } ( \phi_{MR} \otimes \rho(\Lambda_n) )  , \Pi_{MA_0^m: B_0^m \to \emptyset : \overline{M} } \left( \phi_{MR} \otimes \Phi_{A_0B_0}^{\otimes m} \right) \bigg) \\
		&\leq \sup_{ \phi_{MR} } \Delta \bigg( \phi_{MR} \otimes \rho(\Lambda_n)  , \phi_{MR} \otimes \Phi_{A_0B_0}^{\otimes m} \bigg) \\
		&= \Delta \bigg( \rho(\Lambda_n)  , \Phi_{A_0B_0}^{\otimes m} \bigg) = \varepsilon_{\mathrm{ED},\Delta}( \Lambda_n) \, ,
	\end{align}
	where in the last two steps we used the data-processing inequality and the multiplicativity of the trace distance (or purified distance). 
\end{proof}

A further simplification of this exponent is obtained by restricting to entanglement distribution protocols that produce an isotropic state at the output. Recall that the maximally entangled state of dimension \(d \in \mathbb{N} \) is defined as 
\begin{align}
	\ket{\Phi_d} &:= \frac{1}{\sqrt{d}} \sum_{i=0}^{d-1} \ket{i} \ket{i} & \text{and} && \Phi_d := \ketbra{\Phi_d}{\Phi_d}
\end{align}
A \(d\)-dimensional isotropic state then has the general form \( f \Phi_d + (1-f) \tau_d  \), where \(\tau_d \) is the orthogonal complement of \(\Phi_d\) and \(f \) is called the singlet fraction. For these class of protocols, an alternative notion of the error exponent can be defined based on the obtained singlet fraction: 
	\begin{equation}
		E^{(m)}_{\leftrightarrow, f }(\mathcal{N}) := \sup_{ \{\Lambda_n \} } \left\{ \liminf_{n \to \infty} -\frac{1}{n} \log (1-f_n)   : \rho(\Lambda_n) \stackrel{!}{=} f_n \Phi_{2^m} + (1-f_n) \tau_{2^m} , \Lambda_n \in \mathrm{LOCC}(\mathcal{N}^{\times n}) \right\}
	\end{equation}
and \( E_{\leftrightarrow, f }(\mathcal{N}) \) is obtained as before in the limit \(m \to \infty\). Crucially, as we can implement an isotropic twirling with LOCC post-processing, this restriction does not lose any generality. Since trace distance and infidelity coincide on isotropic states, we can thus further strengthen the inequalities between the error exponents from above to equalities.
\begin{lemma}\label{Lem:Exponents}
	For all \(m \in \mathbb{N}\), we have 
	\begin{equation}
		E^{(m)}_{\leftrightarrow, f} (\mathcal{N}) = E^{(m)}_{\leftrightarrow, \mathrm{Tr}} (\mathcal{N}) = E^{(m)}_{\leftrightarrow, \mathrm{F}} (\mathcal{N}) = 2 E^{(m)}_{\leftrightarrow, \mathrm{P}} (\mathcal{N}) \, .
	\end{equation}
	
	Consequently, the following chain of equalities holds:
	\begin{equation}
		E_{\leftrightarrow, f} (\mathcal{N}) = E_{\leftrightarrow, \mathrm{Tr}} (\mathcal{N}) = E_{\leftrightarrow, \mathrm{F}} (\mathcal{N}) = 2 E_{\leftrightarrow, \mathrm{P}} (\mathcal{N}) \, .
	\end{equation}
\end{lemma}

\begin{proof}

Fix \(m \in \mathbb{N} \).  Observe that \( E^{(m)}_{\leftrightarrow, \mathrm{F}} (\mathcal{N}) = 2 E^{(m)}_{\leftrightarrow, \mathrm{P}} (\mathcal{N}) \) holds by definition.

We start by showing that \(  E^{(m)}_{\leftrightarrow, f} (\mathcal{N}) \leq  E^{(m)}_{\leftrightarrow, \mathrm{Tr}} (\mathcal{N}) \). Let \( d = 2^{m}\) and consider a sequence of entanglement distribution protocols \(\Lambda_n \in \mathrm{LOCC}_\leftrightarrow(\mathcal{N}^{\times n})\) that outputs \( \rho(\Lambda_n) = f_n \Phi_d + (1-f_n) \tau_d \). Observe that
	\begin{equation}
			\frac{1}{2} \norm{ \rho(\Lambda_n) - \Phi_{d} }_{1} =  1 - f_n  \, .
		\end{equation}
	Thus, we have the generic lower bound
	\begin{equation}
		E^{(m)}_{\leftrightarrow, \mathrm{Tr}} (\mathcal{N}) \geq \liminf_{n \to \infty} - \frac{1}{n} \log( 1-f_n) 
	\end{equation}
	and taking the supremum on the RHS over all such sequences of protocols shows the claim.
	
	Observe that by the Fuchs–van de Graaf inequalities, we have \(E^{(m)}_{\leftrightarrow, \mathrm{Tr}} (\mathcal{N}) \leq E^{(m)}_{\leftrightarrow, \mathrm{F}} (\mathcal{N})\). Thus, to complete the proof we only need to show that \(  E^{(m)}_{\leftrightarrow, \mathrm{F}} (\mathcal{N}) \leq  E^{(m)}_{\leftrightarrow, f} (\mathcal{N}) \). 
	
	For this, consider a sequence of entanglement distribution protocols \(\Lambda_n \in \mathrm{LOCC}(\mathcal{N}^{\times n})\) that achieves fidelity \( F \left( \rho(\Lambda_n), \Phi_{d} \right) \). We can then straightforwardly modify this sequence of protocols such that the outputs are isotropic states by implementing an isotropic twirl at the end~\cite{Bennett_1996_3, Horodecki_1999_1}. Mathematically, the isotropic twirl of dimension \(d\) is given by
		\begin{equation}
			\mathcal{T}_d( \cdot ) = \int_{\mathrm{U}(d)} U \otimes \overline{U} ( \cdot ) ( U \otimes \overline{U} )^\dagger d \mu_H(U) \, ,
		\end{equation}
		where \( \mathrm{d}\mu_{H}(U) \) denotes the Haar measure over the unitary group of dimension \(d\). Crucially, it is known that this map can be implemented by two-way assisted LOCC using only a finite number of unitaries (resorting to a unitary 2-design). The modified protocols \( \tilde{\Lambda}_n \) now produce
		\begin{equation}
			\rho(\tilde{\Lambda}_n) = \mathcal{T}_d \left( \rho \left(\Lambda_n \right) \right) = f_n \Phi_d + (1- f_n) \tau_d \, ,
		\end{equation}
	with a singlet fraction that satisfies
	\begin{equation}
		f_n =  F( \mathcal{T}_d \left( \rho \left(\Lambda_n \right) \right), \Phi_{d} ) = F \left( \mathcal{T}_d \left( \omega \left(\Lambda_n \right) \right),  \mathcal{T}_d \left( \Phi_{d} \right) \right) \geq F( \rho \left(\Lambda_n \right), \Phi_d ) \, ,
	\end{equation}
	where in the second equality we used the invariance of the maximally entangled state under isotropic twirling and the inequality holds due to the data-processing inequality for the fidelity. Thus, we have the generic lower bound
	\begin{equation}
		E^{(m)}_{\leftrightarrow, f} (\mathcal{N}) \geq \liminf_{n \to \infty} - \frac{1}{n} \log \bigg( 1 - F \big( \rho \left(\Lambda_n \right), \Phi_d \big) \bigg) 
	\end{equation}
	and taking the supremum over all sequences of protocols finishes the proof.
\end{proof}

%%%%%%%%%%%%%%%%%%%%%%%%%%%%%%%%%%%%%%%%%%%%%%%%%%%%%%%%%%%%%%%%%%%%%%%%%%%	
\subsection{A Converse Bound for Teleportation-Simulable Channels}\label{Sec:Upper_Bound_2}

	In the following, we focus on the special class of bosonic channels that can be simulated (asymptotically) with a teleportation protocol using their Choi state. A CV quantum channel \(\mathcal{N}_{A \to B}\) is said to be \emph{teleportation-simulable} with its Choi state if for all \(r \geq 0\) there is an LOCC protocol \(\Lambda_r \in \mathrm{LOCC}(AA^\prime:B \to B) \) (with \( A^\prime \simeq A\)) such that for all states on \(\sigma\) on \(\mathcal{H}\), it holds that
	\begin{equation}
		\lim_{r \to \infty} \Lambda_r \left( \sigma_{A} \otimes \rho_\mathcal{N}(r)  \right) = \mathcal{N}(\sigma)_B \, ,
	\end{equation}
	where the convergence is understood in trace norm and \( \rho_\mathcal{N}(r) \) is the quasi-Choi state of \(\mathcal{N}\) defined in Eq.~\eqref{Def:Choi}. 
	
	Notably, all Gaussian channels, and in fact all \emph{linear bosonic channels}~\cite{Holevo_2001, Lami_2018_3}, are known to be teleportation-simulable, with \(\Lambda_r\) being based on the Braunstein-Kimble CV teleportation protocol~\cite{Braunstein_1998}. For this class of channels, we can bound the (zero-rate) error exponent of two-way assisted quantum communication in terms of the reverse relative entropy of entanglement. Our bound nicely parallels the capacity bound of Pirandola et al.\ \cite{Pirandola_2017} (see also Bennett et al.\ \cite{Bennett_1996_3}) in terms of the (regularised) relative entropy of entanglement. 
	
	\begin{proposition}[Main Finding \ref{Prop:Main_1} in main text]\label{Lem:Error_Exponent_Quantum_Communication}
		Let \(\mathcal{H}_{AB}\) be the Hilbert space of a bosonic \((m_A + m_B)\)-mode CV system and \(\mathcal{S}_{A:B}\) the set of separable states on \(\mathcal{H}_{AB}\). Using the definitions introduced above, for any teleportation-simulable channel \(\mathcal{N}_{A \to B}\), we have that
		\begin{equation}
			Q_{\leftrightarrow, \mathrm{Tr} }(\mathcal{N}_{A \to B}) \leq \liminf_{r \to \infty} D( \mathcal{S}_{A:B} \| \rho_\mathcal{N}(r) ) \, ,
		\end{equation}
		where \( \rho_\mathcal{N}(r) \coloneqq (\mathcal{N} \otimes \mathcal{I})\big(\Phi(r)^{\otimes m}\big) \) is the quasi-Choi state that is obtained by sending one half of the state \(\Phi(r)^{\otimes m}\), where $\Phi(r)$ is the two-mode squeezed vacuum state, through the channel.
	\end{proposition}
	
\begin{proof}

		Note that by Lemma \ref{Lem:QC_Exponent} and \ref{Lem:Exponents}, we can focus w.l.o.g.\ on the exponent \( E_{\leftrightarrow, f} (\mathcal{N}) \). 
		
		Let \( 0 \leq s < E_{\leftrightarrow, \mathrm{f}}(\mathcal{N}) \). By definition, there exists \(m_0\) such that for all \(m \geq m_0\), we have
		\begin{equation}
			E_{\leftrightarrow, \mathrm{f}}^{(m)}(\mathcal{N}) > s 
		\end{equation}
		
		Consider an arbitrary \(m \geq m_0\) and set \(d = 2^m \). Then, by definition of \( E_{\leftrightarrow, \mathrm{f}}^{(m)}(\mathcal{N})  \) there exists a sequence of entanglement distribution protocols \( \left\{ \Lambda_{n}^{(m)} \right\}_{n \in \mathbb{N} } \) with \( \Lambda_{n}^{(m)} \in \mathrm{LOCC}(\mathcal{N}^{\otimes n})\) protocols such that
		\begin{align}
			\rho\left( \Lambda_{n}^{(m)} \right) = f_n \Phi_d + (1-f_n) \tau_d \quad \text{with} \quad \liminf_{n \to \infty} - \frac{1}{n} \log( 1- f_n) > s
		\end{align}
		Equivalently, there exists \(N(m) \) such that for all \(n \geq N(m) \) we have \( f_n \geq 1 - 2^{-ns} \). 
		
		Now consider \(n \geq N(m) \). By the teleportation-simulability of the channel, there exists a sequence of LOCC operations \( \overline{\Lambda}_{n,r}^{(m)} \) such that
		\begin{equation}
			\lim_{r \to \infty } \overline{\Lambda}_{n,r}^{(m)} ( \rho_\mathcal{N}(r)^{\otimes n}) = \rho\left( \Lambda_{n}^{(m)} \right)\, ,		\end{equation}
		where the convergence is understood in trace-norm. 
		
		Now, on the one hand, we then have
		\begin{align}
			D \left( \mathcal{S} \middle\| \rho\left( \Lambda_{n}^{(m)} \right) \right) &\leq \liminf_{r \to \infty} D \left( \mathcal{S} \middle\| \overline{\Lambda}_{n,r}^{(m)} \left( \rho_\mathcal{N}(r)^{\otimes n}\right) \right) \\
			&\leq \liminf_{r \to \infty} D \left( \mathcal{S} \middle\| \rho_\mathcal{N}(r)^{\otimes n}  \right) \\
			&= n \liminf_{r \to \infty} D \left( \mathcal{S} \middle \|  \rho_\mathcal{N}(r)  \right) \, ,
		\end{align}
		where we first used the lower-semicontinuity of the reverse relative entropy of entanglement, then its monotonicity under LOCC processing and in the last step its additivity on tensor products (recall Lemma~\ref{Lem:Reverse_REE_Properties} for these properties).  
		
		On the other hand, we can make use of the fact that the output state of the protocol is an isotropic state of dimension \(d\). First, by the support condition, we can restrict the optimisation w.l.o.g.\ to  the set of \(d\)-dimensional separable states, denoted \(\mathcal{S}_d\). The reverse relative entropy of entanglement can then be evaluated analytically using the methods from~\cite{Vollbrecht_2001}. Recall that, by definition, an isotropic state is invariant under all unitaries of the form \( \mathcal{U} := U \otimes \overline{U} \), where \(U\) is any \(d\)-dimensional unitary \(U \in \mathrm{U}(d)\). Since such an operation preserves separability, we can invoke the unitary invariance and the joint convexity of the relative entropy to deduce that for all \(\sigma \in \mathcal{S}_d\):
		\begin{align}
		D \left( \sigma \middle\| \rho\left( \Lambda_{n}^{(m)} \right) \right)  &= \int_{\mathrm{U}(d)} D \left( \mathcal{U} \sigma \mathcal{U}^\dagger \middle\| \rho\left( \Lambda_{n}^{(m)} \right)  \right) \mathrm{d}\mu_{H}(U) \\
		& \geq D \left( \int_{\mathrm{U}(d)}  \mathcal{U} \sigma \mathcal{U}^\dagger \mathrm{d}\mu_{H}(U) \middle\| \rho\left( \Lambda_{n}^{(m)} \right)  \right) \\
		&= D \left( p \Phi_d + (1-p) \tau_d  \middle\| \rho\left( \Lambda_{n}^{(m)} \right)  \right) = p \log \left( \frac{p}{f_n} \right) + (1-p) \log \left( \frac{1-p}{1-f_n}\right) \, ,
	\end{align}
	where \(\mathrm{d}\mu_{H}(U)\) denotes the Haar measure of the unitary group \(\mathrm{U}(d)\) and \(p = \tr{ \sigma \Phi_d}\). 
	
	As we are searching for a minimum, it follows that we can restrict the optimisation w.l.o.g.\ to isotropic states. Recall that the separability of an isotropic state is equivalent to \(p \in [0, 1/d ]\)~\cite{Horodecki_1999_1}. A quick calculation confirms that for all sufficiently large \(n\), we have 
	\begin{align}
		D \left( \mathcal{S} \middle\| \rho\left( \Lambda_{n}^{(m)} \right) \right) = D_\mathrm{bin} \left( d^{-1} \middle\| f_n \right) \, .
	\end{align}
	Then, using the elementary inequality from Eq.~\eqref{Eq:Elementary_inequality} and the fact that \(f_n \geq 1 - 2^{-ns}\), we get the lower bound 
	\begin{align}
		D \left( \mathcal{S} \middle\| \rho\left( \Lambda_{n}^{(m)} \right) \right) &\geq - \log 2 + ns \left( 1 - \frac{1}{d} \right) \, .
	\end{align}
	
	Overall, we showed that for all \(m \geq m_0\) the following holds for all sufficiently large \(n \):
	\begin{equation}
		- 1 + ns \left( 1 - \frac{1}{2^m} \right) \leq n \liminf_{r \to \infty} D \left( \mathcal{S} \middle \|  \rho_\mathcal{N}(r)  \right)
	\end{equation}
	
	Dividing by \(n\), taking the limit \(n\to\infty\) and then the limit \(m\to \infty\) and finally taking the supremum over \(s\), we obtain the claim.
\end{proof}

\begin{remark}\label{Rem:Tight_Bound}
    Recalling the tightest available bounds on the distillable entanglement, there are \emph{essentially} three: (1)~the regularised relative entropy of entanglement, which corresponds to relaxing the LOCC framework to non-entangling operations \cite{Vedral_1998}; (2)~the negativity, which corresponds to the relaxation to PPT operations \cite{Vidal_2002}; and (3)~the squashed entanglement, which is conceptually different and cannot be phrased in the language of LOCC relaxations \cite{Christandl_2004, Wilde_2016}. Moreover, combining~(1) and~(2) yields the Rains bound, which is tighter than both on their own \cite{Rains_1999_2}. To be precise, \cite{Lami_2024_3} investigated dually non-entangling operations, which give a strictly better bound than~(1) that is, however, also regularised. It is an interesting open problem to investigate if there are versions of~(2) and~(3) in the error exponent setting that we study in this work. However, as is evident by the above, there are not many options that can compete with the relative entropy approach of~(1). Moreover, note that none of the above approaches~(1)--(3) can certify NPT bound entanglement.
\end{remark}

%%%%%%%%%%%%%%%%%%%%%%%%%%%%%%%%%%%%%%%%%%%%%%%%%%%%%%%%%%%%%%%%%%%%%%%%%%%
\subsection{One-Mode Gaussian Channels}\label{App:Example}

We can compute our upper bound explicitly for the most important one-mode Gaussian channels in optical communication. In the following, we use the mathematical description used in Serafini's textbook (see \cite[Section 5.3]{Serafini} for more details)\footnote{However, we make use of the more common notation \(\lambda = \cos(\theta)^2 \) and \(\eta := \cosh(s)^2 \) for the transmissivity of the attenuator and gain of the amplifier, respectively.}.
\begin{enumerate}
	\item The \emph{thermal attenuator channel} is described mathematically by 
		\begin{align}
			\boldsymbol{X} &= \sqrt{\lambda} \boldsymbol{\sigma}_0  && \text{and} & \boldsymbol{Y} &= n (1-\lambda)\boldsymbol{\sigma}_0
		\end{align}
		with \(\lambda \in [0,1]\) describing the transmissivity of the channel and \(  n \geq 1 \) a thermal noise parameter. We find that the reverse relative entropy of entanglement of the asymptotic Choi state is given by
		\begin{equation}
			\lim_{r \to \infty} D\left(\mathcal{S} \middle\| \rho_G\left[ \boldsymbol{V}_{\mathrm{Att}}(r) \right] \right) = \frac{ n_\mathrm{sep} \left(  \arcoth(n) - \arcoth(n_\mathrm{sep}) \right)  }{\ln(2)} +\log \left( \sqrt{ \frac{n^2-1}{n_\mathrm{sep}^2-1} } \right)
		\end{equation}
		for \(1 \leq n \leq n_\mathrm{sep}(\lambda)\) with \( n_\mathrm{sep}(\theta) := \frac{1+\lambda}{1-\lambda} \) and zero otherwise. Note that this diverges for \(n \to 1\), i.e.\ in the case of the pure loss channel.
	
	\item The \emph{thermal amplifier channel} is described mathematically by
		\begin{align}
			\boldsymbol{X} &= \sqrt{\eta} \boldsymbol{\sigma}_0  && \text{and} & \boldsymbol{Y} &= n (\eta-1)\boldsymbol{\sigma}_0
		\end{align}
		with \(\eta \in [1, \infty) \) describes the gain of the channel and \(  n \geq 1 \) a thermal noise parameter. We find that the reverse relative entropy of entanglement of the asymptotic Choi state is given by
	\begin{equation}
		\lim_{r \to \infty} D\left(\mathcal{S} \middle\| \rho_G\left[ \boldsymbol{V}_{\mathrm{Amp}}(r) \right] \right) = \frac{ n_\mathrm{sep} \left(  \arcoth(n) - \arcoth(n_\mathrm{sep}) \right)  }{\ln(2)} +\log \left( \sqrt{ \frac{n^2-1}{n_\mathrm{sep}^2-1} } \right)
	\end{equation}
	for \(1 \leq n \leq n_\mathrm{sep}(\eta)\) with \( n_\mathrm{sep}(\eta) := \frac{\eta+1}{\eta-1} \) and zero otherwise. Note that this diverges for \(n \to 1\), i.e.\ in the case of the quantum limited amplifier. 
	
	\item Finally, the most common example of an \emph{additive noise channel} is described by
		\begin{align}
			\boldsymbol{X} &= \boldsymbol{\sigma}_0  && \text{and} & \boldsymbol{Y} &= \mu \boldsymbol{\sigma}_0
		\end{align}
		with \( \mu \in [0,\infty)\) describing the induced noise in the channel. We find that the reverse relative entropy of entanglement of the asymptotic Choi state is given by
		\begin{equation}
			 \lim_{r \to \infty} D\left(\mathcal{S} \middle\| \rho_G\left[ \boldsymbol{V}_{\mathrm{Noise}}(r) \right] \right) = \frac{2 - \mu}{\mu \ln(2) } + \log\left( \frac{\mu}{2} \right)
		\end{equation}
		for \( 0 \leq \mu \leq 2 \) and zero otherwise. As expected this diverges in the limit \(\mu \to 0 \), i.e.\ the identity channel.
\end{enumerate}

In the following, we provide the technical derivation of these results. We start with a general prescription that works for the complete class of phase-insensitive channels to reduce the problem to a standard multivariate optimisation problem and then delve into details for our specific channels.

\subsubsection{Exploiting the Symplectic Symmetries}

In order to tackle the optimisation problem analytically, we first take advantage of the inherent symmetries in order to simplify it as far as possible. Observe that the input covariance matrix in our examples is of the general form
 \begin{equation}
	\boldsymbol{V}_\rho = \begin{pmatrix}
		x \boldsymbol{\sigma}_0 & z \boldsymbol{\sigma}_3 \\ z \boldsymbol{\sigma}_3 & y \boldsymbol{\sigma}_0
	\end{pmatrix} \, .
\end{equation}
In the literature, this is known as the \emph{normal form} of two-mode covariance matrices~\cite{Simon_2000, Duan_2000}. 

Consider the generic symplectic transformation \( \boldsymbol{S} = \boldsymbol{S}_1 \oplus \boldsymbol{S}_2 \). This acts on the covariance matrix by congruence as
\begin{equation}
	\begin{pmatrix}
		\boldsymbol{S}_1 & 0 \\ 0 & \boldsymbol{S}_2 
	\end{pmatrix} 
	\begin{pmatrix}
		x \boldsymbol{\sigma}_0 & z \boldsymbol{\sigma}_3 \\ z \boldsymbol{\sigma}_3 & y \boldsymbol{\sigma}_0
	\end{pmatrix}
	\begin{pmatrix}
		\boldsymbol{S}_1^T & 0 \\ 0 & \boldsymbol{S}_2^T
	\end{pmatrix} =
	\begin{pmatrix}
		x \boldsymbol{S}_1 \boldsymbol{S}_1^T & z \boldsymbol{S}_1 \boldsymbol{\sigma}_3 \boldsymbol{S}_2^T\\ z \boldsymbol{S}_2 \boldsymbol{\sigma}_3 \boldsymbol{S}_1^T& y \boldsymbol{S}_2 \boldsymbol{S}_2^T
	\end{pmatrix} .
\end{equation}
The matrix is an invariant of this transformation provided the following three conditions hold:
\begin{align}
	\boldsymbol{S}_1 \boldsymbol{S}_1^T &= \boldsymbol{\sigma}_0 & \boldsymbol{S}_2 \boldsymbol{S}_2^T &= \boldsymbol{\sigma}_0 && \text{and} & \boldsymbol{S}_1 \boldsymbol{\sigma}_3 \boldsymbol{S}_2^T = \boldsymbol{\sigma}_3 \implies \boldsymbol{S}_2 = \boldsymbol{\sigma}_3 \boldsymbol{S}_1 \boldsymbol{\sigma}_3 \, .
\end{align}
The first condition forces \(\boldsymbol{S}_1\) to be orthogonal. Moreover, since \( \boldsymbol{\sigma}_3 \boldsymbol{\Omega}_1 \boldsymbol{\sigma}_3 = - \boldsymbol{\Omega}_1 \) holds, it straightforwardly follows that \(\boldsymbol{S}_2 = \boldsymbol{\sigma}_3 \boldsymbol{S}_1 \boldsymbol{\sigma}_3 \) inherits orthogonality and symplecticity from \(\boldsymbol{S}_1\):
\begin{align}
	\boldsymbol{S}_2 \boldsymbol{\Omega}_1 \boldsymbol{S}_2^T = \boldsymbol{\sigma}_3 \boldsymbol{S}_1 \boldsymbol{\sigma}_3 \boldsymbol{\Omega}_1 \boldsymbol{\sigma}_3 \boldsymbol{S}_1^T \boldsymbol{\sigma}_3 = - \boldsymbol{\sigma}_3 \boldsymbol{S}_1 \boldsymbol{\Omega}_1  \boldsymbol{S}_1 ^T \boldsymbol{\sigma}_3 = - \boldsymbol{\sigma}_3 \boldsymbol{\Omega}_1 \boldsymbol{\sigma}_3= \boldsymbol{\Omega}_1 \, .
\end{align}
%\begin{align}
%	\boldsymbol{S}_2 \boldsymbol{S}_2^T = \boldsymbol{\sigma}_3 \boldsymbol{S}_1 \boldsymbol{\sigma}_0 \boldsymbol{S}_1^T \boldsymbol{\sigma}_3 = \boldsymbol{\sigma}_3 \boldsymbol{S}_1 \boldsymbol{S}_1^T \boldsymbol{\sigma}_3 = \boldsymbol{\sigma}_0 \, .
%\end{align}
Hence, a covariance matrix in normal form is invariant under all transformations of the form 
\begin{equation}
	\boldsymbol{S} = \boldsymbol{O} \oplus ( \boldsymbol{\sigma}_3 \boldsymbol{O} \boldsymbol{\sigma}_3 ) \, ,
\end{equation}
where \(\boldsymbol{O}\) is any \(2 \times 2\) orthogonal symplectic matrix. 

In the Hilbert space picture, such a transformation is represented by an infinite-dimensional local Gaussian unitary \(U_{\boldsymbol{S}} \). Using the unitary invariance of the relative entropy, we obtain
\begin{align}
	D( \sigma_G [ \boldsymbol{V} ] \| \rho_G  [ \boldsymbol{V}_{\rho} ] ) &=  D \left( U_{\boldsymbol{S}} \sigma_G [ \boldsymbol{V} ] U_{\boldsymbol{S}}^\dagger \middle\| U_{\boldsymbol{S}} \rho_G [ \boldsymbol{V}_{\rho} ] U_{\boldsymbol{S}}^\dagger \right) \\
		& = D \left( \sigma_G \left[ \boldsymbol{S} \boldsymbol{V} \boldsymbol{S}^T \right] \middle\| \rho_G \left[ \boldsymbol{S} \boldsymbol{V}_{\rho} \boldsymbol{S}^T \right] \right) \\
		&= D \left( \sigma_G \left[ \boldsymbol{S} \boldsymbol{V} \boldsymbol{S}^T  \right] \middle\| \rho_G [ \boldsymbol{V}_{\rho}  ] \right) = \int D\left( \sigma_G \left[ \boldsymbol{S} \boldsymbol{V} \boldsymbol{S}^T  \right]  \middle\| \rho_G [ \boldsymbol{V}_{\rho}  ] \right) \mathrm{d}\mu_{H} (\boldsymbol{O}_S) \, ,
\end{align}
where \( \mathrm{d}\mu_{H} (\boldsymbol{O}_S) \) denotes the Haar measure over the orthogonal subgroup of the symplectic group. Note that this exists as this is a compact subgroup of the symplectic group~\cite{Arvind_1995}.

%\ludo{If one wants to avoid continuous integrals, one can instead take a discrete average over a one-design on the same group. If I am not mistaken, in this case it should suffice to take $\boldsymbol{O}$ belonging to the group $\{\mathbf{1}_2, \boldsymbol{\sigma_x}, \boldsymbol{\sigma_z}, \boldsymbol{\Omega}\}$, or something to that effect.}

Then, using joint convexity of the relative entropy, it follows  that
\begin{align}
	D( \sigma_G [ \boldsymbol{V} ] \| \rho_G  [ \boldsymbol{V}_{\rho} ] ) &= \int D\left( \sigma_G \left[ \boldsymbol{S} \boldsymbol{V} \boldsymbol{S}^T  \right] \middle\| \rho_G [ \boldsymbol{V}_{\rho}  ] \right) \mathrm{d}\mu_{H} (\boldsymbol{O}_S) \\
	&  \geq D \left( \int \sigma_G \left[ \boldsymbol{S} \boldsymbol{V} \boldsymbol{S}^T  \right]  \mathrm{d}\mu_{H} (\boldsymbol{O}_S) 
 \middle \| \rho_G [ \boldsymbol{V}_{\rho}  ] \right) \\
 		& \geq D \left( \left( \int \sigma_G \left[ \boldsymbol{S} \boldsymbol{V} \boldsymbol{S}^T  \right]  \mathrm{d}\mu_{H} (\boldsymbol{O}_S)\right)_G
 \middle \| \rho_G [ \boldsymbol{V}_{\rho}  ] \right)  \\
		& =  D \left( \sigma_G \left[ \int \boldsymbol{S} \boldsymbol{V} \boldsymbol{S}^T  \mathrm{d}\mu_{H} (\boldsymbol{O}_S) \right] \middle\| \rho_G [ \boldsymbol{V}_{\rho}  ] \right) 
\end{align}
where we first used the convexity in the Hilbert space picture followed by our Gaussification argument from the proof of Lemma~\ref{Lem:Gaussian_Reverse_REE}. For the latter, recall that both states have zero first moments and thus the convex mixture on the state level corresponds to a convex mixture of the covariance matrices.

As we are searching for a minimum, we can thus assume w.l.o.g.\ that the optimiser is of the form
 \begin{equation}
	\int \boldsymbol{S} \boldsymbol{V} \boldsymbol{S}^T  \mathrm{d}\mu_{H} (\boldsymbol{O}_S) =
	\begin{pmatrix}
		x \boldsymbol{\sigma}_0 & z \boldsymbol{\sigma}_3 \\ z \boldsymbol{\sigma}_3 & y \boldsymbol{\sigma}_0
	\end{pmatrix}
\end{equation}
with three real-valued parameters \(x,y\) and \(z\), i.e.\ a two-mode covariance matrix in normal form. 

\subsubsection{Details on Two-Mode Normal Forms}\label{Sec:Normal_Form}

The normal form of two-mode covariance matrices is well studied in the literature (see e.g.~\cite{Serafini_2004, Pirandola_2009} and \cite[Section II.C]{Weedbrook_2012} for detailed discussions). Using standard techniques from linear algebra, one finds that the \emph{bona-fide} condition is equivalent to the following constraints on the normal form parameters:
\begin{align}
	x &\geq 1 \, , & y &\geq 1 \, , & \abs{z} &\leq  z_{\max}  :=  \sqrt{ \min \bigg\{  (x-1)(y+1) , (x+1)(y-1)  \bigg\} }=  \sqrt{ xy -1 - |x-y| } \, .
\end{align}

Moreover, by \cite[Theorem]{Simon_2000}, the Peres-Horodecki criterion provides a necessary and sufficient criterion for the separability of two-mode Gaussian states. In terms of the normal form parameters, this can be shown to be equivalent to
\begin{equation}
	z^2 \leq (x-1) (y-1) \, .
\end{equation}

We can then further simplify the optimisation by restricting to the boundary of the separable set. To see that this entails no loss of generality, let \(\boldsymbol{V}_\mathrm{opt}\) be the optimal separable covariance matrix and define \( f(t) := D( (1-t) \boldsymbol{V}_\mathrm{opt} + t \boldsymbol{V}_\rho \| \boldsymbol{V}_\rho ) \) for \(t \in [0,1]\), where \( D( \boldsymbol{V}_1 \| \boldsymbol{V}_2) :=  D( \sigma_G [ \boldsymbol{V}_1 ] \| \rho_G  [ \boldsymbol{V}_2 ] )\). The function \(f\) is convex and non-negative, achieving its minimal value of zero at \(t = 1\). Hence, it is monotone decreasing on the interval \([0,1]\). Since \(\boldsymbol{V}_\mathrm{opt}\) is by assumption the minimiser, it follows that \( (1-t) \boldsymbol{V}_\mathrm{opt} + t \boldsymbol{V}_\rho \) is not separable for all \(t > 0 \). Consequently, \( \boldsymbol{V}_\mathrm{opt}\) must lie on the boundary of the feasible set, i.e.\ the optimiser is so-called border-separable.

Furthermore, the symplectic eigenvalues of a two-mode covariance matrix are given by the handy formula:
\begin{equation}
	\nu_\pm =\sqrt{ \frac{ \Delta(V) \pm \sqrt{ \Delta(V)^2 - 4\det(V) }  }{2} }
\end{equation}
using the two global symplectic invariants
\begin{align}
	\Delta(V) &= x^2 + y^2 - 2z^2 && \text{and} & \det(V) &= (xy - z^2)^2 \, .
\end{align}
In terms of the normal form parameters, the two symplectic eigenvalues are given explicitly by
\begin{align}\label{Eq:Sym_Spectrum}
	\nu_1 &:= \frac{ \sqrt{(x+y)^2 - 4z^2} + (x-y) }{2} & \nu_2 &:= \frac{ \sqrt{(x+y)^2 - 4z^2} + (y-x) }{2} \, ,
\end{align}
where we used the identity
\begin{equation}
	A \pm \sqrt{ A^2 - B^2} = \left( \sqrt{\frac{A+B}{2}} \pm \sqrt{\frac{A-B}{2}}  \right)^2 \, .
\end{equation}

Recall that Williamson's theorem ensures the existence of a symplectic matrix \(\boldsymbol{S}\) that diagonalises the covariance matrix by congruence -- \( \boldsymbol{V} := \boldsymbol{S} \boldsymbol{W} \boldsymbol{S}^T \) -- into its Williamson form \( \boldsymbol{W} := ( \nu_1 \boldsymbol{\sigma}_0 \oplus \nu_2 \boldsymbol{\sigma}_0 )\). For a two-mode covariance matrix in normal form, the symplectic matrix  \(\boldsymbol{S}\) that achieves this transformation is given explicitly by
\begin{align}
	\boldsymbol{S} &= \begin{pmatrix}
		\omega_+ \boldsymbol{\sigma}_0 & \mathrm{sgn}(z) \omega_- \boldsymbol{\sigma}_3 \\ \mathrm{sgn}(z) \omega_- \boldsymbol{\sigma}_3 & \omega_+ \boldsymbol{\sigma}_0
	\end{pmatrix} & \text{with} &&  \omega_\pm &:= \sqrt{ \frac{ x + y \pm \sqrt{(x+y)^2-4z^2} }{ 2 \sqrt{(x+y)^2-4z^2} } } \, .
\end{align}

Using this, we can compute the Gibbs matrix explicitly. We make use of its characterisation in terms of the symplectic action of the real function \(f(x) := \arcoth(x)\). Following \cite[Section IV.B]{Spedalieri_2013}, the symplectic action of \(f\) is defined as \( f_\star(\boldsymbol{V}) := \boldsymbol{S} f(\boldsymbol{W}) \boldsymbol{S}^T\), where \(f(\boldsymbol{W})\) is the standard matrix function of the Williamson form \(\boldsymbol{W}\). Note that this is in general different from the standard matrix function, since the symplectic and spectral decomposition need not coincide. It can be shown that the Gibbs matrix is then given by (cf.\ \cite[Appendix A]{Banchi_2015}):
\begin{equation}
	\boldsymbol{G}[\boldsymbol{V}] = - \boldsymbol{\Omega} \arcoth_\star ( \boldsymbol{V} ) \boldsymbol{\Omega} \, .
\end{equation}
Using the notation from above, the Gibbs matrix of a normal form covariance matrix is then also in normal form, 
\begin{align}
	\boldsymbol{G}[\boldsymbol{V}] &= \begin{pmatrix}
		\alpha \boldsymbol{\sigma}_0 & \gamma \boldsymbol{\sigma}_3  \\ \gamma \boldsymbol{\sigma}_3  & \beta \boldsymbol{\sigma}_0
	\end{pmatrix} \, ,
\end{align}
with the normal form parameters given by
\begin{align}\label{Eq:Gibbs_Normal}
	\alpha &= \omega_+^2 \arcoth(\nu_1) + \omega_-^2 \arcoth(\nu_2) \\
	\beta &= \omega_-^2 \arcoth(\nu_1) + \omega_+^2 \arcoth(\nu_2) \\
	\gamma &= - \mathrm{sgn}(z) \omega_- \omega_+ ( \arcoth(\nu_1) + \arcoth(\nu_2)) \, .
\end{align}

\subsubsection{Setting Up the Optimisation Problem}

We have now everything in place in order to express the objective function solely in terms of the three normal form parameters. We start from the following equivalent characterisations of the relative entropy
\begin{align}
	D( \sigma_G[\boldsymbol{V}] \| \rho_G[\boldsymbol{V}_\rho]) &= \frac{ \tr{ \boldsymbol{V} (\boldsymbol{G}[\boldsymbol{V}_\rho ] - \boldsymbol{G}[\boldsymbol{V} ])  }}{ 2 \ln 2 } + \log \frac{ Z \left[ \boldsymbol{V}_\rho \right]}{Z \left[ \boldsymbol{V} \right]} \\ 
		& =  \frac{ \tr{ \boldsymbol{V} \boldsymbol{G}[\boldsymbol{V}_\rho ]  }}{ 2 \ln 2 } - H( \sigma_G[\boldsymbol{V}] ) + \log Z \left[ \boldsymbol{V}_\rho \right] \label{Eq:Objective_2}
\end{align} 

The normalisation constant \( Z[ \boldsymbol{V}]\) can be expressed solely in terms of the symplectic spectrum. Namely, we can write it as 
\begin{equation}
	\log Z \left[ \boldsymbol{V} \right] = \log \sqrt{ \det \left(\frac{1}{2} (\boldsymbol{V} + i \boldsymbol{\Omega}_2) \right) } = \frac{1}{2} \bigg( \log (\nu_1^2 - 1) + \log (\nu_2^2 - 1) \bigg) - 2
\end{equation}
Moreover, a similar handy expression in terms of the symplectic spectrum is available for the von-Neumann entropy. Following \cite[Proposition 1]{Serafini_2004} (see also~\cite{Holevo_1999}), we can write it as
\begin{align}
	H(\sigma_G[\boldsymbol{V}]) &= g(\nu_1) + g(\nu_2) && \text{with} & g(x) &:= \left( \frac{x +1}{2} \right) \log  \left( \frac{x +1}{2} \right) -  \left( \frac{x -1}{2} \right) \log  \left( \frac{x -1}{2} \right) \, .
\end{align}
Using Eq.~\eqref{Eq:Sym_Spectrum}, this is then given  implicitly in terms of the normal form coefficients. Lastly, since both covariance matrices are given in normal form, we can write the overlap term \(  \tr{ \boldsymbol{V}_1 G[ \boldsymbol{V}_2]  } \) solely in terms of the respective normal form coefficients:
\begin{align}
	\tr{ \boldsymbol{V}_1 G[ \boldsymbol{V}_2]  } = &2 x_1 \alpha_2 + 2y_1 \beta_2 + 4 z_1 \gamma_2 \, ,
\end{align}
where \((\alpha_2, \beta_2, \gamma_2)\) denote the normal form coefficients of the Gibbs matrix \( G[ \boldsymbol{V}_2] \). 

Hence, starting from Eq.~\eqref{Eq:Objective_2} the objective function becomes
\begin{align}\label{Eq:Objective_Normal_Form}
	F(x,y,z) &:= \frac{ x \alpha_{\rho} + y \beta_\rho + 2z \gamma_\rho}{\ln 2} - g(\nu_1(x,y,z)) - g(\nu_2(x,y,z)) + \log Z \left[ \boldsymbol{V}_\rho \right] \, ,
\end{align}
where \((x,y,z)\)  have to satisfy the constraints discussed in the previous sections. The problem is now set up in principle as a multivariate optimisation problem that can be solved with the methods from standard calculus. However, in general this leads to intractable transcendental equations, thus we will resort in the following to an asymptotic analysis in the large squeezing regime to obtain closed form solutions. 

\subsubsection{Analysis for Thermal Attenuator Channel}\label{Sec:Thermal_Attenuator}

Let us begin with the thermal attenuator channel. The quasi-Choi state of this channel is a Gaussian state with covariance matrix
\begin{equation}\label{Eq:Covariance_TA}
	\boldsymbol{V}_\mathrm{Att}( \lambda, n, r) =\begin{pmatrix}
		( \cosh(2r) \lambda + n (1-\lambda) ) \boldsymbol{\sigma}_0  & \sinh(2r) \sqrt{\lambda} \boldsymbol{\sigma}_3 \\ \sinh(2r) \sqrt{\lambda}   \boldsymbol{\sigma}_3   & \cosh(2r) \boldsymbol{\sigma}_0
	\end{pmatrix} \, .
\end{equation}
We find that it is entangled provided that
\begin{equation}
	n \leq n_\mathrm{sep}(\lambda) := \frac{1+\lambda}{1-\lambda} 
\end{equation}
independent of the squeezing parameter \(r\). Thus, we assume \(1 \leq n \leq n_\mathrm{sep}(\lambda) \) in the following.

As we are ultimately interested in the limit \(r \to \infty\), we focus on the large squeezing regime to make further analytical progress. For \(\lambda \in (0,1)\), we find that the symplectic spectrum of the quasi-Choi states in Eq.~\eqref{Eq:Covariance_TA} satisfies
\begin{align}
	\nu_1 ( \lambda, n, r) &= n + \mathcal{O} \left(\frac{1}{\cosh(2r)}\right)  \\
	\nu_2 ( \lambda, n, r) &= \cosh(2r) (1-\lambda) + n \lambda + \mathcal{O} \left(\frac{1}{\cosh(2r)}\right) 
\end{align}
%In the remaining cases (\(\theta =  \{0, \pi, 2\pi\} \)), the symplectic eigenvalues simplify to \(\nu_{1} = \nu_{2} = 1\). 
With this, we can then derive the associated Gibbs matrix. We find for \(\lambda \in (0,1)\) that
\begin{equation}
	\boldsymbol{G}[ \boldsymbol{V}_\mathrm{Att}( \lambda, n, r) ] = \frac{\arcoth(n)}{ (1-\lambda) } \begin{pmatrix}
		\lambda \boldsymbol{\sigma}_0 & - \sqrt{\lambda} \boldsymbol{\sigma}_3 \\ -  \sqrt{\lambda} \boldsymbol{\sigma}_3 & \boldsymbol{\sigma}_0
	\end{pmatrix} + \mathcal{O}\left( \frac{1}{\cosh(2r)} \right) \, .
\end{equation}

Then, focusing on the first term in Eq.~\eqref{Eq:Objective_Normal_Form}, we can rewrite it (ignoring the normalisation) as 
\begin{equation}
	\frac{x+y}{2} (\alpha_\mathrm{Att} + \beta_\mathrm{Att} ) + \frac{x-y}{2}  (\alpha_\mathrm{Att} - \beta_\mathrm{Att} ) + 2z \gamma_\mathrm{Att} \, .
\end{equation}
From this, we are able to deduce that \( x \geq y \geq 1 \) may be assumed and we can replace \(z \to \abs{z}\) in the above w.l.o.g.\ as we are searching for a minimum. The former follows from the fact that \(\alpha_{\mathrm{Att}} \leq \beta_{\mathrm{Att}} \) asymptotically and that the entropy term is invariant under swapping \(x\) and \(y\). Similarly, the latter is due to the optimal \(z\) having the reversed sign of \(\gamma_{\mathrm{Att}}\) (note that the symplectic eigenvalues are invariant under swapping the sign of \(z\)).

We can then explicitly re-parametrise the problem in terms of the two symplectic eigenvalues \( (\nu_1, \nu_2 )\) of the optimisation variable, using that \(\nu_1 \geq \nu_2 \geq 1\) must hold due to \(x \geq y \). We explicitly have
\begin{align}
	x &= \frac{1 + \nu_1 \nu_2 + \nu_1 - \nu_2}{2} & y &= \frac{1 + \nu_1 \nu_2 + \nu_2 - \nu_1}{2}
\end{align}
and
\begin{equation}
	\abs{z} =   \frac{ \sqrt{ (1+\nu_1 \nu_2)^2 - (\nu_1 + \nu_2)^2 }}{2} \, .
\end{equation}
The first-order condition for the minimum then requires
\begin{equation}
	\frac{\partial F}{\partial \nu_1} =  \alpha_\rho \frac{\nu_2 + 1}{2} + \beta_\rho \frac{\nu_2 -1}{2} + \gamma_\rho \nu_1 \sqrt{ \frac{ \nu_2^2 -1 }{ \nu_1^2 -1 } }  - \arcoth(\nu_1) = 0
\end{equation}
and 
\begin{equation}
	\frac{\partial F}{\partial \nu_2} = \alpha_\rho \frac{\nu_1 - 1}{2} + \beta_\rho \frac{\nu_1 + 1}{2} + \gamma_\rho \nu_2 \sqrt{ \frac{ \nu_1^2 -1 }{ \nu_2^2 -1 } }  - \arcoth(\nu_2) = 0 \, .
\end{equation}
However, this is a transcendental system of equations and cannot be solved analytically for all parameters. Observe that in the large squeezing regime, we have
\begin{equation}
	\log Z_\mathrm{Att} = \log \left( (1-\lambda) \cosh(2r) \right) + \frac{1}{2}\log( n^2-1)  -2 + \mathcal{O}\left( \frac{1}{\cosh(2r)} \right) \, ,
\end{equation}
which grows without bounds as \(r \to \infty\). As we are looking for a minimum, we need that \( \nu_1 \nu_2 = \mathcal{O}(\cosh(2r) )\) must hold to precisely cancel this growth and keep the relative entropy finite. Making the ansatz \( \nu_1 = A \cosh(2r) + B \) and \(\nu_2 = D \) with \(r \gg 1\), we have
\begin{equation}
	\frac{\partial F}{\partial \nu_1} = \frac{\alpha_\rho + \beta_\rho}{2} D + \frac{\alpha_\rho - \beta_\rho}{2} + \gamma_\rho \sqrt{D^2 -1 } + \mathcal{O} \left(\frac{1}{\cosh(2r)} \right) \, ,
\end{equation}
which leads to the solution
\begin{equation}
	\nu_2 = \frac{1 + \lambda}{1- \lambda} + \mathcal{O} \left( \frac{1}{\cosh(2r)}\right) = n_\mathrm{sep}(\lambda) + \mathcal{O} \left( \frac{1}{\cosh(2r)}\right) \, .
\end{equation}
Moreover, we have
\begin{equation}
	\frac{\partial F}{\partial \nu_2} = A \left( \frac{\alpha_\rho + \beta_\rho}{2} + \gamma_\rho \frac{D}{\sqrt{D^2-1}}\right) \cosh(2r) + \mathcal{O}(1)
\end{equation}
Plugging in the optimal \(D\), the bracket evaluates to zero; hence, \(A\) is not fixed by this and a careful analysis of the next-order reveals the same happens there. Consequently, we have to resort to a numerical solution of the problem given in Eq.~\ref{Eq:Objective_Normal_Form} to determine that \( \nu_1 = \cosh(2r) \sin(\theta)^2 + \mathcal{O}(1) \) holds in the large squeezing regime. Using this, we can then determine that
\begin{equation}
	F_\mathrm{opt} = \frac{ n_\mathrm{sep} \left( \arcoth(n) - \arcoth(n_\mathrm{sep})  \right) }{\log(2)} + \frac{1}{2} \log\left( \frac{n^2-1}{n_\mathrm{sep}^2-1} \right) + \mathcal{O}\left(\frac{1}{\cosh(2r)} \right) \, .
\end{equation}

\subsubsection{Analysis of Thermal Amplifier Channel}\label{Sec:Thermal_Amplifier}

The covariance matrix for the quasi-Choi of the thermal amplifier is given by 
\begin{equation}
	\boldsymbol{V}_{\mathrm{Amp}}(\eta, n, r) =\begin{pmatrix}
		( \cosh(2r) \eta + n (\eta-1) ) \boldsymbol{\sigma}_0  & \sinh(2r) \sqrt{\eta}\boldsymbol{\sigma}_3 \\ \sinh(2r) \sqrt{\eta} \boldsymbol{\sigma}_3   & \cosh(2r) \boldsymbol{\sigma}_0
	\end{pmatrix} \, .
\end{equation}
with gain parameter \( \eta \in [1, \infty)\) and thermal noise parameter \(n \geq 1\). We find that it is entangled if 
\begin{equation}
	n \leq \frac{\eta +1}{\eta-1} := n_\mathrm{sep}(\eta) 
\end{equation}
independent of the squeezing parameter. For the symplectic eigenvalues, we find for \( \eta>1 \) that
\begin{align}
	\nu_1 ( \eta, n, r) &= \cosh(2r) (\eta-1) + n \eta + \mathcal{O} \left(\frac{1}{\cosh(2r)}\right) \\
	\nu_2 ( \eta, n, r) &= n + \mathcal{O} \left(\frac{1}{\cosh(2r)}\right)
\end{align}
and the associated Gibbs matrix satisfies
\begin{equation}
	\boldsymbol{G}[ \boldsymbol{V}_\mathrm{Amp}( \eta, n, r) ] = \frac{\arcoth(n_\mathrm{th})}{(\eta-1) } \begin{pmatrix}
		 \boldsymbol{\sigma}_0 & - \sqrt{\eta} \boldsymbol{\sigma}_3 \\ -  \sqrt{\eta} \boldsymbol{\sigma}_3 & \eta\boldsymbol{\sigma}_0
	\end{pmatrix} + \mathcal{O}\left( \frac{1}{\cosh(2r)} \right) \, .
\end{equation}
Thus, compared to the previous example of the attenuator, the roles of \(x\) and \(y\) are reversed here. Re-running the above analysis with this consideration, the optimal value is given by
\begin{equation}
	F_\mathrm{opt} = \frac{ n_\mathrm{sep} \left( \arcoth(n) - \arcoth(n_\mathrm{sep})  \right) }{\ln(2)} + \frac{1}{2} \log\left( \frac{n^2-1}{n_\mathrm{sep}^2-1} \right) + \mathcal{O}\left(\frac{1}{\cosh(2r)} \right) \, .
\end{equation}
using the re-defined value \(n_\mathrm{sep}(\eta)\). Notably, this diverges in the case of \(n\to 1 \), which corresponds to the quantum limited amplifier.

\subsubsection{Analysis of Additive Noise Channel}\label{Sec:Noise}

The quasi-Choi of the additive Gaussian noise channel has covariance matrix given by 
\begin{equation}
	\boldsymbol{V}_{\mathrm{Noise}}(\mu, r) =\begin{pmatrix}
		( \cosh(2r) + \mu ) \boldsymbol{\sigma}_0  & \sinh(2r) \boldsymbol{\sigma}_3 \\ \sinh(2r) \boldsymbol{\sigma}_3   & \cosh(2r) \boldsymbol{\sigma}_0
	\end{pmatrix} \, .
\end{equation}
with noise parameter \(\mu \in [0, \infty)\) and we find that it is entangled if \( \mu \in [0,2)\). The symplectic eigenvalues for \(\mu > 0 \) are given by
\begin{align}
	\nu_1 ( \mu, r) &= \sqrt{\mu \cosh(2r) } + \frac{\mu}{2} + \mathcal{O} \left(\frac{1}{\sqrt{\cosh(2r)}}\right) \\
	\nu_2 ( \mu, r) &= \sqrt{\mu \cosh(2r) } - \frac{\mu}{2} + \mathcal{O} \left(\frac{1}{\sqrt{\cosh(2r)}}\right)
\end{align}
and the Gibbs matrix is given by
\begin{equation}
	\boldsymbol{G}[ \boldsymbol{V}_\mathrm{Noise}( \mu, r) ] = \frac{1}{\mu} \begin{pmatrix}
		 \boldsymbol{\sigma}_0 & -\boldsymbol{\sigma}_3 \\ -\boldsymbol{\sigma}_3 & \boldsymbol{\sigma}_0
	\end{pmatrix} + \mathcal{O}\left( \frac{1}{\sqrt{\cosh(2r)}} \right) \, .
\end{equation}
The solution can thus be assumed symmetric in \(x\) and \(y\) and we find with a similar analysis as before that 
\begin{equation}
	F_\mathrm{opt} = \frac{2 - \mu}{\mu \ln(2) } + \log\left( \frac{\mu}{2} \right) + \mathcal{O}\left(\frac{1}{\sqrt{\cosh(2r)}} \right) \, .
\end{equation}

%%%%%%%%%%%%%%%%%%%%%%%%%%%%%%%%%%%%%%%%%%%%%%%%%%%%%%%%%%%%%%%%%%%%%%%%%%%
\subsection{Achievability Bound for Gaussian Channels}\label{App:Achievability}

Lastly, we provide a preliminary analysis of the achievability of our converse bound. For this, we leverage a recently developed random-coding-based achievability bound for the error exponent of LOCC-assisted channel coding in purified distance from \cite{Berta_2026}.

\begin{proposition}
	Let \(\mathcal{N}_{A \to B}\) be a \(m\)-mode Gaussian channel. With definitions as above, we then have that
	\begin{equation}
		Q_{\leftrightarrow, \mathrm{Tr} }(\mathcal{N}_{A \to B}) \geq  \limsup_{r \to \infty} \left\{ - \log \tr{ \left( \mathrm{Tr}_{\overline{A}}\left[ \sqrt{ \rho_{\mathcal{N}}(r) } \right] \right)^2 } \right\} \, ,
	\end{equation} 
	where \( \rho_\mathcal{N}(r) \coloneqq (\mathcal{N} \otimes \mathcal{I})\big(\Phi(r)^{\otimes m}\big) \) is the quasi-Choi state that is obtained by sending one half of the state \(\Phi(r)^{\otimes m}\), where $\Phi(r)$ is the two-mode squeezed vacuum state, through the channel.
\end{proposition}

\begin{proof}
	
	We start with a proposition taken from \cite[Theorem 14]{Berta_2026} that gives an achievability result for the error exponent of LOCC-assisted channel coding in purified distance. They showed that for any finite-dimensional quantum channel \(\mathcal{N}_{A \to B} \) and transmission rate \(R > 0 \), we have
	\begin{equation}
		Q_{\mathrm{LOCC},\mathrm{P}} (\mathcal{N}, R) \geq \frac{1}{2} \sup_{s \in (0,1)} \left\{ s \left( \lim_{m \to \infty} \frac{1}{m}  I^{c}_{\frac{1}{1+s}} \left( \mathcal{N}_{A \to B}^{\otimes m} \right) - R \right) \right\} \, ,
	\end{equation}
	where the subscript LOCC denotes a restriction to non-adaptive protocols and \[I^{c}_{\alpha}(\mathcal{N}_{A \to B}) := \sup_{\phi_{\overline{A}A}} I_\alpha \left( \overline{A} \middle\rangle B \right) = \sup_{\phi_{\overline{A}A}} \inf_{\sigma_B} \overline{D}_\alpha \left( \mathcal{N}(\phi_{\overline{A}A}) \| \mathbb{1}_{\overline{A}} \otimes \sigma_B \right) \] is the Petz R\'enyi coherent information of the channel with \( \overline{A} \simeq A \). Recall that the Petz R\'enyi-divergence is defined as \( \overline{D}_\alpha( \rho \| \sigma ) := \frac{1}{\alpha-1} \log \tr{ \rho^{\alpha} \sigma^{1-\alpha}}  \) \cite{Petz_1986}.
	
	To begin, let us note that 
	\begin{equation}
		Q_{\leftrightarrow, \mathrm{Tr} }(\mathcal{N}) = 2Q_{\leftrightarrow, \mathrm{P} }(\mathcal{N}) \geq 2Q_{\mathrm{LOCC}, \mathrm{P} }(\mathcal{N}) \geq \liminf_{R \to 0^+} 2 Q_{\mathrm{LOCC},\mathrm{P}} (\mathcal{N}, R) \, ,
	\end{equation}
	where the first inequality holds by a subset argument and the last inequality follows by the same argument as given on page 26 of \cite{Girardi_2025_2}.
	
	Now, let us further lower bound the right-hand side by fixing the input state to be a product state, i.e.\ \(  \phi_{\overline{A}^m A^m} = \phi_{\overline{A} A}^{\otimes m} \). Note that this choice is most likely not optimal as the Petz R\'enyi coherent information is known to be not additive in general. This then leads to the simplified lower bound
	\begin{equation}\label{Eq:Start}
		Q_{\leftrightarrow, \mathrm{Tr} }(\mathcal{N}) \geq \liminf_{R \to 0^+} \sup_{ s \in (0,1) } \left\{  s \left( \inf_{ \sigma_B } \overline{D}_{\frac{1}{1+s}}( \rho_{\overline{A}B} \| \mathbb{1}_{\overline{A}}\otimes \sigma_B) - R \right) \right\} \, ,
	\end{equation}
	where \( \rho_{\overline{A}B} = \mathcal{N}_{A \to B} (\phi_{\overline{A}A} ) \).
	
	In a first step, let us rewrite this as follows
	\begin{align}
		\sup_{\alpha \in \left( 1/2, 1\right) } \left\{ \frac{1-\alpha}{\alpha}\left[ \inf_{\sigma_B} \overline{D}_\alpha ( \rho_{\overline{A}B} \| \mathbb{1}_{\overline{A}} \otimes \sigma_B) - R \right] \right\} &= \sup_{\alpha \in \left(1/2, 1\right) } \left\{ \inf_{\sigma_B} -\frac{1}{\alpha} \log \tr{ \rho_{\overline{A}B}^\alpha \sigma_B^{1-\alpha} } - \frac{1-\alpha}{\alpha} R\right\} \nonumber \\
		&=\sup_{\beta \in \left(0,1/2 \right) } \left\{ \inf_{\sigma_B} \frac{1}{\beta-1} \log \tr{ \rho_{\overline{A}B}^{1-\beta} \sigma_B^{\beta} } - \frac{\beta}{1-\beta} R\right\}  \nonumber \\
		&= \sup_{\beta \in \left(0,1/2 \right) } \left\{ \inf_{\sigma_B} \overline{D}_\beta( \mathbb{1}_{\overline{A}} \otimes \sigma_B \| \rho_{\overline{A}B} ) - \frac{\beta}{1-\beta} R\right\} \nonumber \, ,
	\end{align}
	where we first parametrised \( \alpha = \frac{1}{1+s} \) and then re-parametrised to \(\beta = 1 - \alpha \) in a second step in order to obtain a \emph{reversed} version of the coherent information. 
	
	Now, we can carry out the limit \(R \to 0^+\). For this, observe that for fixed \(\beta\) the function
	\begin{equation}
		\inf_{\sigma_B} D_\beta( \mathbb{1}_{\overline{A}} \otimes \sigma_B \| \rho_{\overline{A}B} ) - \frac{\beta}{1-\beta} R
	\end{equation}
	is monotone decreasing in \(R\) since \( \frac{\beta}{1-\beta} > 0 \). This monotonicity is inherited by the supremum over \(\beta\). Consequently, we can compute the limit as a supremum over \(R > 0 \). After an interchange of the two suprema, we can carry out the limit \(R \to 0^+ \) within the supremum over \(s\) to obtain
	\begin{equation}
		\lim_{R \to 0^+}  \sup_{\beta \in \left(0,1/2 \right)  } \left\{ \inf_{\sigma_B} \overline{D}_\beta( \mathbb{1}_{\overline{A}} \otimes \sigma_B \| \rho_{\overline{A}B} ) - \frac{\beta}{1-\beta} R \right\} = \sup_{\beta \in \left(0,1/2 \right) } \inf_{\sigma_B} \overline{D}_\beta( \mathbb{1}_{\overline{A}}  \otimes \sigma_B \| \rho_{\overline{A}B} ) \, .
	\end{equation}
	Note that the optimiser \(\sigma_B\) can be determined explicitly for any \(\beta \in (0,1)\) (cf.\ \cite[Appendix A]{Girardi_2025_2}). We start by rewriting
	\begin{align}
		\inf_{\sigma_B} \overline{D}_\beta( \mathbb{1}_A \otimes \sigma_B \| \rho_{AB} ) = \frac{1}{\beta-1} \log \sup_{\sigma_B} \tr{  \sigma_B^{\beta} \trA{\rho_{AB}^{1-\beta}} } \, .
	\end{align}
	Then, setting \( X_B := \sigma_B^\beta \) and \(Y_B := \trA{ \rho_{AB}^{1-\beta} } \), we have
	\begin{align}
		X_B &\geq 0 & \norm{X_B}_{1/\beta} = \left( \tr{ X_B^{1/\beta} } \right)^\beta = 1 \, .
	\end{align}
	By Hölder's inequality, we can conclude that
	\begin{align}
		\frac{1}{\beta-1} \log \sup_{\sigma_B} \tr{  \sigma_B^{\beta} \trA{\rho_{AB}^{1-\beta}} } &= \frac{1}{\beta-1} \log \sup_{X \geq 0, \norm{X}_{1/\beta} = 1 } \tr{ X Y } \\
		&= \frac{1}{\beta-1} \log \norm{Y}_{\frac{1}{1-\beta}} = - \log \tr{ \left(\trA{ \rho_{AB}^{1-\beta} } \right)^{\frac{1}{1-\beta}} } \, .
	\end{align}
	Finally, as the Petz \(\alpha\)-divergence is monotonically increasing in the order parameter, we get the largest achievability bound for \(\beta \to  1/2\). Overall, we have found that for finite dimensional channels \(\mathcal{N}_{A \to B} \)  that
	\begin{equation}\label{Eq:End}
		Q_{\leftrightarrow, \mathrm{Tr} }(\mathcal{N}) \geq - \log \norm{ \mathrm{Tr}_{\overline{A}}\left[ \sqrt{\rho_{\overline{A}B}} \right] }_2^2 =- \log \tr{ \left( \mathrm{Tr}_{\overline{A}}\left[ \sqrt{ \rho_{\overline{A}B} } \right] \right)^{2} } \, ,
	\end{equation}
	where \( \rho_{\overline{A}B} = \mathcal{N}_{A \to B} (\phi_{\overline{A}A} ) \) for any fixed input state \(\phi_{\overline{A}A}\) .
	
	In a next step, we lift this bound to Gaussian channels. Recall that by definition, for any Gaussian input the output of the channel is Gaussian as well. We can then make use of the LOCC truncation map from \cite[Lemma 10]{Pirandola_2017}. This asserts the existence of an LOCC channel \(\mathbb{T}_d\) that maps any energy-constrained bosonic state, which includes Gaussian states,  into a truncated state \(\tilde{\rho}_{AB}\) defined over a \(d \times d\)-dimensional support such that 
	\begin{equation}
		\frac{1}{2} \norm{\rho_{AB} - \tilde{\rho}_{AB} }_1 \leq \sqrt{\gamma_d} + \gamma_d \quad \text{with} \quad \gamma_d := \frac{E}{\sqrt{d}-1} \, ,
	\end{equation}
	where \( \tr{\rho_{AB} \hat{H} } \leq E \) for the canonical Hamiltonian \( \hat{H} \) of the system. Observe that this implies trace-norm convergence of the output to the true output as \(d \to \infty\). 
	
	Now, we consider the following scheme: The encoder prepares n copies of the two-mode squeezed vacuum state \( \Phi_{\overline{A}A} (r)^{\otimes m} \), sends it via the channel \(\mathcal{N}_{A \to B}^{\otimes n} \) to the receiver. Sender and Receiver then apply the  truncation map \(\mathbb{T}_d\). Note that our achievability bound from above applies to the truncated state for any \(d \in \mathbb{N}\). Now, we can make use of the continuity of the bound in Eq.\ \eqref{Eq:End}, which follows from its characterisation in terms of the Schatten-2 norm. Taking the limit \(d\to\infty\), we then get the following achievability bound for Gaussian channels:
	\begin{equation}
		Q_{\leftrightarrow, \mathrm{Tr} }(\mathcal{N}) \geq - \log \tr{ \left( \mathrm{Tr}_{\overline{A}}\left[ \sqrt{ \rho_{\mathcal{N}}(r) } \right] \right)^2 } \, .
	\end{equation}
	Finally, taking the limit \(r \to \infty\) completes the proof.
\end{proof}

\begin{remark}

	A careful reading of the proof of \cite[Theorem 10]{Berta_2026} shows that one can interchange the role of Alice and Bob in their state-merging protocol to obtain a backward-assisted LOCC protocol for entanglement distillation. This can then in turn be used in the proof of \cite[Theorem 10]{Berta_2026} in order to obtain another lower bound on the error exponent of LOCC-assisted channel coding as
	\begin{equation}
		Q_{\mathrm{LOCC},\mathrm{P}} (\mathcal{N}, R) \geq \frac{1}{2} \sup_{s \in (0,1)} \left\{ s \left( \lim_{m \to \infty} \frac{1}{m}  I^{c, \mathrm{rev} }_{\frac{1}{1+s}} \left( \mathcal{N}_{A \to B}^{\otimes m} \right) - R \right) \right\} \, ,
	\end{equation}
	in terms of the Petz reverse coherent information \( I^{c, \mathrm{rev}}_{\alpha}(\mathcal{N}_{A \to B}) := \sup_{\phi_{\overline{A}A}} \inf_{\sigma_{\overline{A}} } \overline{D}_\alpha \left( \mathcal{N}(\phi_{\overline{A}A}) \| \sigma_{\overline{A}} \otimes \mathbb{1}_B \right) \). Re-running our proof above then gives another achievability bound for our setting as
	\begin{equation}
		Q_{\leftrightarrow, \mathrm{Tr} }(\mathcal{N}_{A \to B}) \geq  \limsup_{r \to \infty} \left\{ - \log \tr{ \left( \mathrm{Tr}_{B}\left[ \sqrt{ \rho_{\mathcal{N}}(r) } \right] \right)^2 } \right\} \, .
	\end{equation} 
\end{remark}

\subsubsection{One-Mode Gaussian Channels}

Recall that the quasi-Choi state of our one-mode channels is Gaussian with zero mean. Thus, our goal is to compute for such states, the following expression for Gaussian states
\begin{equation}
	\mathrm{LB}(\rho_{AB} ) := \left. - \log \tr{ \left(\trA{ \rho_{AB}^{1-\beta} } \right)^{\frac{1}{1-\beta}} } \right|_{\beta = 0.5}
\end{equation}
(as well as the version with the partial trace over the \(B\)-system).

From the Gibbs-type exponential form for Gaussian states, we can immediately see that 
\begin{equation}
	\rho^\alpha = \sigma_\alpha \cdot \tr{ \rho^\alpha }
\end{equation}
where \( \sigma_\alpha := \frac{\rho^{\alpha}}{\tr{\rho^\alpha}}\) is a normalised Gaussian state with Gibbs matrix \( \boldsymbol{G}[ \boldsymbol{V}_\alpha] := \alpha \boldsymbol{G}[\boldsymbol{V}] \). 

We can use this identity to determine \( \boldsymbol{V}_\alpha = \coth( \alpha \arcoth(\boldsymbol{V} i \boldsymbol{\Omega}) ) i \boldsymbol{\Omega} \). Considering the Williamson decomposition \( \boldsymbol{V} = \boldsymbol{S} \left(\bigoplus _i\nu_i \boldsymbol{\sigma}_0 \right) \boldsymbol{S} \), then \( \boldsymbol{V}_\alpha = \boldsymbol{S} \left(\bigoplus_i \nu_i(\alpha) \boldsymbol{\sigma}_0 \right) \boldsymbol{S} \) with symplectic eigenvalues \( \nu_{i}(\alpha) = \coth( \alpha \arcoth (\nu_i )) \). Consequently, if \(\boldsymbol{V}\) is given in two-mode normal form with parameters \((x,y,z)\), then the covariance matrix of the scaled state is also in normal form. The reduced state \( \trA{ \sigma_\alpha} \) thus has covariance matrix \( \boldsymbol{V}_B = y_\alpha \boldsymbol{\sigma}_0 \), where the coefficient can be explicitly computed with the tools from Sec.\ \ref{Sec:Normal_Form} as
\begin{equation}
	y_\alpha = \frac{1}{2} \left( \frac{x+y}{\sqrt{(x+y)^2-4z^2}} (\nu_{2}(\alpha) + \nu_{1}(\alpha) ) + (\nu_{2}(\alpha) - \nu_{1}(\alpha) ) \right) \, .
\end{equation}

\begin{remark}
	For our bound in terms of the partial trace over \(B\), we find that the reduced state has covariance matrix \( \boldsymbol{V}_{\overline{A}} = x_\alpha \boldsymbol{\sigma}_0 \) with 
	\begin{equation}
	x_\alpha = \frac{1}{2} \left( \frac{x+y}{\sqrt{(x+y)^2-4z^2}} (\nu_{2}(\alpha) + \nu_{1}(\alpha) ) + (\nu_{1}(\alpha) - \nu_{2}(\alpha) ) \right) \, .
\end{equation}
\end{remark}

Moreover, we can compute \( \tr{ \rho^\alpha } \) for \(m\)-mode Gaussian states via the symplectic eigenvalues of the corresponding covariance matrix using the following formula (cf.\ Serafini's textbook \cite{Serafini}):
\begin{equation}
	\tr{ \rho^\alpha } = \prod_{i=1}^m \frac{2^{\alpha}}{(\nu_i + 1)^\alpha - (\nu_i - 1)^{\alpha} } = \prod_{i=1}^m g_\alpha(\nu_i)
\end{equation}

Then, note that 
\begin{align}
	\mathrm{LB}(\rho_{AB} ) &= -\frac{1}{1-\beta} \log \tr{ \rho_{AB}^{1-\beta} } - \log \tr{ \trA{ \sigma_{1-\beta}}^{\frac{1}{1-\beta} }}  \\
		&=  -\frac{1}{1-\beta} \bigg( \log(g_{1-\beta}(\nu_1) ) + \log(g_{1-\beta}(\nu_2) ) \bigg) - \log g_{\frac{1}{1-\beta}}(y_\alpha)  \, ,
\end{align}
which is now completely specified in terms of the symplectic eigenvalues and the normal form parameters of the original covariance matrix.

\subsubsection{Thermal Attenuator}

For the thermal attenuator, we find that the bound involving the partial trace over the \(A\) part gives the best bound. From the analysis in Sec.\ \ref{Sec:Thermal_Attenuator}, we can directly give the relevant terms in leading order in the squeezing parameter as
\begin{align}
	\nu_1 &= n & \nu_2 &= (1- \lambda ) \cosh(2r) & y_\alpha &= \frac{\cosh(2r)}{\alpha} \, .
\end{align}
With this, we get the expression
\begin{equation}
	  \frac{1}{1-\beta} \log \left( (n +1)^{1-\beta} - (n -1)^{1-\beta} \right) - \frac{\beta}{1-\beta} \log \left( 1- \lambda \right) - \log(2) \, .
\end{equation}
Evaluating the latter expression at \(\beta=0.5\) gives us
\begin{align}
	 &2 \log \left( \sqrt{n+1}- \sqrt{n-1} \right) - \log(1-\lambda) - \log(2)  =  - \log\left( (1-\lambda) \left( n + \sqrt{n^2-1} \right) \right)
\end{align}
and this is positive for 
\begin{equation}
	1 \geq \lambda > 1 - \frac{1}{ n+ \sqrt{n^2-1}} \, .
\end{equation}

\subsubsection{Thermal Amplifier}

For the thermal amplifier, we find that the bound involving the partial trace over the \(B\) part gives the best bound. We get with our analysis from Sec.\ \ref{Sec:Thermal_Amplifier} the lower bound
\begin{align}
	- \log\left( \frac{\eta-1}{\eta} \left( n + \sqrt{n^2-1} \right)  \right) \, ,
\end{align}
which is positive for 
\begin{equation}
	1 \leq \eta \leq 1 + \frac{1}{ n -1 + \sqrt{n^2-1}} \, .
\end{equation}

\subsubsection{Additive Gaussian Noise}

Lastly, for the additive Gaussian noise channel, we get by the same machinery that both of our bounds evaluate to
\begin{equation}
	- \log( 2 \mu)
\end{equation}
which is positive for \( 0 \leq \mu \leq 0.5 \).

\end{document}